%% file: single_column_submission.tex
\documentclass[12pt,onecolumn]{IEEEtran}
\usepackage{amsmath,amsfonts}
\usepackage{algorithmic}
\usepackage{algorithm}
\usepackage{array}
\usepackage[caption=false,font=normalsize,labelfont=sf,textfont=sf]{subfig}
\usepackage{textcomp}
\usepackage{stfloats}
\usepackage{url}
\usepackage{verbatim}
\usepackage{graphicx}
\usepackage{cite}
\usepackage{ieee}

\usepackage[dvipsnames]{xcolor}
\usepackage[hidelinks]{hyperref}

\begin{document}

\title{Exponential Error Bounds for Information Bottleneck Source Coding Problems}


\author{Han Wu and Hamdi Joudeh
        \thanks{Han Wu was with Eindhoven University of Technology, The Netherlands. He is currently with Tokyo University of Agriculture and Technology, Japan. Email: hanwu@go.tuat.ac.jp.
        Hamdi Joudeh is with Eindhoven University of Technology, The Netherlands. Email: h.joudeh@tue.nl.
        This work was supported in part by the European Research Council (ERC) under Grant 101116550.}
}



\maketitle

\begin{abstract}
We study the information bottleneck (IB) source coding problem, also known as remote lossy source coding under logarithmic loss.
Based on a rate-limited description of noisy observations, the receiver produces a soft estimate for the remote source, i.e., a probability distribution, evaluated under the logarithmic loss.
We focus on the excess distortion probability of IB source coding and investigate how fast it converges to \(0\) or \(1\), depending on whether the rate is above or below the rate-distortion function.
The latter case is also known as the exponential strong converse.
We establish both the exact error exponent and the exact strong converse exponent for IB source coding by deriving matching upper and lower exponential bounds.
The obtained exponents involve optimizations over auxiliary random variables.
The matching converse bounds are derived through non-trivial extensions of existing sphere packing and single-letterization techniques, which we adapt to incorporate auxiliary random variables.

In the second part of this paper, we establish a code-level connection between IB source coding and source coding with a helper, also known as the Wyner-Ahlswede-Körner (WAK) problem.
We show that every code for the WAK problem is a code for IB source coding.
This requires noticing that IB source coding, under the excess distortion criterion, is equivalent to source coding with a helper available at \emph{both} the transmitter \emph{and} the receiver; the latter in turn relates to the WAK problem.
Through this connection, we re-derive the best known sphere packing exponent of the WAK problem, and provide it with an operational interpretation.
\end{abstract}
\section{Introduction}
We study the information bottleneck (IB) source coding problem, also known as remote lossy source coding under logarithmic loss, a two-hop network comprising a source, a helper, and a receiver (see Fig. \ref{fig:source_model_IB_source} further on).
The helper observes the source sequence \(X^n\) through a DMC \(P_{Y|X}\) and forwards a description of its observation \(Y^n\) to the receiver via a noiseless rate-limited link of capacity \(R\).
For every forwarded index from the helper, the receiver produces a ``soft'' estimate \(\hat{P}\) for the source sequence, i.e., a probability distribution, rather than reconstruct a ``hard'' estimate given by a single sequence.
The distortion between a source sequence \(x^n\) and its soft estimate \(\hat{P}\) is measured by the logarithmic loss
\begin{equation}
    \bar{d}(x^n, \hat{P}) \triangleq \frac{1}{n} \log \frac{1}{\hat{P}(x^n)}. \label{eq:intro:log-loss-definition}
\end{equation}
We allow the soft estimate to be any distribution supported on \(\mathcal{X}^n\), not necessarily a product distribution; thus the distortion measure \(\bar{d}(x^n, \hat{P})\) is not additive (or separable) in general.
Soft estimates often arise in inference tasks, and also deserve attention in compression, given their connections \cite{rissanenUniversalCodingInformation1984,merhavUniversalPrediction1998}.

The rate-distortion function of this setting under distortion level \(\Delta\), i.e., the minimum achievable rate under which the average logarithmic loss distortion does not exceed \(\Delta\), is given by
\begin{equation}
    R(\Delta) = \min_{P_{U|Y}} I(Y;U) \qquad \text{s.t.} \quad H(X|U) \leq \Delta, \ X \to Y \to U.
\end{equation}
See, e.g., \cite[Section III. F]{courtadeMultiterminalSourceCoding2014}. This rate-distortion function closely relates to the IB method \cite{tishbyInformationBottleneckMethod1999}, which is why the setting is also known as the IB source coding problem.
We henceforth use the two terms interchangeably and refer to this setting simply as the IB source.
Besides the average distortion, another key criterion for lossy source coding is the excess distortion probability---the probability of source sequences that are reconstructed with distortion exceeding \(\Delta\) (i.e., all \(x^n\) such that \(\bar{d}(x^n, \hat{P}) > \Delta\)).
Both distortion criteria share the same rate-distortion function in the IB source setting, but the latter is better suited for the study of exponential error bounds, and hence is the main focus of this paper.

The connection between the IB method and lossy source coding can be traced back to Gilad-Bachrach \emph{et al.} \cite{gilad-bachrachInformationTheoreticTradeoff2003} and Harremoës and Tishby \cite{harremoesInformationBottleneckRevisited2007}.
Later on, Courtate and Weismann crystallized this connection by formalizing it in terms of soft reconstructions and the logarithmic loss, and further extended it to multiterminal settings \cite{courtadeMultiterminalSourceCoding2014}.
More recently, Shkel and Verdú \cite{shkelSingleshotApproachLossy2018} derived single-shot bounds under both the excess and average logarithmic loss distortion criteria in the non-remote setting, where the encoder directly observes the source.
In doing so, they established a fundamental connection between logarithmic loss under the excess-distortion criterion and list decoding.
This connection is particularly useful to our current work, as we shall see further on.
Additional related literature is reviewed later in this section.
\subsection{Exponential Error Bounds}
Despite extensive studies of lossy source coding under logarithmic loss, exponential error bounds for its excess-distortion probability have received little attention, apart from our preliminary work \cite{joudehErrorExponentsSource2024} where we study the non-remote setting.
To address this gap, we investigate two regimes for the excess distortion probability of the IB source, depending on whether the rate is above or below \(R(\Delta)\), i.e., how fast the excess distortion probability converges to \(0\) or \(1\) for a given rate.
The latter case is also known as the exponential strong converse.
We will derive matching upper and lower exponential bounds for both regimes; thus establish the exact error exponent and strong converse exponent for the IB source.

We note that the error and strong converse exponents for standard remote lossy source coding, i.e., under “hard” reconstructions and additive distortion measures, have been established in \cite{weissmanTradeoffsExcesscodelengthExponent2002} and \cite{wuStrongConverseExponentRemote2025}, respectively.
Nevertheless, the distortion measure in the IB source problem is not additive; therefore, these previous results are not applicable here.
Moreover, unlike the standard case, the exponents in the logarithmic loss setting involve optimizations over auxiliary random variables, which we denote by $U$, making the problem distinct and more challenging, particularly in the converse parts of the proofs.
\subsection{Connections to Source Coding with a Helper}
For reasons that will become clear shortly, we now turn our attention to source coding with a helper, also known as the Wyner-Ahlswede-Körner (WAK) problem.
In this setting, a helper observes a source sequence \(X^n\) through a DMC \(P_{Y|X}\) and forwards a description of its observation \(Y^n\) to a receiver via a noiseless link of capacity \(B\).
Meanwhile, a transmitter describes the source sequence \(X^n\) independently to the receiver via another noiseless link.
From the two descriptions, the receiver attempts to reconstruct the source sequence, whose output will be written as \(\hat{X}^n\).

Let \(R_{\text{h}}(B)\) denote the minimum rate for the transmitter's description such that the reconstruction error probability \(\P\{X^n \neq \hat{X}^n\}\) vanishes.
Wyner \cite{wynerSourceCodingSide1975} and Ahlswede and Körner \cite{ahlswedeSourceCodingSide1975} showed that
\begin{equation}
    R_{\text{h}}(B) = \min_{P_{U|Y}} H(X|U) \qquad \text{s.t.}  \quad  I(Y; U) \leq B, \ X \to Y \to U. \label{eq:characterization_minimum_rate_coded_side_information}
\end{equation}
This minimum rate closely relates to the IB method \cite{tishbyInformationBottleneckMethod1999}, which is why the WAK problem has also been recognized as an IB problem \cite{gilad-bachrachInformationTheoreticTradeoff2003,zaidiInformationBottleneckProblems2020}.
Following the above observation, it is intriguing to ask whether there is a deeper level of connection between the IB source and the WAK problem, beyond their common rate limits.
For example, can coding schemes developed for one problem be applied to the other?

In this paper, we uncover such a code-level connection by showing that every code for the WAK problem is a code for the IB source.
We accomplish this by establishing an equivalence between the IB source (specifically under the excess-distortion criterion) and a more general version of the WAK problem where the helper is connected to \emph{both} the transmitter \emph{and} the receiver.
Through this connection, we re-derive the best known sphere packing exponent of the WAK problem, previously established by Kang and Liu in \cite{kangUpperBoundError2018}, and provide the exponent with an operational interpretation.
\subsection{Contributions and Organization}
We now summarize the main technical contributions of this paper. First, we establish the exact error exponent of the IB source in Theorem \ref{thm:exponent_remote_log_loss} by deriving matching achievability and converse bounds.
For the achievability, we use the type covering lemma \cite{bergerRateDistortionTheory1971,csiszarInformationTheoryCoding2011} and the close connection between the logarithmic loss and list decoding \cite{shkelSingleshotApproachLossy2018} to design the coding scheme.
The error exponent is then derived by investigating the intersection between conditional type classes.
The converse, on the other hand, requires a non-trivial extension of the conventional sphere packing argument, as we elaborate below.
We also show that the error exponent is strictly positive if and only if the rate is above the rate-distortion function; and furthermore, we recover a previous result from \cite{joudehErrorExponentsSource2024} for lossy source coding under logarithmic loss.

A key feature of the error exponent revealed by the achievability proof is that the excess-distortion probability is dominated by a joint type \( Q_{XYU} \) that may not satisfy the Markov chain \( X \to Y \to U \). This is in sharp contrast to the rate-distortion function \( R(\Delta) \), where the Markov chain is essential.
As a result, a conditional mutual information term \( I_Q(X;U|Y) \) appears in the exponent, which has been dubbed the \emph{soft} Markov chain, or Markov penalty, in previous works studying related settings \cite{oohamaExponentialStrongConverse2018,tyagiStrongConverseUsing2020,takeuchiTightExponentialStrong2025}.
Hence, a key challenge in the converse proof is deriving this term.
A direct extension of the change of measure argument---used in the conventional sphere packing approach \cite{haroutunianBoundsExponentProbability1968,martonErrorExponentSource1974}---to multiterminal problems leads to bounds involving optimizations over Markov chains; see, e.g., \cite[Problem 13.7]{csiszarInformationTheoryCoding2011}.
To overcome this limitation,
we carefully select a dummy (or test) source that is not a DMS such that it does not adhere to the Markov chain constraint \( X \to Y \to U \).
Combining this with the technique of marginalizing, we introduce the soft Markov chain term \( I_Q(X;U|Y) \) into the final bound and obtain a tight converse.

Second, we establish the exact strong converse exponent of the IB source in Theorem \ref{thm:strong_converse_exponent_remote_log_loss} by deriving matching achievability and converse bounds.
The achievability scheme is similar to the one in Theorem \ref{thm:exponent_remote_log_loss}, with the main difference being that here we do not restrict the output of covering to satisfy the rate-limit \(R\).
Instead, we handle the rate limit by choosing the best \(e^{nR}\) codewords from the covering set, which enables us to find a lower bound for the probability of interest.
For the converse, we draw inspiration from the converse proof of \cite[Theorem 15.9]{polyanskiyInformationTheoryCoding2025} in the context of binary hypothesis testing.
We generalize and extend this change of measure argument to multiterminal coding problems by introducing an auxiliary random variable \(U\) and the soft Markov chain \(I_Q(X;U|Y)\), and deriving the corresponding single-letter lower bound.
We also show that the exponent is positive if and only if the rate is below the rate-distortion function, and study the special case of lossy source coding under logarithmic loss.

Finally, we establish a code-level connection between the IB source and the WAK problem in Theorem \ref{thm:chp_sphere_packing_exponent_WAK/source_coding_helper_both_sides_equivalent_IB_source} and Corollary \ref{thm:chp_sphere_packing_exponent_WAK/WAK_connections_IB_source}.
In particular, we show that the IB source under the excess distortion criterion is equivalent to source coding with a helper available at both the transmitter and the receiver: each code for one problem can be transformed into a code for the other.
This is done by noticing that the latter problem can be viewed as conditional source coding based on the helper's common description, which then relates to the IB source through the connection between the logarithmic loss and list decoding.
Thus, the two problems share the same error probability and error exponent.
From this, we see that codes for the WAK problem form a subset of codes for the IB source, and the error exponent of the latter becomes an upper bound for all achievable error exponents of the former.
While this upper bound was previously established by Kang and Liu in \cite{kangUpperBoundError2018}, our proof reveals the code-level connection between the IB source and the WAK problem and provides this upper bound with an operational interpretation.

The rest of the paper is organized as follows.
We provide a literature review on the IB source and the WAK problem in the next subsection, followed by describing the notation used in this paper.
In Section \ref{sec:problem_setup}, we describe the IB source problem in further detail and formally present the main subjects of interest in this paper.
All the main results of this paper are presented and discussed in Section \ref{sec:main_results}.
We then discuss some useful preliminary results in Section \ref{sec:preliminaries}, before proceeding to the proofs of our main results from Section \ref{sec:chp_remote_log_loss/achievability} to Section \ref{sec:chp_sphere_packing_exponent_WAK/source_coding_helper_both_sides_equivalent_IB_source}. Proofs of some technical lemmas are deferred to the appendices.
Lastly, we provide concluding remarks and future directions in Section \ref{sec:conclusions}.

\subsection{Literature Review}
\label{sec:chp_intro/literature_review}
Remote lossy source coding was first studied in \cite{dobrushinInformationTransmissionAdditional1962}, with subsequent classical works including, e.g., \cite{wolfTransmissionNoisyInformation1970, bergerRateDistortionTheory1971, witsenhausenIndirectRateDistortion1980}.
Some variants of this problem have been considered in, e.g., multiterminal \cite{yamamotoSourceCodingTheory1980}, the CEO problem \cite{bergerCEOProblem1996,oohamaRatedistortionFunctionQuadratic1998}, Gaussian \cite{oohamaIndirectDirectGaussian2014,eswaranRemoteSourceCoding2019}, nonasymptotic regime \cite{kostinaNonasymptoticNoisyLossy2016,liuNonasymptoticObliviousRelaying2025}, and \(f\)-separable distortion measures \cite{stavrouIndirectRateDistortion2023}.
The connection between remote lossy source coding under logarithmic loss and the IB method is uncovered in \cite{gilad-bachrachInformationTheoreticTradeoff2003,harremoesInformationBottleneckRevisited2007,courtadeMultiterminalSourceCoding2014}.
Further studies on lossy source coding under logarithmic loss consider, e.g., multiterminal settings \cite{courtadeMultiterminalSourceCoding2014,ugurVectorGaussianCEO2020}, single-shot and universal coding \cite{shkelSingleshotApproachLossy2018,shkelUniversalLossyCompression2017,shkelUniversalCompressionList2018}, generalized length functions \cite{wuSoftGuessingLogarithmic2023}, and stationary sources \cite{ulgerOnSourceCoding2025}.
The error exponent and strong converse exponent for lossy source coding are established in \cite{martonErrorExponentSource1974} and \cite[Problem 9.6]{csiszarInformationTheoryCoding2011} respectively, and for remote lossy source coding, the error exponent and strong converse exponent are characterized in  \cite{weissmanTradeoffsExcesscodelengthExponent2002} and \cite{wuStrongConverseExponentRemote2025} respectively.
Error exponents for lossy source coding under logarithmic loss (with side information) are studied in \cite{joudehErrorExponentsSource2024}.
The exponential strong converse for channel coding problems has attracted interests since \cite{arimotoConverseCodingTheorem1973,dueckReliabilityFunctionDiscrete1979}, and for multiterminal source coding problems it has been explored in, e.g., \cite{oohamaUniversalCodingSlepianWolf1994,oohamaExponentialStrongConverse2018,oohamaExponentialStrongConverse2019,takeuchiTightExponentialStrong2025,watanabeTightExponentialStrong2025}.

The rate limit of the WAK problem is established in \cite{wynerSourceCodingSide1975,ahlswedeSourceCodingSide1975}.
The nonasymptotic regime is studied in, e.g., \cite{watanabeNonasymptoticSecondorderAchievability2015,liuDispersionBoundWynerAhlswedeKorner2021}.
A strong converse for the WAK problem via the Gray-Wyner network \cite{graySourceCodingSimple1974} is provided in \cite{watanabeConverseBoundWynerAhlswedeKorner2017}.
The best known achievable error exponent for the WAK problem first appeared in \cite{kellyReliabilitySourceCoding2012}, where an upper bound for all achievable error exponents is also provided.
This achievable error exponent was recently re-drived in \cite{wuErrorExponentsObliviousRelaying2025} by establishing a connection between the WAK problem and oblivious relaying \cite{sanderovichCommunicationDecentralizedProcessing2008}, also known as the IB channel.
The best known upper bound for all achievable error exponents was first derived in \cite{kangUpperBoundError2018}.
The exponential strong converse for the WAK problem is studied in \cite{oohamaExponentialStrongConverse2019,takeuchiTightExponentialStrong2025}.

\subsection{Notation}
\label{sec:chp_intro/notation}
We describe the notation that will be used throughout the paper.
Given a finite alphabet \(\mathcal{X}\), we use \(\mathcal{P}(\mathcal{X})\) to denote the set of all probability mass functions (pmfs) \(P_X\) on \(\mathcal{X}\).
We write \(\bm{x}=(x_1, x_2, \ldots, x_n)\) for an \(n\)-length sequence from \(\mathcal{X}^n\).
A random vector on \(\mathcal{X}^n\)  is denoted by \(\bm{X}=(X_1, X_2, \ldots, X_n)\).
Depending on the context, we may also write \(x^n\) and \(X^n\) instead of \(\bm{x}\) and \(\bm{X}\).
In the same manner, we adopt the notation \(\bm{y} = (y_1, y_2, \ldots, y_n)\) or \(\bm{u}=(u_1, u_2, \ldots, u_n)\), and \(\bm{Y}\) or \(\bm{U}\), on \(\mathcal{Y}^n\) or \(\mathcal{U}^n\) respectively.
All alphabets in this paper are finite.
Following convention, the hat symbol \(\hat{P}\) is used whenever we are looking at the empirical distribution induced by some deterministic sequences.
For example, for a sequence \(\bm{x} \in \mathcal{X}^{n}\), we use \(\hat{P}_{\bm{x} }\) to denote its vector of relative frequencies of all symbols \(x \in \mathcal{X}\), i.e., its type.
\(\hat{P}_{\bm{x}\bm{y}}\) denotes the joint type of a sequence pair \((\bm{x}, \bm{y})\), while \(\hat{P}_{\bm{y} | \bm{x}}\) is the conditional type from \(\bm{x}\) to \(\bm{y}\) induced by  \(\hat{P}_{\bm{x}\bm{y}}\).
The set of all possible types \(\hat{P}_{\bm{x}}\) on \(\mathcal{X}^{n}\) is written as \(\mathcal{P}_n(\mathcal{X})\), while the set of all possible conditional types \(\hat{P}_{\bm{y} | \bm{x}}\) for sequences from \(\mathcal{X}^n\) and \(\mathcal{Y}^n\) is written as \(\mathcal{P}_n(\mathcal{Y} | \mathcal{X})\).
The type class \(\mathcal{T}_n(P_X)\) consists of all sequences \(\bm{x}\) that have the same type \(P_X \in \mathcal{P}_n(\mathcal{X})\).
For a given sequence \(\bm{x}\), the conditional type class \(\mathcal{T}_n(P_{Y|X} | \bm{x})\) is the set of all sequences \(\bm{y}\) such that the conditional type from \(\bm{x}\) to \(\bm{y}\) is \(P_{Y|X} \in \mathcal{P}_n(\mathcal{Y} | \mathcal{X})\).
In a slight abuse of notation, we also use \(\hat{P}\) for the soft estimates at the receiver.
This, however, should create no confusion, as types are indexed by sequences, e.g., \(\hat{P}_{\bm{x} }\), while soft estimates are not.

The entropy of \(P_{X}\) is written as \(H(X)\) or \(H(P_X)\) and the conditional entropy between two random variables \(X\) and \(Y\) is denoted by \(H(Y|X)\) or \(H(P_{Y|X} | P_X)\), while the mutual information between \(X\) and \(Y\) is written as \(I(X;Y)\) or \(I(P_X, P_{Y|X})\).
\(D(Q_X\|P_X)\) is the KL-divergence between two pmfs \(Q_X\) and \(P_X\), and \(D(Q_{Y|X}\|P_{X|Y} | P_X)\) denotes the conditional KL-divergence.
Given an event \(\mathcal{A}\), we use \(P[\mathcal{A}]\) to denote the probability of \(\mathcal{A}\) under the probability measure \(P\), while \(\idc{\mathcal{A}}\) is the indicator function of \(\mathcal{A}\) and \(\abs{\mathcal{A}}\) is its cardinality or size.
Given two sets \(\mathcal{A}\) and \(\mathcal{B}\), we use \(\mathcal{A} - \mathcal{B}\) to denote the elements from \(\mathcal{A}\) but not in \(\mathcal{A} \cap \mathcal{B}\).
For a conditional distribution \(P_{Y|X}\) with \(X \overset{P_{Y|X}}{\to} Y\), we use \(P_X \cdot P_{Y|X}\) to denote the distribution of \(Y\) when the input distribution is \(P_X\).
For a Markov chain \(X \overset{P_{Y|X}}{\to} Y \overset{P_{U|Y}}{\to} U\), we use \(P_{Y|X} \cdot P_{U|Y}\) to denote the conditional distribution between \(X\) and \(U\) through the Markov chain.
Let \(a_n \ndot{\leq} b_n\) if \(\limsup_{n \to \infty}\frac{1}{n}\log (a_n/b_n) \leq 0\) and \(a_n \ndot{\geq} b_n\) if \(\liminf_{n \to \infty}\frac{1}{n} \log (a_n/b_n) \geq 0\).
We will write \(a_n \ndot{=} b_n\) if both hold.
For a positive integer constant \(N\), we use \([N]\) to denote \(\{1,2,\ldots,N\}\).
Let \(|a|^{+} \triangleq \max\{0,a\}\).
The base of exponential and log functions is chosen as the natural base.


\section{Problem setup}
\label{sec:problem_setup}
Consider a pmf \(P_{XY}\).
Assume that we have a DMS pair \((X^n, Y^n)\) following the distribution
\begin{equation}
    P_{X^n Y^n} ( x^n, y^n) = \prod_{i=1}^{n}P_{XY}(x_i, y_i), \qquad \forall (x^n, y^n) \in \mathcal{X}^n \times \mathcal{Y}^n.
\end{equation}
We can think of \(X^n\) as a remote source and \(Y^n\) as a noisy observation through the DMC \(P_{Y|X}\).
As shown in Fig. \ref{fig:source_model_IB_source}, a helper has access to the noisy observation \(Y^n\) and describes it to a receiver through a rate-limited link \(f_n:\mathcal{Y}^n \to [e^{nR}]\).
Given a forwarded index \(m \in [e^{nR}]\), the receiver produces a ``soft'' estimate of the source sequence through a decoder \(g_n: [e^{nR}] \to \mathcal{P}(\mathcal{X}^n) \).
That is, for every forwarded index \(m \in [e^{nR}]\), the receiver produces a pmf as an estimate, denoted by \(\hat{P}_m\), rather than reconstruct a deterministic sequence.
Note that here the soft estimate \(\hat{P}_m\) can be any distribution supported on \(\mathcal{X}^n\); hence is not limited to the product (i.e., memoryless) distributions \(\mathcal{P}(\mathcal{X})^n\).

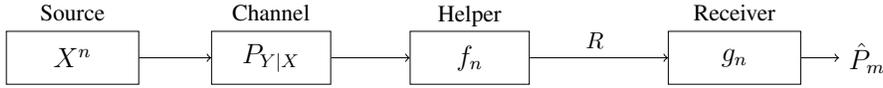
\begin{figure}[hbt!]
    \centering
    \scalebox{0.88}{\input{figures/IB_source_model.tex}}
    \captionsetup{justification=centering}
    \caption{Information Bottleneck Source}
    \label{fig:source_model_IB_source}
\end{figure}

The distortion between a source sequence \(x^n\) and its soft reconstruction \(\hat{P}\) is measured by the logarithmic loss defined in \eqref{eq:intro:log-loss-definition}, which we restate here for convenience:
\begin{equation}
    \bar{d}(x^n, \hat{P}) \triangleq - \frac{1}{n} \log \hat{P}(x^n).
\end{equation}
Since \(\hat{P}(x^n) \leq 1\) for every \(\hat{P}\), we have \(  \bar{d}(x^n, \hat{P}) \geq 0\).
It follows that \(\bar{d}(x^n, \hat{P}) = 0\) if and only if \(\hat{P}(x^n) = 1\), i.e., the receiver produces a correct ``hard'' reconstruction.
As we are interested in exponential error bounds, it is natural to adopt an "excess distortion" formulation, see, e.g., \cite{martonErrorExponentSource1974}. In such a formulation, an excess distortion event (i.e., reconstruction "error" event) occurs if  the logarithmic loss between a source sequence and its soft reconstruction exceeds a pre-specified distortion level $\Delta$.
To this end, we define
\begin{equation}
    p_{\mathrm{e}} (n, R, \Delta) \triangleq \min_{f_n, g_n} \P \{ \bar{d}(X^n, g_n(f_n(Y^n)) ) > \Delta \},
\end{equation}
i.e., the minimum excess distortion probability for a given rate-distortion pair \((R, \Delta)\).
Note that throughout this paper we assume \(\Delta \geq \Delta_{\min}\), where
\begin{equation}
    \Delta_{\min} \triangleq \min_{\varphi_n: \mathcal{Y}^n \to \mathcal{P}(\mathcal{X}^n)} \E [\bar{d}(X^n, \varphi_n(Y^n))] = H(X|Y).
\end{equation}
Next, define the rate-distortion function of the IB source as
\begin{equation}
    R(\Delta) = \min_{P_{U|Y}: H(X|U) \leq \Delta} I(Y;U), \label{eq:chp_remote_log_loss/definition_rate_distoriton_function}
\end{equation}
where we have the joint distribution \(P_{XYU} = P_{X|Y}P_{Y}P_{U|Y}\) and \(|\mathcal{U}| \leq |\mathcal{Y}| + 1\).
In this paper, we are mainly interested in the error and strong converse exponents, formally defined as follows.

\subsection{Error and Strong Converse Exponents}
We first investigate how fast \( p_{\mathrm{e}} (n, R, \Delta) \) decays to \(0\) as \(n\) grows, captured by
\begin{equation}
    E(R, \Delta) \triangleq \lim_{n \to \infty} -\frac{1}{n} \log p_{\mathrm{e}} (n, R, \Delta),
\end{equation}
which is known as the error exponent, or the reliability function.
We will provide a complete characterization for \(E(R, \Delta)\) in this work, and also show that \(E(R, \Delta) > 0\) if and only if \(R > R(\Delta)\).

We next study the behavior of \(p_{\mathrm{e}} (n, R, \Delta)\) in the regime \(R < R(\Delta)\).
From the weak converse, we know that \(p_{\mathrm{e}} (n, R, \Delta)\) is bounded away from zero if \(R < R(\Delta)\).
Here we are interested in an exponential strong converse, i.e., showing that \(p_{\mathrm{e}}(n, R, \Delta)\) converges to \(1\) exponentially fast when \(R < R(\Delta)\), and characterizing the corresponding exponent.

To this end, we now consider the non-excess distortion probability (i.e., "correct" reconstruction probability), defined as
\begin{equation}
    p_{\mathrm{c}} (n, R, \Delta) \triangleq \max_{f_n, g_n} \P \{ \bar{d}(X^n, g_n(f_n(Y^n)) ) \leq \Delta \}.
\end{equation}
Since \(p_{\mathrm{c}}(n, R, \Delta) = 1 - p_{\mathrm{e}}(n, R, \Delta)\), it is equivalent to investigating how fast \(p_{\mathrm{c}} (n, R, \Delta)\) converges to \(0\) as \(n\) grows.
Let
\begin{equation}
    F(R, \Delta) \triangleq \lim_{n \to \infty} - \frac{1}{n} \log p_{\mathrm{c}}(n, R, \Delta),
\end{equation}
which is known as the strong converse exponent.
We will provide a complete characterization of \(F(R, \Delta)\) in this work,  and also show that \(F(R, \Delta) > 0\) if and only if \(R < R(\Delta)\).
\subsection{Differences from Previous Studies}
Before presenting the main results, we clarify the differences between this work and previous studies on the error and strong converse exponents of remote lossy source coding (RLSC).
We first illustrate that the IB source and standard RLSC differ in their problem settings, and then discuss how this difference affects the information-theoretic results (e.g., rate-distortion functions and error exponents) and their proofs.

In standard RLSC, for every forwarded index \(m \in [e^{nR}]\), the receiver reconstructs a deterministic sequence from \(\hat{\mathcal{X}}^n\), where \(\hat{\mathcal{X}}\) is a discrete reconstruction alphabet.
The distortion between a source sequence \(x^n\) and its reconstruction \(\hat{x}^n\) is also assumed to be additive, i.e.,
\begin{equation}
    \bar{d}(x^n, \hat{x}^n) = \frac{1}{n} \sum_{i=1}^{n}d(x_i, \hat{x}_i),
\end{equation}
where \(d: \mathcal{X} \times \hat{\mathcal{X}} \to \mathbb{R} \) is a distortion measure.
The IB source does not fall under the umbrella of this standard setting, because here the reconstruction alphabet \(\mathcal{P}(\mathcal{X}^n)\) is neither a product alphabet nor discrete (in fact, it is a simplex).
Further, the distortion measure \(\bar{d}(x^n, \hat{P})\) is not additive either.
For \(\bar{d}(x^n, \hat{P})\) to be additive, the soft reconstruction \(\hat{P}\) must be a product distribution, which need not be the case in general.
The difference in problem settings leads to different results, featured by an auxiliary random variable appearing in those of the IB source.
The error and strong converse exponents derived in this work also involve an auxiliary random variable, which is absent from the counterpart exponents of RLSC \cite{weissmanTradeoffsExcesscodelengthExponent2002,wuStrongConverseExponentRemote2025}.

We next discuss whether it suffices to consider product distributions in establishing achievability for the IB source.
As far as the rate-distortion function is concerned, Courtade and Weissman \cite[Section III-C]{courtadeMultiterminalSourceCoding2014} noticed that there is no loss in achievable rates when restricting reconstructions to product distributions.
However, the same phenomenon is not immediately observable in the case of error and strong converse exponents.
As we shall see, in the achievability proofs of this paper, the soft reconstructions employed at the receiver are not product distributions; and it is not immediately clear whether the same exponents are achievable using product distributions. This may be worth exploring as future work.

\section{Main Results and Discussions}
\label{sec:main_results}

\subsection{Error Exponent}
We first provide a complete characterization of \(E(R, \Delta)\).
\begin{theorem}
    \label{thm:exponent_remote_log_loss}
    For a DMS pair \(P_{XY}^n\), we have
    \begin{equation}
        E(R, \Delta) = \min_{Q_Y} \! \! \! \max_{\substack{  Q_{U|Y}:  \\ I(Q_Y, Q_{U|Y}) \leq R }   }  \min_{ \substack{  Q_{X|YU} : \\ H_Q(X|U) \geq \Delta }   } D(Q_{XYU} \| P_{XY}Q_{U|Y}), \label{eq:main_results_exponent_remote_log_loss}
    \end{equation}
    where we have the joint distribution \(Q_{XYU} = Q_Y Q_{U|Y} Q_{X|YU}\) and \(|\mathcal{U}| \leq 2|\mathcal{X}| |\mathcal{Y}| + 2\).
\end{theorem}
\begin{proof}
    See Section \ref{sec:chp_remote_log_loss/achievability} and Section \ref{sec:chp_remote_log_loss/converse}.
    Note that we can also write
    \begin{align}
        D(Q_{XYU} \| P_{XY}Q_{U|Y}) & = D(Q_{XY} \| P_{XY}) + D(Q_{XYU} \| Q_{XY}Q_{U|Y}) \\
        & = D(Q_{XY} \| P_{XY}) + I_Q(X;U|Y),
    \end{align}
    which is the form we arrive at in the achievability part.
\end{proof}

We next look into the rate regime for which \(E(R, \Delta)> 0\).
\begin{proposition}
    \label{prop:remote_log_loss_exponent_rate_distortion_function}
    For a fixed \(\Delta\), we have \(E(R, \Delta) > 0\) if and only if \(R > R(\Delta)\).
\end{proposition}
\begin{proof}
    See Appendix \ref{sec:chp_remote_log_loss/exponent_rate_distortion_function}.
\end{proof}

If \(P_{XY}\) satisfies \(\P\{X = Y\} = 1\), i.e., \(Y^n\) is a noiseless observation of \(X^n\), then the IB source is reduced to lossy source coding under logarithmic loss. Hence, we can recover \cite[Theorem 1]{joudehErrorExponentsSource2024} from Theorem \ref{thm:exponent_remote_log_loss}.
\begin{corollary}
    \label{cor:exponent_lossy_log_loss}
    If \(P_{XY}\) satisfies \(\P\{X = Y\} = 1\), then we have
    \begin{equation}
        E(R, \Delta) = \min_{Q_{X}:H_Q(X) \geq R + \Delta } D(Q_{X} \| P_X).
    \end{equation}
\end{corollary}
\begin{proof}
    See Appendix \ref{sec:chp_remote_log_loss/lossy_speical_case}.
\end{proof}

We now provide insights into the proof of Theorem \ref{thm:exponent_remote_log_loss}.
For the achievability part, we employ the type covering lemma to design the helper's compress-forward scheme: for every type class \(\mathcal{T}_n(Q_Y)\), we select a conditional type \(Q_{U|Y}\) such that \(I(Q_Y, Q_{U|Y}) \leq R\) to satisfy the rate-limit; then use the type covering lemma to construct the compress-forward codebook.
The receiver's soft estimates are constructed based on the connection between the logarithmic loss and list decoding, established by Shkel and Verdú in \cite{shkelSingleshotApproachLossy2018}.
Due to the importance of this connection to our proofs, we will review it in detail shortly.
For now, it suffices to note that the receiver's soft estimate under the logarithmic loss and distortion level \(\Delta\) is equivalent to a list with size \(e^{n\Delta}\).
As a result, we construct the list by considering all conditional type classes \(\mathcal{T}_n(Q_{X|U} | \bm{u}_m)\) satisfying \(|\mathcal{T}_n(Q_{X|U} | \bm{u}_m)| \leq e^{n\Delta}\) for a forwarded index \(m\) and its codeword \(\bm{u}_m\).
From this, we can establish the error exponent of the coding scheme by investigating the intersection of conditional type classes.

As for the converse part, a key difficulty is to derive the term \(I_Q(X;U|Y)\) in the exponent.
Note that this term is sometimes referred to as the \emph{soft} Markov chain or Markov penalty \cite{oohamaExponentialStrongConverse2018,tyagiStrongConverseUsing2020,takeuchiTightExponentialStrong2025}, for reasons that will become clear further on.
To fully appreciate the key difficulty in the proof, recall that the conventional converse proof for error exponents, commonly known as the sphere packing bound, dates back to Haroutunian's proof for the sphere packing exponent of DMCs \cite{haroutunianBoundsExponentProbability1968} and Marton's proof for the error exponent of lossy source coding \cite{martonErrorExponentSource1974}.\footnote{Here we focus on the approach to establishing the primal form of the sphere packing exponent. The dual form (i.e., Gallager's style) of the sphere packing exponent of DMCs was derived in \cite{shannonLowerBoundsError1967a} using a different method. The primal form of the sphere packing exponent of DMCs also appeared in \cite{blahutHypothesisTestingInformation1974}.}
Such proof relies on a change of measure argument, i.e., given any DMS \(P_{XY}^n\) and \(Q_{XY}^n\), if a pair \((\bm{x}, \bm{y})\) satisfies
\begin{equation}
    \frac{1}{n}\log \frac{Q_{XY}^n(\bm{x}, \bm{y})}{P_{XY}^n(\bm{x}, \bm{y})} \leq D(Q_{XY} \| P_{XY}) + \epsilon,
\end{equation}
where \(\epsilon>0\) is a small constant, then for such \((\bm{x}, \bm{y})\) we must have
\begin{equation}
    P_{XY}^n(\bm{x}, \bm{y}) \geq Q_{XY}^n(\bm{x}, \bm{y}) \times e^{-n(D(Q_{XY} \| P_{XY}) + \epsilon)}.
\end{equation}
From this, we can relate the decoding error probability under \(P_{XY}^n\) to that under a dummy (or test) source \(Q_{XY}^n\), which in turn can be evaluated via the weak converse for \(Q_{XY}^n\).

However, this conventional approach falls short when applied to multiterminal problems, i.e., the obtained sphere packing exponent is not sufficiently tight.
For example, a standard application of the conventional approach to the IB source only yields
\begin{equation}
    E(R, \Delta) \leq \min_{Q_{XY}: R_Q(\Delta) \geq R} D(Q_{XY} \| P_{XY}),
\end{equation}
where \(R_{Q}(\Delta)\) denotes the rate-distortion function of \(Q_{XY}^n\) defined in \eqref{eq:chp_remote_log_loss/definition_rate_distoriton_function}.
Later on, when studying source coding with a helper, i.e., the WAK problem, Kelly and Wagner \cite{kellyReliabilitySourceCoding2012} refined the conventional approach by incorporating Fano's lower bound for the decoding error probability of \(Q_{XY}^n\) into the analysis.
They derived an improved sphere packing exponent for the WAK problem, where the improvement can be seen from a sandwiched maximization step in \cite[Theorem 2]{kellyReliabilitySourceCoding2012}.
Following Kelly and Wagner's refined approach, we can show that
\begin{equation}
    E(R, \Delta) \leq \min_{Q_Y} \! \! \! \max_{\substack{  Q_{U|Y}:  \\ I(Q_Y, Q_{U|Y}) \leq R }   }  \min_{ \substack{  Q_{X|Y} : \\ H_Q(X|U) \geq \Delta }   } D(Q_{XY} \| P_{XY}), \label{eq:chp_remote_log_loss/main_results/refined_with_markov_chain}
\end{equation}
where we have the joint distribution \(Q_{XYU} = Q_{X|Y}Q_YQ_{U|Y}\).
The difference between \eqref{eq:chp_remote_log_loss/main_results/refined_with_markov_chain} and \eqref{eq:main_results_exponent_remote_log_loss} is that we are limited to minimizing over Markov chains \(X \to Y \to U\) in \eqref{eq:chp_remote_log_loss/main_results/refined_with_markov_chain}.

Such limitation results from the fact that the conventional approach, despite Kelly and Wagner's refinement, is built upon the weak converse for the dummy DMS \(Q_{XY}^n\), leading to the Markov chain in the exponent.
In this work, we overcome this limitation by choosing a different dummy source \(Q_{X^nY^n}\) that is not a DMS.
This dummy source \(Q_{X^nY^n}\) is carefully constructed based on the weak converse for \(Q_{XY}^n\), so the two distributions are closely related to each other but not the same.
Hence, unlike \(Q_{XY}^n\), the source \(Q_{X^nY^n}\) concentrates on sequences whose ``joint types'' do not adhere to the Markov chain as \(n \to \infty\), leading to the appearance of \(I_Q(X;U|Y)\) in the exponent.
When constructing the source \(Q_{X^nY^n}\), we draw inspiration from \cite[Section IV]{graczykGuessingBasedCompressed2022}, which deals with a problem related to ours; though we will show that our method to construct the dummy source as well as the obtained source is different from the one applied there, see Remark \ref{rem:chp_remote_log_loss/connection_to_guessing} for further elaboration.

To summarize, for multiterminal coding problems, we further improve the conventional sphere packing approach by choosing a carefully constructed dummy source rather than a DMS \(Q_{XY}^n\).
Combining it with the technique of marginalizing, we are able to introduce the soft Markov chain \(I_Q(X;U|Y)\) into the exponent.
For the IB source, we also employ the reverse Markov's inequality to simplify the converse proof, instead of using Fano's lower bound, which can be achieved due to the connection to list decoding.

\subsection{Strong Converse Exponent}
We first provide a complete characterization of \(F(R, \Delta)\).
\begin{theorem}
    \label{thm:strong_converse_exponent_remote_log_loss}
    For a DMS \(P_{XY}^n\), we have
    \begin{equation}
        F(R, \Delta) = \min_{ \substack{ Q_{XYU} : \\ H_{Q}(X|U) \leq \Delta } } D(Q_{XYU} \| P_{XY}Q_{U|Y}) + | I_{Q}(Y;U) - R |^{+}, \label{eq:strong_converse_exponent_remote_log_loss}
    \end{equation}
    where we have the joint distribution \(Q_{XYU} = Q_{Y}Q_{U|Y}Q_{X|YU}\) and \(|\mathcal{U}| \leq |\mathcal{X}| |\mathcal{Y}| + 2\).
\end{theorem}
\begin{proof}
    See Section \ref{sec:chp_strong_converse_remote_log_loss/achievability} and Section \ref{sec:chp_strong_converse_remote_log_loss/converse}.
\end{proof}

We next look into the rate regime for which \(F(R, \Delta) > 0\).
\begin{proposition}
    \label{prop:chp_strong_converse_remote_log_loss_rate_regime}
    For a fixed \(\Delta\), we have \(F(R, \Delta) > 0\) if and only if \(R < R(\Delta)\).
\end{proposition}
\begin{proof}
    See Appendix \ref{sec:chp_strong_converse_remote_log_loss/rate_regime}.
\end{proof}

If \(P_{XY}\) satisfies \(\P\{X = Y\} = 1\), i.e., \(Y^n\) is a noiseless observation of \(X^n\), then the IB source becomes lossy source coding under logarithmic loss. This leads to the following result.
\begin{corollary}
    \label{cor:strong_converse_exponent_lossy_log_loss}
    If \(P_{XY}\) satisfies \(\P\{X = Y\} = 1\), then we have
    \begin{equation}
        F(R, \Delta) = \min_{Q_X} D(Q_X \| P_X) + | H_Q(X)  - R - \Delta |^{+}.
    \end{equation}
\end{corollary}
\begin{proof}
    See Appendix \ref{sec:chp_strong_converse_remote_log_loss/lossy_speical_case}.
\end{proof}

We now discuss Theorem \ref{thm:strong_converse_exponent_remote_log_loss} as well as its proof.
Takeuchi and Watanabe established the strong converse exponent for the WAK problem in \cite[Theorem 2]{takeuchiTightExponentialStrong2025},  where the characterization of the exponent is identical to \eqref{eq:strong_converse_exponent_remote_log_loss}.
That is, the strong converse exponents of the IB source and the WAK problem are the same.
The reason for this phenomenon is the deep connection between the two problems.
Details on the connection are deferred to the next subsection; for now, it suffices to note that codes for the WAK problem form a subset of codes for the IB source.

Due to the connection between the two problems, our proof of Theorem \ref{thm:strong_converse_exponent_remote_log_loss} shares similarities in technique with that of \cite[Theorem 2]{takeuchiTightExponentialStrong2025}, with the key difference lying in the converse part.
For the achievability part, similar to that of Theorem \ref{thm:exponent_remote_log_loss}, we employ the type covering lemma to design the coding scheme, and construct the soft estimates based on the connection between the logarithmic loss and list decoding.
However, here we will not restrict the selected \(Q_{U|Y}\) to satisfy \(I(Q_Y, Q_{U|Y}) \leq R\).
To satisfy the rate-limit \(R\) when \(I(Q_Y, Q_{U|Y}) > R\), we select the best \(e^{nR}\) sequences from the covering set \(\mathcal{A}_n(Q_Y)\), enabling us to find a lower bound for the size of \(\bm{y} \in \mathcal{T}_n(Q_Y)\) being covered.
From this, we analyze the non-excess distortion probability by investigating the intersection between conditional type classes.
Overall, our achievability proof is similar to those of \cite[Theorem 2]{takeuchiTightExponentialStrong2025} and \cite[Theorem 1]{wuStrongConverseExponentRemote2025}; the main difference between the achievability proof of \cite[Theorem 2]{takeuchiTightExponentialStrong2025} and ours is that \cite[Theorem 2]{takeuchiTightExponentialStrong2025} also employs random binning and the minimum entropy decoder, which are not needed here.

For the converse part, we draw inspiration from the converse proof of \cite[Theorem 15.9]{polyanskiyInformationTheoryCoding2025}, which deals with hypothesis testing and large deviations.
In particular, we use the two following techniques: connecting \(p_{\mathrm{c}}(n,R,\Delta)\) to  the KL-divergence between the source \(P_{XY}^n\) and its conditional distribution given the zero-error event\footnote{This technique, linking the probability of an event to the KL-divergence, dates back at least to Marton \cite{martonSimpleProofBlowingup1986}.}; and then
for any pmf \(Q_{X^nY^n}\) and DMS \(P_{XY}^n\), we find a single-letter lower bound of the following form
\begin{equation}
    \frac{1}{n}D(Q_{X^nY^n} \| P_{XY}^n) \geq D(Q_{X_JY_J}\| P_{XY}),
\end{equation}
where \(J\) is uniformly distributed over \([n]\), commonly known as a time-sharing random variable.
However, a straightforward application of the two techniques falls short for the IB source: it cannot produce the auxiliary random variable \(U\) in the lower bound; and the second term \(|I_Q(Y;U)-R|^{+}\) is still missing.
In this work, we generalize and extend this single-letter lower bound to multiterminal problems, by introducing an auxiliary random variable \(U\) into the picture.
Instead of a lower bound, we provide an exact single-letter identity (or decomposition) of \(\frac{1}{n}D(Q_{X^nY^n} \| P_{XY}^n)\) through an auxiliary random variable \(U\), as we will see in Lemma \ref{lem:single_letterization_remote_log_loss}.
The desired lower bound will follow immediately from the identity, leading to the establishment of the converse.

\begin{remark}
Takeuchi and Watanabe established the converse proof of \cite[Theorem 2]{takeuchiTightExponentialStrong2025} using the method developed in \cite{tyagiStrongConverseUsing2020}.
Discussions of the differences between our approach and theirs will be provided in Remark \ref{rem:chp_strong_converse_remote_log_loss/difference}, after presenting our single-letter identity and lower bound.
\end{remark}

\subsection{Connections to the WAK Problem}
\label{chp:sphere_packing_exponent_WAK_problem}
Finally, we establish a code-level connection between the IB source and the WAK problem.
For reasons that will become clear shortly, we first describe source coding with a helper available at \emph{both} the transmitter \emph{and} receiver, of which the WAK Problem is a special case.
Consider a pmf \(P_{XY}\).
Assume that we have a DMS pair \((X^n, Y^n)\) following the distribution
\begin{equation}
    P_{X^nY^n}(x^n, y^n) = \prod_{i=1}^{n}P_{XY}(x_i, y_i), \qquad \forall (x^n, y^n) \in \mathcal{X}^n \times \mathcal{Y}^n.
\end{equation}
We can think of \(X^n\) as a source and \(Y^n\) as its side information.
As seen in Fig. \ref{fig:source_coding_both_sides}, a helper observes the side information \(Y^n\) and describes it through an encoder \(\varphi_n: \mathcal{Y}^n \to [e^{nB}]\) to \emph{both} a transmitter \emph{and} a receiver.
The transmitter observes the source \(X^n\)  and the rate-limited description of \(Y^n\) from the helper; and describes them to the receiver through another encoder \(\tilde{f}_n: \mathcal{X}^n \times [e^{nB}]  \to [e^{nR}]\).
A receiver reconstructs \(\hat{X}^n\) using a decoder \(\phi_n:[e^{nR}] \times [e^{nB}] \to \mathcal{X}^n\) after receiving the two descriptions.

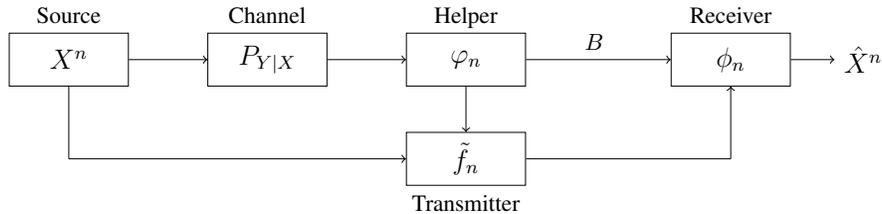
\begin{figure}[t]
    \centering
    \scalebox{0.88}{\input{figures/source_coding_helper_both_transmitter_receiver.tex}}
    \captionsetup{justification=centering}
    \caption{Source Coding with a Helper at Both Sides}
    \label{fig:source_coding_both_sides}
\end{figure}

We call the mapping vector \((\tilde{f}_n, \varphi_n, \phi_n)\) an \((n, R, B)\)-code for the DMS pair \((X^n, Y^n)\).
The performance of an \((n, R, B)\)-code is measured by the decoding error probability
\begin{equation}
    \tilde{p}_{\mathrm{e}}(n,R,B) \triangleq \P \{ \hat{X}^n \neq X^n \}.
\end{equation}
We say that the rate \(R\) is achievable if there exists a sequence of \((n,R,B)\)-codes such that \(\tilde{\lambda}(n, R, B) \to 0\).
The optimal (i.e. minimum) achievable rate for this problem \cite[Problem 10.10]{gamalNetworkInformationTheory2011} is known to be
\begin{equation}
    R_{\text{h}}(B) = \min_{P_{U|Y}} H(X|U) \qquad \text{s.t.}  \quad  I(Y; U) \leq B, \ X \to Y \to U
\end{equation}
where  \(|\mathcal{U}| \leq |\mathcal{Y}| + 1\).
That is, source coding with a helper at both sides shares the same rate limit as the WAK problem.
Here, we are interested in the reliability function \(\tilde{E}_{\text{h}}(R, B)\), i.e., the maximum \(\beta \geq 0\) for which there exists a sequence of \((n, R, B)\)-codes such that
\begin{equation}
    \liminf_{n \to \infty} -\frac{1}{n} \log \tilde{p}_{\mathrm{e}}(n, R, B) \geq \beta, \quad \text{where } R > R_{\text{h}}(B).
\end{equation}
Similarly, we also consider the strong converse exponent of this problem, and denote it by \(\tilde{F}_{\mathrm{h}}(R,B)\).

We first show that the IB source under the excess-distortion criterion is equivalent to source coding with a helper at both sides.
To differentiate the problems, we will denote the error exponent and strong converse exponent of the IB source in Theorem \ref{thm:exponent_remote_log_loss} and Theorem \ref{thm:strong_converse_exponent_remote_log_loss} by \(E_{\mathrm{IB}}(R, \Delta)\) and \(F_{\mathrm{IB}}(R, \Delta)\) respectively.


\begin{theorem}
    \label{thm:chp_sphere_packing_exponent_WAK/source_coding_helper_both_sides_equivalent_IB_source}
    Source coding with a helper at both sides under rate-limit pair \((R, B)\) is equivalent to the IB source under rate-distortion pair \((B, R)\): each code for one problem can be transformed into a code for the other. Hence, we have
    \begin{equation}
        \tilde{E}_{\mathrm{h}}(R, B) = E_{\mathrm{IB}}(B, R) \quad \text{and} \quad \tilde{F}_{\mathrm{h}}(R, B) = F_{\mathrm{IB}}(B, R).
    \end{equation}
\end{theorem}
\begin{proof}
See Section \ref{sec:chp_sphere_packing_exponent_WAK/source_coding_helper_both_sides_equivalent_IB_source}.
\end{proof}

We next turn our attention to the WAK problem, differing in that the link between the helper and the transmitter in Fig. \ref{fig:source_coding_both_sides} is now absent.
Denote its error exponent and strong converse exponent by \(E_{\mathrm{h}}(R, B)\) and \(F_{\mathrm{h}}(R,B)\) respectively.
Then, we have this immediate result.
\begin{corollary}
    \label{thm:chp_sphere_packing_exponent_WAK/WAK_connections_IB_source}
    Codes for the WAK problem under rate-limit pair \((R, B)\) form a subset of codes for the IB source under rate-distortion pair \((B, R)\). Hence, we have
    \begin{equation}
        E_{\mathrm{h}}(R, B) \leq  E_{\mathrm{IB}}(B, R) \quad \text{and} \quad F_{\mathrm{h}}(R, B) \geq F_{\mathrm{IB}}(B,R).
    \end{equation}
\end{corollary}

A sphere packing exponent for the WAK problem, which coincides with the one obtained from
\begin{equation*}
    E_{\mathrm{h}}(R, B) \leq  E_{\mathrm{IB}}(B, R),
\end{equation*}
has been derived by Kang and Liu in \cite[Theorem 4]{kangUpperBoundError2018} using the inherently typical subset lemma as well as its related results (see equation (34) in \cite{kangUpperBoundError2018}), developed by Ahlswede \emph{et al.} in \cite{ahlswedeIdentificationCompressedData1997}.
The derivation we present here is perhaps more direct. Moreover, it provides an operational interpretation for the bound by linking it to source coding with a helper at both sides, and uncovers the connections between the WAK problem and the IB source.
Further, our converse proof of Theorem \ref{thm:exponent_remote_log_loss} is of probabilistic nature, which contrasts sharply with the proof of the inherently typical subset lemma in \cite[Section III]{ahlswedeIdentificationCompressedData1997} that relies on a quantization argument; and the proof of \cite[Theorem 4]{kangUpperBoundError2018} that follows the method of types built upon the inherently typical subset lemma.
Last but not least, we also presented a cardinality bound in Theorem \ref{thm:exponent_remote_log_loss}, i.e., \(|\mathcal{U}| \leq 2|\mathcal{X}| |\mathcal{Y}| + 2\), showing the computability of the sphere packing bound, which is absent from \cite[Theorem 4]{kangUpperBoundError2018}.

The other bound \(F_{\mathrm{h}}(R, B) \geq F_{\mathrm{IB}}(B, R)\) has also been derived in \cite{takeuchiTightExponentialStrong2025}.
Furthermore, it holds that \(F_{\mathrm{h}}(R, B) = F_{\mathrm{IB}}(B, R)\), as shown in \cite[Theorem 2]{takeuchiTightExponentialStrong2025}.
The reason for this equality is that employing random binning at the transmitter and the minimum entropy decoder at the receiver can enable the receiver to correctly decode the sequence, even when the helper is absent at the transmitter.
Similar to the sphere packing exponent above, our derivation here provides an operational interpretation for the bound \(F_{\mathrm{h}}(R, B) \geq F_{\mathrm{IB}}(B, R)\) by linking the WAK problem to the IB source.

\section{Preliminaries}
\label{sec:preliminaries}
Before proceeding to the proofs of the main results, we discuss some preliminary results and techniques that will be useful down the line.

\subsection{Connection Between Log-loss and List Decoding}
\label{sec:chp_remote_log_loss/connection_list_decoding}
Shkel and Verdú \cite{shkelSingleshotApproachLossy2018} established a fundamental connection between the logarithmic loss and list decoding.
To this end, for a fixed soft reconstruction \(\hat{P} \in \mathcal{P}(\mathcal{X}^n)\) and a distortion level \(\Delta \geq 0\), we say a sequence \(x^n \in \mathcal{X}^n\) is \(\Delta\)-covered by \(\hat{P}\) if \(\bar{d}(x^n, \hat{P}) \leq \Delta\).
If \(\hat{P}\) \(\Delta\)-covers every element from a set (or list) \(\mathcal{L}_n \subseteq \mathcal{X}^n\), then the set \(\mathcal{L}_n\) is also said to be \(\Delta\)-covered by \(\hat{P}\).
\begin{lemma}
    \label{lem:connection_log_loss_list_decoding}
    For every set \(\mathcal{L}_n \subseteq \mathcal{X}^n\), there exists a soft reconstruction \(\hat{P} \in \mathcal{P}(\mathcal{X}^n)\) that \(\Delta\)-covers \(\mathcal{L}_n\) if and only if
    \begin{equation}
        | \mathcal{L}_n | \leq \floor{\exp{n\Delta}}.
    \end{equation}
\end{lemma}
\begin{proof}
The proof follows from \cite[Lemma 1]{shkelSingleshotApproachLossy2018}.
Given any \(x^n\) and \(\hat{P}\),  we have \(\bar{d}(x^n, \hat{P}) \leq \Delta\) if and only if \(\hat{P}(x^n) \geq e^{-n\Delta}\).
Hence, for a set \(\mathcal{L}_n\), there exists a \(\hat{P}\) that \(\Delta\)-covers \(\mathcal{L}_n\) only if there exists a \(\hat{P}\) satisfying \(\hat{P}(x^n) \geq e^{-n\Delta}\) for every \(x^n \in \mathcal{L}_n\).
Consequently, we must have
\begin{equation}
    1 = \sum_{x^n \in \mathcal{X}^n} \hat{P}(x^n) \geq \sum_{x^n \in \mathcal{L}_n} \hat{P}(x^n) \geq |\mathcal{L}_n| \times e^{-n\Delta},
\end{equation}
which implies \(| \mathcal{L}_n | \leq \floor{\exp{n\Delta}}\).
Note that we tightened the upper bound through the floor function.
On the other hand, if a set \(\mathcal{L}_n\) satisfies \(| \mathcal{L}_n | \leq \floor{\exp{n\Delta}}\), then we can choose \(\hat{P}\) to be uniformly distributed over \(\mathcal{L}_n\).
We see that \(\mathcal{L}_n\) is \(\Delta\)-covered by this \(\hat{P}\), since
\begin{equation}
    \hat{P}(x^n) = \frac{1}{|\mathcal{L}_n|} \geq \frac{1}{\floor{\exp{n\Delta}}} \geq e^{-n\Delta} \qquad \forall x^n \in \mathcal{L}_n,
\end{equation}
which completes the proof.
\end{proof}

We next discuss how to connect the logarithmic loss to list decoding using Lemma \ref{lem:connection_log_loss_list_decoding}.
For every  \(\hat{P} \in \mathcal{P}(\mathcal{X}^n)\), we write the set of all sequences that are \(\Delta\)-covered by \(\hat{P}\) as
\begin{equation}
    \mathcal{F}(\hat{P}, \Delta) \triangleq \{ x^n \in \mathcal{X}^n : \bar{d}(x^n, \hat{P}) \leq \Delta \}.
\end{equation}
Consider an \(x^n \in \mathcal{X}^n\) and a soft reconstruction \(\hat{P}\).
Under the excess distortion criterion, we are interested in whether \(\bar{d}(x^n, \hat{P}) \leq \Delta\) or not, i.e., whether \(x^n \in \mathcal{F}(\hat{P}, \Delta)\) or not.
Note that this is different from the average distortion criterion, where instead we are interested in the value of \(\bar{d}(x^n, \hat{P}) \).
Since \(\mathcal{F}(\hat{P}, \Delta)\) is \(\Delta\)-covered by \(\hat{P}\), from Lemma \ref{lem:connection_log_loss_list_decoding}, we have
\begin{eqnarray}
    | \mathcal{F}(\hat{P}, \Delta) | \leq \floor{ \exp{n\Delta} }.
\end{eqnarray}
Therefore, a soft reconstruction \(\hat{P}\) under the logarithmic loss can be treated as a list
\begin{equation}
    \mathcal{L}_n \triangleq \mathcal{F}(\hat{P}, \Delta)
\end{equation}
with list size \(|\mathcal{L}_n| \leq \floor{\exp{n\Delta}}\).
On the other hand, by Lemma \ref{lem:connection_log_loss_list_decoding}, every list \(\mathcal{L}_n\) satisfying \(|\mathcal{L}_n| \leq \floor{\exp{n\Delta}}\) can also be converted into a soft reconstruction \(\hat{P}\), e.g., we can choose \(\hat{P}\) to be uniformly distributed over \(\mathcal{L}_n\).
This shows the connection between the two under the excess distortion criterion.
Owing to this connection, a key technique in the following proofs is approaching the logarithmic loss through list decoding with list size \(\floor{\exp{n\Delta}}\).

\subsection{Type Covering Lemma}
We next review a key tool, known as the type covering lemma, that will be useful further on when we design lossy compression schemes.
The type covering lemma was proposed by Berger \cite{bergerRateDistortionTheory1971} to show the existence of good codes in lossy source coding.
The original version presented in \cite{bergerRateDistortionTheory1971} considered sequences that are jointly typical.
Since here our main interest is the exponents, it is more convenient to work with a modified version of the lemma, where we look into sequences that have the exact joint type.
The version presented here appears in other literature, see, e.g., \cite[Lemma 3.34]{moserAdvancedTopicsInformation2022a}.
\begin{lemma}
    \label{lemma:type_covering}
    For every joint type \(Q_{XY} \in \mathcal{P}_n(\mathcal{X} \times \mathcal{Y})\), there exists a set \(\mathcal{A}_n \subseteq \mathcal{T}_n(Q_Y)\) with
    \begin{equation}
        |\mathcal{A}_n| \ndot{\leq} e^{nI(Q_X, Q_{Y|X})}
    \end{equation}
    such that for every \(\bm{x} \in \mathcal{T}_n(Q_X)\) we can find a \(\bm{y} \in \mathcal{A}_n\) satisfying \(\hat{P}_{\bm{x}  \bm{y}} = Q_{XY}\).
\end{lemma}
\begin{proof}
    The proof is a modification of the proof in \cite[Lemma 9.1]{csiszarInformationTheoryCoding2011}, where sequences with the exact joint type are considered instead of sequences that are jointly typical. See \cite[Lemma 3.34]{moserAdvancedTopicsInformation2022a}.
\end{proof}

\subsection{Intersection of Conditional Type Classes}
We now present a key technique that will be used in the achievability proofs, followed by a discussion of its motivation.
We begin with a known result in the literature.
\begin{lemma}[{\cite[Problem 2.10]{csiszarInformationTheoryCoding2011}}]
\label{lem:intersection_conditional_type_class}
For every \(\bm{x}, \bm{u}\), and conditional type pair \((Q_{Y|X}, Q_{Y|U})\) such that \(\mathcal{T}_n(Q_{Y|X} | \bm{x}) \cap \mathcal{T}_n(Q_{Y|U} | \bm{u})  \neq \emptyset\), we have
\begin{equation}
    |\mathcal{T}_n(Q_{Y|X} | \bm{x}) \cap \mathcal{T}_n(Q_{Y|U} | \bm{u}) | \ndot{=} \max_{Q_{XYU}}\exp{nH_{Q}(Y|XU)},
\end{equation}
where the maximization is over all \(Q_{XYU}\) satisfying the three marginal distribution conditions
\begin{equation*}
    Q_{XY} = \hat{P}_{\bm{x}} \times Q_{Y|X}, \ Q_{YU} = \hat{P}_{\bm{u}} \times Q_{Y|U},\ \text{and} \ Q_{XU} = \hat{P}_{\bm{x} \bm{u}}.
\end{equation*}
\end{lemma}

With the help of Lemma \ref{lem:intersection_conditional_type_class}, we can derive the following result.
\begin{lemma}
    \label{lem:probability_intersection_conditional_type_class}
    For every \(\bm{x}, \bm{u}\), and conditional type \( Q_{Y|U}\), we have
    \begin{equation}
        P_{Y|X}^n[  \mathcal{T}_n(Q_{Y|U} | \bm{u})   | \bm{x} ] \ndot{=} \max_{Q_{XYU}} \exp\{ -n(D(Q_{Y|X} \| P_{Y|X} | \hat{P}_{\bm{x}}) + I_{Q} ( Y;U | X ) )\},  \label{eq:probability_intersection_conditional_type_class}
    \end{equation}
    where the maximization is over all \(Q_{XYU}\) satisfying the two marginal distribution conditions
    \begin{equation*}
        Q_{YU} = \hat{P}_{\bm{u}} \times Q_{Y|U} \quad \text{and} \quad Q_{XU} = \hat{P}_{\bm{x} \bm{u}}.
    \end{equation*}
\end{lemma}
\begin{proof}
The proof of Lemma \ref{lem:probability_intersection_conditional_type_class} is provided in Appendix \ref{apd:proof_probability_intersection_conditional_type_class}.
In \eqref{eq:probability_intersection_conditional_type_class}, we can also write
\begin{equation*}
    D(Q_{Y|X} \| P_{Y|X} | \hat{P}_{\bm{x}}) + I_{Q} ( Y;U | X ) = D(Q_{Y|XU} \| P_{Y|X} | \hat{P}_{\bm{x} \bm{u}}),
\end{equation*}
which will be used in Remark \ref{rem:apd_proof_intersection_conditional_type_class} of Appendix \ref{apd:proof_probability_intersection_conditional_type_class}.
\end{proof}
We discuss the motivation behind Lemma \ref{lem:probability_intersection_conditional_type_class}.
The IB source involves reconstructing a remote source \(X^n\) based on a rate-limited description of its side information sequence \(Y^n\).
The rate-limited description will be constructed using the type covering lemma; hence represented by a certain sequence \(\bm{u}\).
The receiver uses the forwarded \(\bm{u}\) to reconstruct \(X^n\).
As a result, the reconstruction error probability is analyzed by considering different conditional type classes of the form \(\mathcal{T}_n(Q_{X|U}|\bm{u})\), evaluated under the source conditional probability law \(P_{X|Y}^n\).
Given a \(\bm{Y} = \bm{y}\), this gives rise to the probability \(P_{X|Y}^n[  \mathcal{T}_n(Q_{X|U} | \bm{u})   | \bm{y} ]\), whose exponential behavior is characterized in Lemma \ref{lem:probability_intersection_conditional_type_class}, from which we see why the mutual information term \(I_Q(X;U|Y)\) appears in the exponents of information bottleneck and related problems.

\section{Achievability Proof of Theorem \ref{thm:exponent_remote_log_loss}}
\label{sec:chp_remote_log_loss/achievability}
We first focus on the error exponent of the IB source.
In particular, we establish the achievability of Theorem \ref{thm:exponent_remote_log_loss} in this section.
That is, we show that for a DMS pair \(P_{XY}^n\), there exists a sequence of coding schemes \((f_n, g_n)\) with \(||f_n|| \leq e^{nR}\) such that
\begin{align*}
    \P \{ \bar{d}(X^n, g_n(f_n(Y^n)) ) > \Delta \} \ndot{\leq} e^{-nE_{\mathrm{a}}(R, \Delta)},
\end{align*}
where we denote by \(E_{\mathrm{a}}(R,\Delta)\) the RHS of \eqref{eq:main_results_exponent_remote_log_loss}.

\subsection{Coding Scheme}
\label{sec:chp_remote_log_loss/achievability/coding_scheme}
The employed coding scheme \((f_n, g_n)\) is designed using the type covering lemma, i.e., Lemma \ref{lemma:type_covering}.
First, fix an arbitrary auxiliary alphabet \(\mathcal{U}\).
For every type \(Q_{Y} \in \mathcal{P}_n(\mathcal{Y})\), we select a conditional type \(Q_{U|Y} \in \mathcal{P}_n(\mathcal{U} | \mathcal{Y})\) such that
\(I(Q_Y, Q_{U|Y}) \leq R\).
Note that \(Q_{U|Y}\) can vary for different \(Q_Y\).
Let \(Q_U = Q_Y \cdot Q_{U|Y}\).
Following Lemma \ref{lemma:type_covering}, for every type class \(\mathcal{T}_n(Q_Y)\), we find a codebook \(\mathcal{A}_n(Q_Y) = \{\bm{u}_1, \bm{u}_2, \ldots, \bm{u}_{K_n}\} \subseteq \mathcal{T}_n(Q_U) \) that covers \(\mathcal{T}_n(Q_Y)\) under \(Q_{U|Y}\), where
\begin{equation}
    K_n \ndot{\leq} e^{nI(Q_Y, Q_{U|Y})} \leq e^{nR}.
\end{equation}
Hence, we obtain a collection of codebooks \(\{ \mathcal{A}_n(Q_Y)\}_{Q_Y \in \mathcal{P}_n(\mathcal{Y})}\).

The encoder and decoder operate on a type-per-type basis with respect to the observed sequence \(\bm{y}\), accomplished by sending an additional index indicating the type \(\hat{P}_{\bm{y}}\) of \(\bm{y}\).
Since there are at most \((n+1)^{|\mathcal{Y}|}\) types, adding this type index does not break the rate-limit \(R\) asymptotically.
Bearing this in mind, now we are ready to present the coding scheme in detail.
\begin{enumerate}
    \item Consider a type \(Q_Y \in \mathcal{P}_n(\mathcal{Y})\). For every \(\bm{y} \in \mathcal{T}_n(Q_Y)\), the encoder \(f_n\) looks up the codebook \(\mathcal{A}_n(Q_Y)\) and finds a codeword \(\bm{u}_m \in \mathcal{A}_n(Q_Y)\) such that \(\hat{P}_{\bm{u}_m | \bm{y}} = Q_{U|Y}\).
    Since the codebook \(\mathcal{A}_n(Q_Y)\) is constructed using the type covering lemma, the existence of such \(\bm{u}_m\) is guaranteed.
    If there are multiple such codewords, the encoder selects one of them arbitrarily.
    Next, the encoder \(f_n\) sends the index \(m\) to the receiver.
    \item Given a forwarded index \(m \in [K_n]\), the receiver identifies the codeword \(\bm{u}_m\) according to \(\mathcal{A}_n(Q_Y)\).
    Based on \(\bm{u}_m\), the soft estimate \(\hat{P}_m\) is constructed as follows:
    fix a small constant \(\epsilon > 0\) and consider the set
    \begin{equation*}
        \mathcal{L}_n(\bm{u}_m) \triangleq \hspace{-12pt} \bigcup_{  \substack{ Q_{X|U} : \\  H(Q_{X|U} | \hat{P}_{\bm{u}_m}) < \Delta - \epsilon} } \hspace{-12pt} \mathcal{T}_n( Q_{X|U} | \bm{u}_m),
    \end{equation*}
    i.e., \(\mathcal{L}_n(\bm{u}_m)\) is the union of all conditional type classes \( \mathcal{T}_n( Q_{X|U} | \bm{u}_m) \) satisfying \( H(Q_{X|U} | \hat{P}_{\bm{u}_m}) < \Delta - \epsilon \);
    and \(\hat{P}_m\) is chosen to be uniformly distributed over \(\mathcal{L}_n(\bm{u}_m)\).
\end{enumerate}
The above procedure is repeated for every type \(Q_Y\).
Since \(K_n \ndot{\leq} e^{nR}\) holds for every type class \(\mathcal{T}_n(Q_Y)\), the rate-limit \(R\) is asymptotically satisfied.

\subsection{Exponent Analysis}
\label{sec:chp_remote_log_loss/achievability/exponent_analysis}
Consider a type \(Q_Y\) as well as a sequence \(\bm{y} \in \mathcal{T}_n(Q_Y)\).
Assume without loss of generality that under the described coding scheme, we have \(f_n(\bm{y}) = m\).
Thus, conditioned on \(\bm{y}\), the receiver identifies \(\bm{u}_m\) from the codebook \(\mathcal{A}_n(Q_Y)\), and produces the soft estimate \(\hat{P}_m\) that is uniformly distributed over \(\mathcal{L}_n(\bm{u}_m)\).
If \( Q_{X|U}\) satisfies \( H(Q_{X|U} | \hat{P}_{\bm{u}_m}) < \Delta - \epsilon \), then
\begin{align}
    | \mathcal{T}_n( Q_{X|U} | \bm{u}_m) | & \leq (n+1)^{|\mathcal{X}| \times |\mathcal{U}|} \exp \{n H(Q_{X|U} | \hat{P}_{\bm{u}_m}) \} \\
    & < (n+1)^{|\mathcal{X}| \times |\mathcal{U}|} \exp \{n(\Delta-\epsilon)\},
\end{align}
Hence, for sufficiently large \(n\), we have
\begin{align}
    |\mathcal{L}_n(\bm{u}_m)| & \leq \hspace{-12pt} \sum_{  \substack{ Q_{X|U} : \\  H(Q_{X|U} | \hat{P}_{\bm{u}_m}) < \Delta - \epsilon} } \hspace{-12pt} | \mathcal{T}_n( Q_{X|U} | \bm{u}_m)  | \\
    & < (n+1)^{2 \times |\mathcal{X}| \times |\mathcal{U}|} \exp \{n(\Delta-\epsilon)\} \\
    & < e^{n\Delta}.
\end{align}
As a result, when \(n\) is sufficiently large, we see that
\begin{equation}
    \bar{d}( \bm{x}, \hat{P}_m ) = - \frac{1}{n} \log \frac{1}{ |\mathcal{L}_n(\bm{u}_m)| } < \Delta, \qquad \forall \bm{x} \in \mathcal{L}_n(\bm{u}_m).
\end{equation}
Therefore, conditioned on \(\bm{Y} = \bm{y}\), we have \(\bar{d}(\bm{X}, g_n(f_n(\bm{y}))) < \Delta\) if \( \bm{X} = \bm{x} \in \mathcal{L}_n(\bm{u}_m)\).
Then, the excess distortion probability can be upper bounded through
\begin{align}
    &  \P \{ \bar{d}(\bm{X}, g_n(f_n(\bm{Y}))) > \Delta | \bm{Y} = \bm{y} \} \nonumber \\
    & = \P \{ \bar{d}(\bm{X}, \hat{P}_m) > \Delta | \bm{Y} = \bm{y} \}  \\
    & \leq 1 - \P \{ \bm{X} \in \mathcal{L}_n( \bm{u}_m ) | \bm{Y} = \bm{y} \} \\
    & = \P \{ \bm{X} \in \mathcal{L}_n( \bm{u}_m )^c | \bm{Y} = \bm{y} \},
\end{align}
where
\begin{equation}
    \mathcal{L}_n(\bm{u}_m)^{c} \triangleq \hspace{-12pt} \bigcup_{  \substack{ Q_{X|U} : \\  H(Q_{X|U} | \hat{P}_{\bm{u}_m}) \geq \Delta - \epsilon} } \hspace{-12pt} \mathcal{T}_n( Q_{X|U} | \bm{u}_m).
\end{equation}
From the union bound, we obtain
\begin{equation}
    \P \{ \bm{X} \in \mathcal{L}_n( \bm{u}_m )^c | \bm{Y} = \bm{y} \} \leq \hspace{-12pt} \sum_{  \substack{ Q_{X|U} : \\  H(Q_{X|U} | \hat{P}_{\bm{u}_m}) \geq \Delta - \epsilon} } \hspace{-12pt} \P \{ \bm{X} \in \mathcal{T}_n(Q_{X|U} | \bm{u}_m)  | \bm{Y} = \bm{y}\}. \label{eq:chp_remote_log_loss/achievability/sum_over_Q_X_U}
\end{equation}
Now, the task is reduced to evaluating the probability \(\P \{ \bm{X} \in \mathcal{T}_n(Q_{X|U} | \bm{u}_m)  | \bm{Y} = \bm{y}\}\).

Recall that the joint distribution of \((\bm{X}, \bm{Y})\) is \(P_{XY}^n\).
Hence, we can treat the connection from  \(\bm{Y}\) to \(\bm{X}\) as the DMC \(P_{X|Y}\).
In other words, conditioned on \(\bm{Y} = \bm{y}\), it holds that
\begin{equation}
    \P \{ \bm{X} \in \mathcal{T}_n(Q_{X|U} | \bm{u}_m)  | \bm{Y} = \bm{y}\} = P_{X|Y}^n[ \mathcal{T}_n(Q_{X|U} | \bm{u}_m) | \bm{y} ].
\end{equation}
From Lemma \ref{lem:probability_intersection_conditional_type_class}, we see that
\begin{align}
    P_{X|Y}^n[ \mathcal{T}_n(Q_{X|U} | \bm{u}_m) | \bm{y} ] \ndot{=} \max_{Q_{XYU}} \exp \{ -n( D(Q_{X|Y} \| P_{X|Y} | \hat{P}_{\bm{y}}) + I_{Q}(X;U | Y) )  \},
\end{align}
where we consider all \(Q_{XYU}\) satisfying the two marginal distribution conditions
\begin{equation*}
    Q_{XU} = \hat{P}_{\bm{u}_m} \times Q_{X|U} \quad \text{and} \quad Q_{YU} = \hat{P}_{\bm{y} \bm{u}_m}.
\end{equation*}
Recall the assumption that \(\bm{y} \in \mathcal{T}_n(Q_Y)\) and that  we have \(\hat{P}_{\bm{u}_m | \bm{y}} = Q_{U|Y}\) under the described coding scheme, where \(Q_{U|Y}\) is the conditional type selected for \(Q_Y\).
Let \(Q_U = \hat{P}_{\bm{u}_m} = Q_Y \cdot Q_{U|Y}\).
Thus, conditioned on \(\bm{Y} = \bm{y} \in \mathcal{T}_n(Q_Y)\), we obtain
\begin{align}
    & \P \{ \bm{X} \in \mathcal{T}_n(Q_{X|U} | \bm{u}_m)  | \bm{Y} = \bm{y}\} \nonumber \\
    & \ndot{=} \max_{Q_{XYU}} \exp \{ -n( D(Q_{X|Y} \| P_{X|Y} | Q_Y) + I_{Q}(X;U | Y) )  \}, \label{eq:chp_remote_log_loss/achievability/upper_bound_Q_X_U}
\end{align}
where we consider all \(Q_{XYU}\) satisfying the two marginal distribution conditions
\begin{equation*}
    Q_{XU} = Q_U \times Q_{X|U} \quad \text{and} \quad Q_{YU} = Q_Y \times Q_{U|Y}.
\end{equation*}
Substituting \eqref{eq:chp_remote_log_loss/achievability/upper_bound_Q_X_U} into \eqref{eq:chp_remote_log_loss/achievability/sum_over_Q_X_U} yields that conditioned on \(\bm{Y} = \bm{y} \in \mathcal{T}_n(Q_Y)\), we have
\begin{align}
    &  \P \{ \bar{d}(\bm{X}, g_n(f_n(\bm{Y}))) > \Delta | \bm{Y} = \bm{y} \} \nonumber \\
    & \ndot{\leq}  \hspace{-12pt} \sum_{  \substack{ Q_{X|U} : \\  H(Q_{X|U} | \hat{P}_{\bm{u}_m}) \geq \Delta - \epsilon} } \hspace{-12pt}  \max_{Q_{XYU}} \exp \{ -n( D(Q_{X|Y} \| P_{X|Y} | Q_Y) + I_{Q}(X;U | Y) )  \} \\
    & \ndot{ = } \hspace{-12pt} \max_{  \substack{ Q_{X|U} : \\  H(Q_{X|U} | \hat{P}_{\bm{u}_m}) \geq \Delta - \epsilon} } \hspace{-12pt}  \max_{Q_{XYU}} \exp \{ -n( D(Q_{X|Y} \| P_{X|Y} | Q_Y) + I_{Q}(X;U | Y) )  \} \\
    & = \hspace{-12pt} \max_{  \substack{ Q_{X|U} : \\  H(Q_{X|U} | Q_U) \geq \Delta - \epsilon} } \hspace{-12pt}  \max_{Q_{XYU}} \exp \{ -n( D(Q_{X|Y} \| P_{X|Y} | Q_Y) + I_{Q}(X;U | Y) )  \} \\
    & = \hspace{-12pt} \max_{  \substack{ Q_{X|YU} : \\  H_Q(X|U) \geq \Delta - \epsilon} } \hspace{-12pt}   \exp \{ -n( D(Q_{X|Y} \| P_{X|Y} | Q_Y) + I_{Q}(X;U | Y) )  \}, \label{eq:chp_remote_log_loss/achievability/combination_two_maximizations}
\end{align}
where in \eqref{eq:chp_remote_log_loss/achievability/combination_two_maximizations} we combine the two maximizations, i.e., now wo consider all \(Q_{X|YU}\) such that the joint distribution \(Q_{XYU} = Q_Y Q_{U|Y} Q_{X|YU}\) satisfies \(H_Q(X|U) > \Delta - \epsilon\).

Since \eqref{eq:chp_remote_log_loss/achievability/combination_two_maximizations} holds for any type \(Q_Y\) and \(\bm{y} \in \mathcal{T}_n(Q_Y)\), the error probability of the described coding scheme can be asymptotically upper bounded through
\begin{align}
    & \P \{ \bar{d}(\bm{X}, g_n(f_n(\bm{Y}))) > \Delta \} \nonumber \\
    & = \sum_{Q_Y} \sum_{\bm{y} \in \mathcal{T}_n(Q_Y)} \P\{\bm{Y} = \bm{y}\} \times \P \{ \bar{d}(\bm{X}, g_n(f_n(\bm{Y}))) > \Delta | \bm{Y} = \bm{y} \} \nonumber  \\
    & \ndot{\leq} \sum_{Q_Y} \sum_{\bm{y} \in \mathcal{T}_n(Q_Y)} \P\{\bm{Y} = \bm{y}\} \times  \hspace{-12pt} \max_{  \substack{ Q_{X|YU} : \\  H_Q(X|U) \geq \Delta - \epsilon} } \hspace{-12pt}   \exp \{ -n( D(Q_{X|Y} \| P_{X|Y} | Q_Y) + I_{Q}(X;U | Y) )  \}\nonumber  \\
    & \ndot{=} \sum_{Q_Y} \exp\{ -nD(Q_Y \| P_Y) \} \times \hspace{-12pt}  \max_{  \substack{ Q_{X|YU} : \\  H_Q(X|U) \geq \Delta - \epsilon} } \hspace{-12pt}   \exp \{ -n( D(Q_{X|Y} \| P_{X|Y} | Q_Y) + I_{Q}(X;U | Y) )  \} \nonumber \\
    & = \sum_{Q_Y} \max_{  \substack{ Q_{X|YU} : \\  H_Q(X|U) \geq \Delta - \epsilon} } \hspace{-12pt}   \exp \{ -n( D(Q_{XY} \| P_{XY}) + I_{Q}(X;U | Y) )  \} \\
    & \ndot{=} \max_{Q_Y} \max_{  \substack{ Q_{X|YU} : \\  H_Q(X|U) \geq \Delta - \epsilon} } \hspace{-12pt}   \exp \{ -n( D(Q_{XY} \| P_{XY}) + I_{Q}(X;U | Y) )  \} \\
    & = \max_{Q_Y} \min_{ \substack{ Q_{U|Y} : \\ I(Q_Y, Q_{U|Y}) \leq R }} \max_{  \substack{ Q_{X|YU} : \\  H_Q(X|U) \geq \Delta - \epsilon} } \hspace{-12pt}   \exp \{ -n( D(Q_{XY} \| P_{XY}) + I_{Q}(X;U | Y) )  \} \label{eq:chp_remote_log_loss/achievability/optimization_over_P_U_Y},
\end{align}
where in \eqref{eq:chp_remote_log_loss/achievability/optimization_over_P_U_Y} we assume the best \(Q_{U|Y}\) is selected when constructing the compression codebooks.
Achievability is established after noticing \eqref{eq:chp_remote_log_loss/achievability/optimization_over_P_U_Y} holds for any \(\epsilon > 0\) as \(n \to \infty\).

\section{Converse Proof of Theorem \ref{thm:exponent_remote_log_loss}}
\label{sec:chp_remote_log_loss/converse}
Now that the achievability of Theorem \ref{thm:exponent_remote_log_loss} is proved, we turn to the converse in this section.
That is, we show that for a DMS pair \(P_{XY}^n\), every sequence of coding schemes \((f_n, g_n)\) with \(||f_n|| \leq e^{nR}\) must satisfy
\begin{equation}
    \P \{ \bar{d}(X^n, g_n(f_n(Y^n)) ) > \Delta \} \ndot{\geq} e^{-nE_{\mathrm{a}}(R, \Delta)},
\end{equation}
where we recall that \(E_{\mathrm{a}}(R, \Delta)\) denotes the RHS of \eqref{eq:main_results_exponent_remote_log_loss}.
To this end, fix an arbitrary sequence of coding schemes \((f_n, g_n)\) with \(||f_n|| \leq e^{nR}\), and denote the excess distortion region under \((f_n, g_n)\) by
\begin{equation}
    \mathcal{E}_n \triangleq \{ (x^n, y^n) : \bar{d}(x^n, g_n(f_n(y^n))) > \Delta \}.
\end{equation}
Hence, our task here is to show that
\begin{equation}
    P_{XY}^n [\mathcal{E}_n] \ndot{\geq} e^{-nE_{\mathrm{a}}(R, \Delta)}.
\end{equation}
We accomplish this by choosing a test source \(Q_{X^nY^n}\) that is not a DMS.


\subsection{Test Source Construction}
\label{sec:chp_remote_log_loss/auxilliary_source_construction}


First, fix a product alphabet \(\mathcal{U}^n = \mathcal{U}_1 \times \mathcal{U}_2 \times \cdots \times \mathcal{U}_n\), where the auxiliary alphabet \(\mathcal{U}_i\) for every \(i \in [n]\) will be specified further on.
Next, we consider a pmf \(Q_Y\).
For every \(i \in [n]\), we select an arbitrary conditional distribution pair \((Q_{X_i | Y_i U_i},Q_{U_i|Y_i})\) and let
\begin{equation}
    Q_{X_i Y_i U_i} \triangleq Q_{X_i | Y_i U_i} Q_{Y_i} Q_{U_i| Y_i}, \label{eq:chp_remote_log_loss/converse/replacing_Q_X_Y_U}
\end{equation}
where \(Q_{Y_i} = Q_Y\) for every \(i \in [n]\).
From \eqref{eq:chp_remote_log_loss/converse/replacing_Q_X_Y_U}, we construct a new joint distribution by
\begin{equation}
    Q_{X^n Y^n U^n} \triangleq \prod_{i=1}^{n} Q_{X_i Y_i U_i}.
\end{equation}
By marginalizing, we obtain a test source \(Q_{X^n Y^n}\) from \(Q_{X^nY^nU^n}\).
It is evident that \(Q_{X^nY^n}\) is not necessarily a DMS.
A lower bound on \(P_{XY}^n [ \mathcal{E}_n]\) will be established through \(Q_{X^n Y^n}[\mathcal{E}_n]\), i.e., the excess-distortion probability of the test source \(Q_{X^n Y^n}\).

For reasons that will become clear shortly, we turn our attention to the true source \( P_{XY}^n\).
From its marginal distribution \(P_{X_i Y_i} = P_{X_i | Y_i} P_{Y_i} = P_{X|Y}P_{Y}\) for every  \(i \in [n]\), we construct another distribution by
\begin{equation}
    P_{X_i Y_i U_i} \triangleq P_{X_i | Y_i} P_{Y_i} Q_{U_i | Y_i},
\end{equation}
where \(Q_{U_i | Y_i}\) is the same one that we selected in \eqref{eq:chp_remote_log_loss/converse/replacing_Q_X_Y_U}.
Finally, we define
\begin{equation}
    P_{X^n Y^n U^n} \triangleq \prod_{i=1}^{n} P_{X_i Y_i U_i}. \label{eq:chp_remote_log_loss/converse_definition_second_source}
\end{equation}
It is evident that the marginal distribution of \(P_{X^nY^nU^n}\) satisfies \(P_{X^n Y^n} = P_{XY}^n\).
Notice that
\begin{align}
    &\frac{1}{n} D( Q_{X^n Y^n U^n} \| P_{X^n Y^n U^n} ) \nonumber \\
    & = \frac{1}{n} \E_{Q_{X^n Y^n U^n}} \left[  \log \frac{ Q_{X^n Y^n U^n} }{ P_{X^n Y^n U^n} } \right] \\
    & = \frac{1}{n} \sum_{i=1}^n \E_{Q_{X_i Y_i U_i}} \left[  \log \frac{ Q_{X_i Y_i U_i} }{ P_{X_i Y_i U_i} } \right] \label{eq:chp_remote_log_loss/converse/memoryless_divergence} \\
    & = \frac{1}{n} \sum_{i=1}^n D(  Q_{X_i Y_i U_i} \|  P_{X_i Y_i U_i} ) \\
    & \triangleq D(  Q_{X_J Y_J U_J} \|  P_{X_J Y_J U_J} | Q_J )  \label{eq:chp_remote_log_loss/converse/time_sharing_variable} \\
    & = D(  Q_{X_J Y_J U_J J} \|  P_{X_J Y_J U_J J} ),  \label{eq:chp_remote_log_loss/converse/divergence_expectation}
\end{align}
where \eqref{eq:chp_remote_log_loss/converse/memoryless_divergence} is because the two distributions are both memoryless due to the constructions; in \eqref{eq:chp_remote_log_loss/converse/time_sharing_variable} we define the time-sharing random variable \(J \sim Q_J\) that is uniformly distributed over \([n]\).
With the above prelude, now we can proceed to the main part of the proof.

\subsection{Sphere Packing Bound}
We may write \(Q(x^n, y^n, u^n)\) instead of \(Q_{X^n Y^n U^n}(x^n, y^n, u^n)\) if there is no ambiguity. The same also applies to \(P(x^n, y^n, u^n)\) that represents \(P_{X^n Y^n U^n}(x^n, y^n, u^n)\).
For a small constant \(\epsilon > 0\), we define the divergence typical set as
\begin{equation*}
    \mathcal{D}_n^{\epsilon} \triangleq \left\{  (x^n, y^n, u^n) : \Big|\frac{1}{n} \log \frac{Q(x^n, y^n, u^n)}{ P(x^n, y^n, u^n) } - D(Q_{X_J Y_J U_J J} \| P_{X_J Y_J U_J J}) \Big| \leq \epsilon  \right\}.
\end{equation*}
For every  \(i \in [n]\), under the law of \(Q_{X_i Y_i U_i}\), let
\begin{equation}
    Z_i \triangleq \log \frac{ Q_{X_i Y_i U_i} }{ P_{X_i Y_i U_i} }.
\end{equation}
Due to the definitions of \(Q_{X_i Y_i U_i}\) and \( P_{X_i Y_i U_i}\), we see that
\begin{align}
    Z_i & = \log \frac{Q_{X_i Y_i U_i}}{ P_{X_i Y_i U_i} } \\
    & = \log \frac{ Q_{X_i | Y_i U_i} Q_{Y_i} Q_{U_i| Y_i} }{ P_{X_i | Y_i} P_{Y_i} Q_{U_i | Y_i} } \\
    & = \log \frac{ Q_{X_i | Y_i U_i} Q_{Y_i} }{ P_{X_i | Y_i} P_{Y_i} }.
\end{align}
Hence, we get
\begin{equation}
    |Z_i| = \Big| \log \frac{ Q_{X_i | Y_i U_i} Q_{Y_i} }{ P_{X_i | Y_i} P_{Y_i} } \Big| \leq \max_{x,y} \log \frac{1}{ P_{X_iY_i}(x,y) }  = \max_{x,y} - \log P_{XY}(x,y),
\end{equation}
which means that \(Z_i\) is uniformly bounded for every \(i \in [n]\), i.e., \(Z_i\) has uniformly bounded variance for every \(i \in [n]\).\footnote{We assume without loss of generality that \(P_{XY}(x,y) > 0\) for every \((x,y)\).}
From \eqref{eq:chp_remote_log_loss/converse/divergence_expectation} and Chebyshev's inequality, we have
\begin{equation}
    Q_{X^n Y^n U^n} [ \mathcal{D}_n^{\epsilon} ] = \P \big\{ \big|\frac{1}{n}\sum_{i=1}^{n}Z_i -  D(  Q_{X_J Y_J U_J J} \|  P_{X_J Y_J U_J J} ) \big| \leq \epsilon \big\} \geq 1- \alpha_n,
\end{equation}
where \(\alpha_n\) is a linear function of \(\frac{1}{n\epsilon^2}\).
It is clear that for every \((x^n, y^n, u^n) \in \mathcal{D}_n^{\epsilon}\),
\begin{equation}
    P(x^n, y^n, u^n) \geq Q(x^n, y^n, u^n) \exp \{-n(  D(Q_{X_J Y_J U_J J} \| P_{X_J Y_J U_J J}) + \epsilon )\}.
\end{equation}

As mentioned earlier, we lower bound \(P_{XY}^n[ \mathcal{E}_n ]\) through \(Q_{X^n Y^n} [ \mathcal{E}_n]\).
First, notice that
\begin{align}
    Q_{X^n Y^n} [ \mathcal{E}_n ] & = \sum_{(x^n, y^n) \in \mathcal{E}_n} Q_{X^nY^n}(x^n, y^n) \\
    & = \sum_{(x^n, y^n) \in \mathcal{E}_n} \sum_{u^n \in \mathcal{U}^n } Q(x^n, y^n, u^n) \\
    & = Q_{X^n Y^n U^n} [ \mathcal{E}_n \times \mathcal{U}^n ],
\end{align}
where we define the set
\begin{equation}
    \mathcal{E}_n \times \mathcal{U}^n \triangleq \{ (x^n, y^n, u^n) : (x^n, y^n) \in \mathcal{E}_n, u^n \in \mathcal{U}^n \}.
\end{equation}
Consider the intersection of the two sets \(\mathcal{E}_n \times \mathcal{U}^n\) and \(\mathcal{D}_n^{\epsilon}\), i.e., \( \mathcal{E}_n \times \mathcal{U}^n \cap \mathcal{D}_n^{\epsilon}\).
Since we have \(P[\mathcal{A} \cap \mathcal{B}] \geq P[\mathcal{A}] + P[\mathcal{B}]-1\) for any sets \(\mathcal{A}\) and \(\mathcal{B}\), it follows that
\begin{align}
    & Q_{X^n Y^n U^n} [\mathcal{E}_n \times \mathcal{U}^n \cap \mathcal{D}_n^{\epsilon} ] \nonumber \\
    & \geq Q_{X^n Y^n U^n} [ \mathcal{E}_n \times \mathcal{U}^n] + Q_{X^n Y^n U^n} [ \mathcal{D}_n^{\epsilon} ] - 1 \\
    & \geq Q_{X^n Y^n U^n} [ \mathcal{E}_n \times \mathcal{U}^n] - \alpha_n.
\end{align}
Now, we observe that
\begin{align}
    & P_{XY}^n[ \mathcal{E}_n ] \nonumber \\
    & = P_{X^n Y^n} [ \mathcal{E}_n] \\
    & = P_{X^n Y^n U^n} [ \mathcal{E}_n \times \mathcal{U}^n] \label{eq:chp_remote_log_loss/converse/marginalizing_technique} \\
    & \geq P_{X^n Y^n U^n} [ \mathcal{E}_n \times \mathcal{U}^n \cap \mathcal{D}_n^{\epsilon} ] \\
    & \geq Q_{X^n Y^n U^n} [ \mathcal{E}_n \times \mathcal{U}^n \cap \mathcal{D}_n^{\epsilon} ] \times  \nonumber \\
    & \hspace{3cm} \exp \{-n(  D(Q_{X_J Y_J U_J J} \| P_{X_J Y_J U_J J}) + \epsilon )\} \\
    & \geq ( Q_{X^n Y^n U^n} [ \mathcal{E}_n \times \mathcal{U}^n] - \alpha_n ) \times  \nonumber \\
    & \hspace{3cm} \exp \{-n(  D(Q_{X_J Y_J U_J J} \| P_{X_J Y_J U_J J}) + \epsilon )\} \\
    & = ( Q_{X^n Y^n} [ \mathcal{E}_n ] - \alpha_n ) \exp \{-n(  D(Q_{X_J Y_J U_J J} \| P_{X_J Y_J U_J J}) + \epsilon )\}. \label{eq:chp_remote_log_loss/converse/sphere_packing_bound}
\end{align}
Hence, the task reduces to finding a lower bound for \(Q_{X^n Y^n} [ \mathcal{E}_n ]\).

\subsection{Lower Bound for the Excess Distortion Probability}
\label{sec:chp_remote_log_loss/converse/lower_bound_excess_distortion_probability}
We now focus on finding a lower bound for \(Q_{X^n Y^n} [ \mathcal{E}_n ]\).
Recall that \(M = f_n(Y^n)\).
Applying the coding scheme \((f_n, g_n)\) to the test source \(Q_{X^n Y^n}\), we obtain a joint distribution
\begin{equation}
    Q_{X^n Y^n M} = Q_{X^n Y^n} Q_{M|Y^n},
\end{equation}
where \(Q_{M|Y^n} = \idc{M = f_n(Y^n)}\).
By marginalizing, we get \(Q_{X^n M}\).
It follows that
\begin{align}
    Q_{X^n Y^n} [ \mathcal{E}_n ] & = Q_{X^n Y^n} [ \bar{d}(X^n, g_n (f_n(Y^n))) > \Delta ] \\
    & = Q_{X^n M} [ \bar{d}(X^n, g_n (M)) > \Delta ] \\
    & = \sum_{m \in [e^{nR}]} Q_{M}(m) \times Q_{X^n | M = m} [ \bar{d}(X^n, g_n (m)) > \Delta  ]. \label{eq:chp_remote_log_loss/converse/average_over_M}
\end{align}
Given a \(m \in [e^{nR}]\), let \(g_n(m) = \hat{P}\).
Recall the set
\begin{equation}
    \mathcal{F}(\hat{P}, \Delta) \triangleq \{ x^n \in \mathcal{X}^n : \bar{d}(x^n, \hat{P}) \leq \Delta \}.
\end{equation}
Hence, we have
\begin{equation}
    Q_{X^n | M = m} [ \bar{d}(X^n, g_n (m)) > \Delta  ] = 1 - Q_{X^n | M = m} [ X^n \in \mathcal{F}(\hat{P}, \Delta)  ].
\end{equation}
For a random sequence \(X^n\), we denote by \(G: \mathcal{X}^n \to \{1, 2, \ldots, |\mathcal{X}^n| \}\) the probability rank function that ranks the sequences in decreasing order of their probabilities.
\(G(x^n)\) is also known as the optimal guessing function.
Since \(|\mathcal{F}(\hat{P}, \Delta) | \leq \floor{ \exp\{ n \Delta\} }\), we see that
\begin{equation}
    Q_{X^n | M = m} [ X^n \in \mathcal{F}(\hat{P}, \Delta)  ] \leq Q_{X^n | M = m} [ G(X^n) \leq \floor{ \exp\{ n \Delta\} } ],
\end{equation}
where \(G(x^n)\) is the rank function under \(Q_{X^n | M = m}\); the equality holds if \(\mathcal{F}(\hat{P}, \Delta)\) contains the first \( \floor{ \exp\{ n \Delta\} }\) elements that have the largest probabilities.
Hence, we obtain
\begin{equation}
    Q_{X^n | M = m} [ \bar{d}(X^n, g_n (m)) > \Delta  ] \geq Q_{X^n | M = m} [ G(X^n) > \floor{ \exp\{ n \Delta\} } ]. \label{eq:chp_remote_log_loss/converse/rank_function_bound}
\end{equation}
To proceed, we need to recall two known results in the literature, commonly known as the reverse Markov's inequality and the reverse Wyner's inequality.
\begin{lemma}[The Reverse Markov's Inequality]
    \label{lem:reverse_markov_inequality}
    If \(X\) is a nonnegative random variable and \(\P\{X \leq d \} = 1\) for some constant \(d > 0\), then we have
    \begin{equation}
        \P \{X > a \} \geq \frac{\E[X] - a}{d - a}, \qquad \text{for} \ a < \E[X].
    \end{equation}
\end{lemma}

\begin{lemma}[The Reverse Wyner's Inequality]
    \label{lem:reverse_wyner_inequality}
    Given a random variable \(X \sim P_X\) with finite alphabet \(\mathcal{X}\), for every bijective function \(f: \mathcal{X} \to \{1, 2, \ldots, |\mathcal{X}|\}\), we have
    \begin{equation}
        \E[\log f(X)] \geq H(X) - \log (1 + \log|\mathcal{X}| ).
    \end{equation}
\end{lemma}
Proofs of Lemma \ref{lem:reverse_markov_inequality} and Lemma \ref{lem:reverse_wyner_inequality} are provided in Appendix \ref{apd:proofs_reverse_markov_wyner_inequality} for completeness.
From \eqref{eq:chp_remote_log_loss/converse/rank_function_bound}, we now proceed with
\begin{align}
    & Q_{X^n | M = m} [ \bar{d}(X^n, g_n (m)) > \Delta  ] \nonumber \\
    & \geq Q_{X^n | M = m} [ G(X^n) > \floor{ \exp\{ n \Delta\} } ] \\
    & \geq  Q_{X^n | M = m} [ G(X^n) >  \exp\{ n \Delta\}  ] \\
    & = Q_{X^n | M = m} \big[ \frac{1}{n} \log G(X^n) >  \Delta  \big] \\
    & \geq \frac{ \frac{1}{n}\E_{X^n | M=m }[\log G(X^n)] - \Delta }{|\mathcal{X}| -\Delta } \\
    & \geq \frac{   \frac{1}{n}H_{Q}(X^n | M = m) - \frac{1}{n}\log(1 + n\log |\mathcal{X}|) - \Delta   }{   |\mathcal{X}| - \Delta   } \\
    & \geq \frac{   \frac{1}{n}H_{Q}(X^n | M = m) - \tau - \Delta   }{   |\mathcal{X}| - \Delta   }, \label{eq:chp_remote_log_loss/converse/lower_bound_reverse_markov_wyner_inequality}
\end{align}
where \(\tau > 0\) is a small constant; and we see that \eqref{eq:chp_remote_log_loss/converse/lower_bound_reverse_markov_wyner_inequality} holds when \(n\) is sufficiently large.
Substituting \eqref{eq:chp_remote_log_loss/converse/lower_bound_reverse_markov_wyner_inequality} into \eqref{eq:chp_remote_log_loss/converse/average_over_M}, we conclude that
\begin{align}
    Q_{X^n Y^n} [ \mathcal{E}_n ] \geq \frac{   \frac{1}{n}H_{Q}(X^n | M) - \tau - \Delta   }{   |\mathcal{X}| - \Delta   }.
\end{align}
Notice that
\begin{align}
    \frac{1}{n}H_{Q}(X^n | M) & = \frac{1}{n} \sum_{i=1}^{n}H_{Q}(X_i | M, X^{i-1}) \\
    & \geq \frac{1}{n} \sum_{i=1}^{n}H_{Q}(X_i | M, X^{i-1}, Y^{i-1}). \label{eq:chp_remote_log_loss/converse/select_U_i}
\end{align}
Recall the definition
\begin{equation}
    Q_{X^nY^nU^n} = \prod_{i=1}^{n}Q_{X_iY_iU_i} = \prod_{i=1}^{n}Q_{X_i|Y_iU_i} Q_{Y_i} Q_{U_i|Y_i},
\end{equation}
where \(Q_{Y_i} = Q_Y\) and the conditional distribution pair \((Q_{X_i|Y_iU_i}, Q_{U_i|Y_i})\) can be arbitrarily selected.
We now specify \(\mathcal{U}^n \triangleq \mathcal{U}_1 \times \mathcal{U}_2 \times \cdots \times \mathcal{U}_n\), where \(\mathcal{U}_i \triangleq ( [e^{nR}], \mathcal{X}^{i-1}, \mathcal{Y}^{i-1} )\) for every \(i \in [n]\), and let \(U_i \triangleq (M, X^{i-1}, Y^{i-1})\) that appeared in \eqref{eq:chp_remote_log_loss/converse/select_U_i}.
That is, we assume the actual \(Q_{X^nY^nU^n}\) selected at the beginning of the proof is given by
\begin{align}
    Q_{X^nY^nU^n} & = \prod_{i=1}^{n}Q_{X_i|Y_iU_i} Q_{Y_i} Q_{U_i|Y_i}, \label{eq:chp_remote_log_loss/converse/construction_under_specified_U}
\end{align}
where \(U_i = (M, X^{i-1}, Y^{i-1})\); \(Q_{X_i|Y_iU_i}\) is arbitrarily selected for every \(i \in [n]\); \(Q_{Y_i} = Q_Y\); \(Q_{U_i|Y_i}\) is well-defined and can be determined on the fly from \(i =1\) to \(n\), because \((M,Y^{i-1})\) only depends on \(Q_Y^n\) and the encoder \(f_n\), while the distribution of the past \(X^{i-1}\) is already known at index \(i\) for \(Y_i\) (we can verify this recursively from \(i =1\) to \(n\)).
Under this choice of \(Q_{X^nY^nU^n}\), we can assert that
\begin{align}
    \frac{1}{n}H_{Q}(X^n | M) & \geq \frac{1}{n}\sum_{i=1}^{n}H_{Q}(X_i | M, X^{i-1}, Y^{i-1}) \\
    &= \frac{1}{n} \sum_{i=1}^{n}H_{Q}(X_i | U_i)  \label{eq:chp_remote_log_loss/converse/replacement_U} \\
    & = H_{Q}(X_J | U_J, J).
\end{align}
Assume the chosen \(Q_{X^nY^nU^n}\) satisfies
\begin{equation}
    H_{Q}(X_J | U_J, J) \geq \Delta + \tau + \nu,
\end{equation}
where \(\nu > 0\) is a small constant.
Then, it is clear that we will have
\begin{equation}
    Q_{X^n Y^n} [ \mathcal{E}_n ] \geq \frac{\nu}{|\mathcal{X}| - \Delta}. \label{eq:chp_remote_log_loss/converse/lower_bound_decoding_error_auxiliary_source}
\end{equation}

%

\subsection{Optimizing over Test Sources}
We can optimize over the pmf \(Q_{Y}\) and the conditional distribution \(Q_{X_i | Y_i, U_i}\) for every \(i \in [n]\) in \eqref{eq:chp_remote_log_loss/converse/construction_under_specified_U} to obtain the tightest lower bound.
Hence, from \eqref{eq:chp_remote_log_loss/converse/sphere_packing_bound} and \eqref{eq:chp_remote_log_loss/converse/lower_bound_decoding_error_auxiliary_source}, we see that
\begin{align}
    & P_{X Y}^n [\mathcal{E}_n] \nonumber \\
    & \geq \hspace{-0.3cm} \max_{ \substack{  Q_{X^n Y^n U^n} : \\ H_{Q}(X_J | U_J, J) \geq \Delta + \tau + \nu }} \hspace{-0.3cm} \big( \frac{ \nu }{ |\mathcal{X}| - \Delta } - \alpha_n  \big) \times \nonumber \\
    & \hspace{3.5cm} \exp \{-n(  D(Q_{X_J Y_J U_J J} \| P_{X_J Y_J U_J J}) + \epsilon )\} \label{eq:chp_remote_log_loss/converse/optimizing_over_test_source} \\
    & \ndot{ = } \hspace{-0.3cm} \max_{ \substack{  Q_{X^n Y^n U^n} : \\ H_{Q}(X_J | U_J, J) \geq \Delta + \tau + \nu }} \hspace{-0.3cm} \exp \{-n(  D(Q_{X_J Y_J U_J J} \| P_{X_J Y_J U_J J}) + \epsilon )\} \label{eq:chp_remote_log_loss/converse/reason_sqrt_n} \\
    & = \hspace{-0.3cm} \max_{ \substack{  Q_{X_J Y_J U_J J} : \\ H_{Q}(X_J | U_J, J) \geq \Delta + \tau + \nu }} \hspace{-0.3cm} \exp \{-n(  D(Q_{X_J Y_J U_J J} \| P_{X_J Y_J U_J J}) + \epsilon )\} \label{eq:chp_remote_log_loss/converse/only_depends_on_time_sharing},
\end{align}
where in \eqref{eq:chp_remote_log_loss/converse/optimizing_over_test_source} we consider all auxiliary sources \(Q_{X^n Y^n U^n}\) constructed in the fashion stated in \eqref{eq:chp_remote_log_loss/converse/construction_under_specified_U} as well as under the condition \(H_{Q}(X_J | U_J, J) \geq \Delta + \tau + \nu\);
in \eqref{eq:chp_remote_log_loss/converse/reason_sqrt_n} we notice that for sufficiently large \(n\), we have
\begin{equation}
    \frac{ \nu }{ |\mathcal{X}| - \Delta } - \alpha_n > 0;
\end{equation}
in \eqref{eq:chp_remote_log_loss/converse/only_depends_on_time_sharing}, we notice that both the objective function and the constraint only depend on the time-sharing distribution \(Q_{X_J Y_J U_J J}\) of \(Q_{X^n Y^n U^n}\).

%

Recall the definitions of \(P_{X^nY^nU^n}\) and \(Q_{X^nY^nU^n}\), where we have \(P_{X_iY_i} = P_{XY}\) and \(Q_{Y_i} = Q_Y\).
Hence, we can write
\begin{align}
    P_{X_J Y_J U_J J} & = P_{XY}Q_{U_J, J | Y} \\
    Q_{X_J Y_J U_J J} & = Q_{Y} Q_{U_J, J | Y} Q_{X_J | Y, U_J, J}.
\end{align}
From \eqref{eq:chp_remote_log_loss/converse/only_depends_on_time_sharing}, we proceed with
\begin{align}
    & P_{X Y}^n [\mathcal{E}_n] \nonumber \\
    & \ndot{\geq} \hspace{-0.3cm} \max_{ \substack{  Q_{X_J Y_J U_J J} : \\ H_{Q}(X_J | U_J, J) \geq \Delta + \tau + \nu }} \hspace{-0.3cm} \exp \{-n(  D(Q_{X_J Y_J U_J J} \| P_{X_J Y_J U_J J}) + \epsilon )\} \\
    & = \hspace{-0.3cm} \max_{ \substack{  Q_{Y} Q_{U_J,J | Y} Q_{X_J | Y U_J,J}  : \\ H_{Q}(X_J | U_J,J) \geq \Delta + \tau + \nu }} \hspace{-0.3cm} \exp \{-n(  D( Q_{Y} Q_{U_J,J | Y} Q_{X_J | Y U_J,J} \| P_{XY} Q_{U_J,J | Y}) + \epsilon )\} \label{eq:chp_remote_log_loss/converse/independent_time_sharing} \\
    & \geq \max_{Q_Y} \hspace{-5pt} \min_{ \substack{Q_{U|Y} : \\ I(Q_Y, Q_{U|Y}) \leq R } }  \max_{ \substack{ Q_{X | Y U}  : \\ H_{Q}(X | U) \geq \Delta + \tau + \nu }} \hspace{-10pt} \exp \{-n(  D( Q_{Y} Q_{U|Y} Q_{X | Y U} \| P_{XY} Q_{U|Y} ) + \epsilon )\} \label{eq:chp_remote_log_loss/converse/optimization_over_P_U_Y} \\
    & = \max_{Q_Y} \hspace{-5pt} \min_{ \substack{Q_{U|Y} : \\ I(Q_Y, Q_{U|Y}) \leq R } }  \max_{ \substack{ Q_{X | Y U}  : \\ H_{Q}(X | U) \geq \Delta + \tau + \nu }} \hspace{-10pt} \exp \{-n(  D( Q_{XYU} \| P_{XY} Q_{U|Y} ) + \epsilon )\} \label{eq:chp_remote_log_loss/converse/sphere_packing_exponent},
\end{align}
where \eqref{eq:chp_remote_log_loss/converse/optimization_over_P_U_Y} is due to the following reasons: 
from the familiar converse technique, we have
\begin{align}
    R & \geq \frac{1}{n}H_{Q}(M)  \\
    & \geq \frac{1}{n}I_{Q}(M;Y^n) \\
    & = \frac{1}{n} \sum_{i=1}^{n} I_{Q}(Y_i;M  | Y^{i-1}) \\
    & = \frac{1}{n} \sum_{i=1}^{n} I_{Q}(Y_i; M, Y^{i-1}) \\
    & = \frac{1}{n}  \sum_{i=1}^{n} I_{Q}(Y_i; M, Y^{i-1}, X^{i-1}) \label{eq:chp_remote_log_loss/converse/weak_converse_technique} \\
    & =  \frac{1}{n} \sum_{i=1}^{n} I_{Q}(Y_i; U_i) \\
    & = I(Q_Y, Q_{U_J, J | Y}),
\end{align}
in which \eqref{eq:chp_remote_log_loss/converse/weak_converse_technique} is because the distribution of \(X_{i-1}\) only depends on \((M, X^{i-2}, Y^{i-1})\) in \eqref{eq:chp_remote_log_loss/converse/construction_under_specified_U} and hence from recursion we see that \(X^{i-1}\) only depends on \((M, Y^{i-1})\), leading to the Markov chain \(X^{i-1} \to (M, Y^{i-1}) \to Y_i\);
thus for every selected \(Q_{Y}\) we can lower bound \eqref{eq:chp_remote_log_loss/converse/independent_time_sharing} by minimizing over all \(Q_{U|Y}\) satisfying \(I(Q_Y, Q_{U|Y}) \leq R\).

Now, notice that \eqref{eq:chp_remote_log_loss/converse/sphere_packing_exponent} holds for any \(\tau, \nu, \epsilon > 0\) as \(n \to \infty\).
The proof of the cardinality bound \(|\mathcal{U}| \leq 2|\mathcal{X}| |\mathcal{Y}| + 2\) is built upon the idea of combining the KKT conditions with the support lemma, proposed by Kelly and Wagner in \cite{kellyReliabilitySourceCoding2012}. See Appendix \ref{apd:cardinality_bound_for_error_exponent_remote_log_loss} for the proof.
With this, the converse proof is completed.

\begin{remark}
As seen from the proof, \(Q_{X^n Y^n U^n}\) defined in \eqref{eq:chp_remote_log_loss/converse/construction_under_specified_U} is constructed based on the weak converse for a DMS \(Q_{XY}^n\) where for  every \(i \in [n]\) we have the Markov chain \(Q_{X_i|Y_i}Q_{Y_i}Q_{U_i | Y_i}\) with \(U_i = (M, X^{i-1}, Y^{i-1})\).
Hence, the idea behind \(Q_{X^n Y^n U^n}\) is that we construct a new product distribution by replacing \(Q_{X_i|Y_i}\) with an arbitrarily selected \(Q_{X_i | Y_i U_i}\) to overcome the Markov chain constraint.
\end{remark}

\begin{remark}
    \label{rem:chp_remote_log_loss/connection_to_guessing}
    Graczyk \emph{et al.} \cite{graczykGuessingBasedCompressed2022} studied a guessing variant of remote lossy source coding.
    As the IB source also falls under remote lossy source coding, our converse proof shares some similarities with that of \cite[Section IV]{graczykGuessingBasedCompressed2022}, e.g., we draw inspiration from their construction of the auxiliary source in \cite[Equation (93a)]{graczykGuessingBasedCompressed2022} where the Markov chain \(X^{i-1} \to (M, Y^{i-1}) \to Y_i\) also emerges.
    Yet, the distinct nature of our problem demands a different approach and also gives rise to new challenges.
    Here we highlight a limitation that arises if we attempt to follow their proof in the IB source.
    The converse proof for guessing problems relies on a different style of change of measure arguments, namely, for a positive function \(f(x)\), we have
    \begin{equation}
        \E_{P}[f(X)] = \E_{Q}[ f(X) e^{\log P(X)/Q(X)} ] \geq \exp \{ - D(Q\|P) + \E_Q[\log f(X)] \}.
    \end{equation}
    In the IB source, the objective function \(f(x)\) is an indicator function as we are interested in the excess distortion probability, which will pose difficulties under \(\log f(x)\).
    In our proof, we proceed from the conventional sphere packing approach, and use a second auxiliary source \(P_{X^nY^nU^n}\) defined in \eqref{eq:chp_remote_log_loss/converse_definition_second_source}, together with the marginalizing technique in \eqref{eq:chp_remote_log_loss/converse/marginalizing_technique} and \eqref{eq:chp_remote_log_loss/converse/sphere_packing_bound}, to bring an auxiliary random variable \(U\) into the picture.
    The source \(Q_{X^nY^nU^n}\) constructed in \eqref{eq:chp_remote_log_loss/converse/construction_under_specified_U} also differs from that in \cite[Equation (93a)]{graczykGuessingBasedCompressed2022}, which can be seen by comparing \(Q_{U_i |Y_i}\) in our case to the deterministic mapping \(Q_{U_i | M X^{i-1}Y^{i-1}}\) in theirs (see \cite[Equation (93c)]{graczykGuessingBasedCompressed2022}).
\end{remark}


\section{Achievability Proof of Theorem \ref{thm:strong_converse_exponent_remote_log_loss}}
\label{sec:chp_strong_converse_remote_log_loss/achievability}
Having established the error exponent, we move on to deriving the strong converse exponent for the IB source.
We prove the achievability of Theorem \ref{thm:strong_converse_exponent_remote_log_loss} in this section.
That is, we show that for a DMS pair \(P_{XY}^n\), there exists a sequence of coding schemes \((f_n, g_n)\) with \(||f_n|| \leq e^{nR}\) such that
\begin{equation}
    \P \{ \bar{d}(X^n, g_n(f_n(Y^n))) \leq \Delta \} \ndot{\geq} e^{-nF_{\mathrm{a}}(R, \Delta)},
\end{equation}
where we denote by \(F_{\mathrm{a}}(R, \Delta)\) the RHS of \eqref{eq:strong_converse_exponent_remote_log_loss}.

\subsection{Coding Scheme}
\label{sec:chp_strong_converse_remote_log_loss/achievability/coding_scheme}
The coding scheme is similar to the one presented in Section  \ref{sec:chp_remote_log_loss/achievability/coding_scheme}, with the main difference being that here we do not restrict the selected \(Q_{U|Y}\) to satisfy \(I(Q_Y, Q_{U|Y}) \leq R\).
The employed coding scheme \((f_n, g_n)\) is designed using the type covering lemma, i.e., Lemma \ref{lemma:type_covering}.
First, fix an arbitrary auxiliary alphabet \(\mathcal{U}\).
For every type \(Q_{Y} \in \mathcal{P}_n(\mathcal{Y})\), we select a conditional type \(Q_{U|Y} \in \mathcal{P}_n(\mathcal{U} | \mathcal{Y})\), which can vary for different \(Q_Y\).
Let \(Q_U = Q_Y \cdot Q_{U|Y}\).
Following Lemma \ref{lemma:type_covering}, for every type class \(\mathcal{T}_n(Q_Y)\), we find a set \(\mathcal{A}_n(Q_Y) = \{\bm{u}_1, \bm{u}_2, \ldots, \bm{u}_{K_n}\} \subseteq \mathcal{T}_n(Q_U) \) that covers \(\mathcal{T}_n(Q_Y)\) under \(Q_{U|Y}\), where
\begin{equation}
    K_n \ndot{\leq} e^{nI(Q_Y, Q_{U|Y})}.
\end{equation}
For every \(\bm{u}_i \in \mathcal{A}_n(Q_Y)\), we construct a soft estimate \(\hat{P}_{i}\) in the same fashion as Section \ref{sec:chp_remote_log_loss/achievability/coding_scheme}.
Fix a small constant \(\epsilon > 0\) and consider the set
\begin{equation*}
    \mathcal{L}_n(\bm{u}_i) \triangleq \hspace{-12pt} \bigcup_{  \substack{ Q_{X|U} : \\  H(Q_{X|U} | \hat{P}_{\bm{u}_i}) \leq \Delta - \epsilon} } \hspace{-12pt} \mathcal{T}_n( Q_{X|U} | \bm{u}_i),
\end{equation*}
i.e., \(\mathcal{L}_n(\bm{u}_i)\) is the union of all conditional type classes \( \mathcal{T}_n( Q_{X|U} | \bm{u}_i) \) satisfying \( H(Q_{X|U} | \hat{P}_{\bm{u}_i}) \leq \Delta - \epsilon \);
and \(\hat{P}_i\) is chosen to be uniformly distributed over \(\mathcal{L}_n(\bm{u}_i)\).

Since we may have \(K_n \ndot{>} e^{nR}\), we make the following preparations in order to satisfy the rate-limit \(R\).
According to \(\mathcal{A}_n(Q_Y) = \{\bm{u}_1, \bm{u}_2, \ldots, \bm{u}_{K_n}\}\), we partition \(\mathcal{T}_n(Q_Y)\) into \(K_n\) subregions \(\mathcal{W}_1, \mathcal{W}_2, \ldots, \mathcal{W}_{K_n}\) such that every pair \((\bm{u}_i, \mathcal{W}_i)\) for \(i \in [K_n] \) satisfies
\begin{equation}
    \hat{P}_{\bm{y}\bm{u}_i} = Q_{Y} \times Q_{U|Y}, \qquad \forall \bm{y} \in \mathcal{W}_i.
\end{equation}
The existence of such partition follows from the type covering lemma.
Further, we assume without loss of generality that
\begin{equation}
    |\mathcal{W}_1| \geq |\mathcal{W}_2| \geq \cdots \geq |\mathcal{W}_{K_n}|. \label{eq:chp_strong_converse_remote_log_loss/achievability/size_decreasing}
\end{equation}

The encoder and decoder operate on a type-per-type basis with respect to the observed sequence \(\bm{y}\), accomplished by sending an additional index indicating the type \(\hat{P}_{\bm{y}}\) of \(\bm{y}\).
Since there are at most \((n+1)^{|\mathcal{Y}|}\) types, adding this type index does not break the rate-limit \(R\) asymptotically.
Moreover, we allow the encoder to send one more index \(m=0\) besides \(m \in [e^{nR}]\) (whose purpose will become clear shortly).
In the same manner, this does not break the rate-limit \(R\) either.
Depending on whether the selected \(Q_{U|Y}\) satisfies \(I(Q_Y, Q_{U|Y}) \leq R\) or not, the coding scheme proceeds as follows:
\begin{enumerate}
    \item Suppose \(I(Q_Y, Q_{U|Y}) \leq R\), which corresponds to \(K_n \ndot{\leq} e^{nR}\).
    For every \(\bm{y} \in \mathcal{T}_n(Q_Y)\), if \(\bm{y} \in \mathcal{W}_{m}\), then the encoder sends \(m\); the receiver identifies \(\bm{u}_m\) from \(\mathcal{A}_n(Q_Y)\) and produces \(\hat{P}_m\).
    \item Suppose \(I(Q_Y, Q_{U|Y}) > R\), which corresponds to \(K_n \ndot{>} e^{nR}\).
    For every \(\bm{y} \in \mathcal{T}_n(Q_Y)\), if \(\bm{y} \in \mathcal{W}_{m}\) for a \(m \in [e^{nR}]\), then the encoder sends \(m\); the receiver identifies \(\bm{u}_m\) from \(\mathcal{A}_n(Q_Y)\) and produces \(\hat{P}_m\).
    If \(\bm{y} \in \mathcal{W}_{m}\) where \(m > e^{nR}\), then the encoder sends \(m = 0\) and the decoder produces an arbitrary soft estimate.
\end{enumerate}
The above procedure is repeated for every type \(Q_Y\).

\subsection{Exponent Analysis}
\label{sec:chp_strong_converse_remote_log_loss/achievability/exponent_analysis}
For a type \(Q_Y\), if the selected \(Q_{U | Y}\) is such that \(I(Q_{Y}, Q_{U | Y}) > R\), then there exist some \(\bm{y} \in \mathcal{T}_n(Q_Y)\) such that \(\bm{y} \in \mathcal{W}_m\) where \(m > e^{nR}\), and hence \(f_n(\bm{y}) = 0\).
Since we are looking at \(\P \{ \bar{d}(X^n, g_n(f_n(Y^n)) ) \leq \Delta \}\), these \(\bm{y}\) leading to \(f_n(\bm{y}) = 0\) are out of our interest, and we simply discard such \(\bm{y}\) and declare an error.
Hence, we need to investigate the set
\begin{equation}
    \mathcal{H}_n(Q_Y) \triangleq \{ \bm{y} \in \mathcal{T}_n (Q_Y) : f_n(\bm{y}) \in [e^{nR}]\},
\end{equation}
in particular the size of \(\mathcal{H}_n(Q_Y)\).
If the selected \(Q_{U | Y}\) satisfies \(I(Q_{Y}, Q_{U | Y}) \leq R\), then \(f_n(\bm{y}) = 0\) does not occur and we have
\begin{equation}
    |\mathcal{H}_n(Q_Y)| = |\mathcal{T}_n(Q_Y)|. \label{eq:chp_strong_converse_remote_log_loss/achievability/lower_bound_set_size_less_than_R}
\end{equation}
On the other hand, if \(I(Q_{Y}, Q_{U | Y}) > R\), then  we have
\begin{equation}
    \mathcal{H}_n(Q_Y) = \bigcup_{i=1}^{e^{nR}} \mathcal{W}_i,
\end{equation}
Due to the partitioning, it holds that
\begin{equation}
    | \mathcal{H}_n(Q_Y) | = \sum_{i=1}^{e^{nR}} |\mathcal{W}_i|. \label{eq:chp_strong_converse_remote_log_loss/achievability/sum_size}
\end{equation}
Our aim here is to lower bound \(\P \{ \bar{d}(X^n, g_n(f_n(Y^n)) ) \leq \Delta \}\), and for this it is desirable to find a lower bound for \eqref{eq:chp_strong_converse_remote_log_loss/achievability/sum_size}.
Observe that
\begin{equation}
    \frac{1}{e^{nR}} \sum_{i=1}^{e^{nR}} |\mathcal{W}_i| \geq \frac{1}{K_n}\sum_{i=1}^{K_n} |\mathcal{W}_i|, \label{eq:chp_strong_converse_remote_log_loss/achievability/lower_bound_size_through_average}
\end{equation}
since in the RHS of \eqref{eq:chp_strong_converse_remote_log_loss/achievability/lower_bound_size_through_average} we include smaller sets in the average (recall \eqref{eq:chp_strong_converse_remote_log_loss/achievability/size_decreasing}).
Thus, from \eqref{eq:chp_strong_converse_remote_log_loss/achievability/lower_bound_size_through_average}, we obtain
\begin{align}
    | \mathcal{H}_n(Q_Y) | & \geq e^{n(R - \frac{1}{n}\log K_n)} \sum_{i=1}^{K_n} |\mathcal{W}_i| \\
    & \ndot{ \geq }  e^{n(R - I_Q(Y;U) )} |\mathcal{T}_n(Q_Y)|, \label{eq:chp_strong_converse_remote_log_loss/achievability/lower_bound_set_size_greater_than_R}
\end{align}
where \eqref{eq:chp_strong_converse_remote_log_loss/achievability/lower_bound_set_size_greater_than_R} is due to \(K_n \ndot{\leq} e^{nI_Q(Y;U)}\) and the partitioning.
Combining \eqref{eq:chp_strong_converse_remote_log_loss/achievability/lower_bound_set_size_less_than_R} and \eqref{eq:chp_strong_converse_remote_log_loss/achievability/lower_bound_set_size_greater_than_R}, no matter whether \(I(Q_Y, Q_{U|Y}) > R\) or not, we see that
\begin{align}
    | \mathcal{H}_n(Q_Y) | & \ndot{\geq} e^{-n | I_Q(Y;U) -R |^{+}} |\mathcal{T}_n(Q_Y)| \\
    & \ndot{=} e^{-n( |I_Q(Y;U) -R |^{+} - H(Q_Y) )}. \label{eq:chp_strong_converse_remote_log_loss/achievability/unified_size_lower_bound}
\end{align}

With the above prelude, we now proceed to the main part of the proof. Under the described coding scheme, we have
\begin{align}
    & \P \{ \bar{d}(\bm{X}, g_n ( f_n ( \bm{Y})) ) \leq \Delta \} \nonumber \\
    & = \! \! \! \! \sum_{Q_Y \in \mathcal{P}_n(\mathcal{Y})} \sum_{\bm{y} \in \mathcal{T}_n(Q_Y)} \! \! \!  \P \{ \bm{Y} = \bm{y}\} \times \P \{ \bar{d}(\bm{X}, g_n ( f_n ( \bm{y})) ) \leq \Delta   \} \\
    & \geq \! \! \! \! \sum_{Q_Y \in \mathcal{P}_n(\mathcal{Y})} \sum_{\bm{y} \in \mathcal{H}_n(Q_Y)} \! \! \!  \P \{ \bm{Y} = \bm{y}\} \times \P \{ \bar{d}(\bm{X}, g_n ( f_n ( \bm{y})) ) \leq \Delta   \}.
\end{align}
For a \(\bm{y} \in \mathcal{H}_n(Q_Y)\), let \(f_n(\bm{y}) = m \in [e^{nR}]\).
Hence, the receiver identifies \(\bm{u}_m\) satisfying \(\hat{P}_{\bm{y}\bm{u}_m} = Q_Y \times Q_{U|Y}\); and produces \(\hat{P}_m\) as its soft estimate, which is uniformly distributed over the set
\begin{equation*}
    \mathcal{L}_n(\bm{u}_m) = \hspace{-12pt} \bigcup_{  \substack{ Q_{X|U} : \\  H(Q_{X|U} | \hat{P}_{\bm{u}_m}) \leq \Delta - \epsilon} } \hspace{-12pt} \mathcal{T}_n( Q_{X|U} | \bm{u}_m).
\end{equation*}
If \( Q_{X|U}\) satisfies \( H(Q_{X|U} | \hat{P}_{\bm{u}_m}) \leq \Delta - \epsilon \), then
\begin{align}
    | \mathcal{T}_n( Q_{X|U} | \bm{u}_m) | & \leq (n+1)^{|\mathcal{X}| \times |\mathcal{U}|} \exp \{n H(Q_{X|U} | \hat{P}_{\bm{u}_m}) \} \\
    & \leq (n+1)^{|\mathcal{X}| \times |\mathcal{U}|} \exp \{n(\Delta-\epsilon)\},
\end{align}
Hence, for sufficiently large \(n\), we have
\begin{align}
    |\mathcal{L}_n(\bm{u}_m)| & \leq \hspace{-12pt} \sum_{  \substack{ Q_{X|U} : \\  H(Q_{X|U} | \hat{P}_{\bm{u}_m}) \leq \Delta - \epsilon} } \hspace{-12pt} | \mathcal{T}_n( Q_{X|U} | \bm{u}_m)  | \\
    & \leq (n+1)^{2 \times |\mathcal{X}| \times |\mathcal{U}|} \exp \{n(\Delta-\epsilon)\} \\
    & \leq e^{n\Delta}.
\end{align}
As a result, when \(n\) is sufficiently large, we see that
\begin{equation}
    \bar{d}( \bm{x}, \hat{P}_m ) = - \frac{1}{n} \log \frac{1}{ |\mathcal{L}_n(\bm{u}_m)| } \leq \Delta, \qquad \forall \bm{x} \in \mathcal{L}_n(\bm{u}_m).
\end{equation}
Therefore, for a \(\bm{y} \in \mathcal{H}_n(Q_Y)\), we have \(\bar{d}(\bm{X}, g_n(f_n(\bm{y}))) \leq \Delta\) if \( \bm{X} = \bm{x} \in \mathcal{L}_n(\bm{u}_m)\).
Thus, the non-excess distortion probability for a \(\bm{y} \in \mathcal{H}_n(Q_Y)\) can be lower bounded through
\begin{align}
    & \P \{ \bar{d}(\bm{X}, g_n ( f_n ( \bm{y})) ) \leq \Delta  \} \nonumber \\
    & \geq \P\{ \bm{X} \in \mathcal{L}_n(\bm{u}_m) | \bm{Y} = \bm{y}\} \\
    & \geq \hspace{-12pt} \max_{  \substack{ Q_{X|U} : \\  H(Q_{X|U} | \hat{P}_{\bm{u}_m}) \leq \Delta - \epsilon} } \hspace{-12pt} \P \{  \mathcal{T}_n( Q_{X|U} | \bm{u}_m) | \bm{Y} = \bm{y} \}.  \label{eq:chp_strong_converse_remote_log_loss/achievability/max_over_Q_X_U}
\end{align}

Recall that the joint distribution of \((\bm{X}, \bm{Y})\) is \(P_{XY}^n\).
Hence, we can treat the connection from  \(\bm{Y}\) to \(\bm{X}\) as the DMC \(P_{X|Y}\).
In other words, conditioned on \(\bm{Y} = \bm{y} \in \mathcal{H}_n(Q_Y)\), we have
\begin{equation}
    \P \{ \bm{X} \in \mathcal{T}_n(Q_{X|U} | \bm{u}_m)  | \bm{Y} = \bm{y}\} = P_{X|Y}^n[ \mathcal{T}_n(Q_{X|U} | \bm{u}_m) | \bm{y} ].
\end{equation}
From Lemma \ref{lem:probability_intersection_conditional_type_class}, we see that
\begin{align}
    P_{X|Y}^n[ \mathcal{T}_n(Q_{X|U} | \bm{u}_m) | \bm{y} ] \ndot{=} \max_{Q_{XYU}} \exp \{ -n( D(Q_{X|Y} \| P_{X|Y} | \hat{P}_{\bm{y}}) + I_{Q}(X;U | Y) )  \},
\end{align}
where we consider all \(Q_{XYU}\) satisfying the two marginal distribution conditions
\begin{equation*}
    Q_{XU} = \hat{P}_{\bm{u}_m} \times Q_{X|U} \quad \text{and} \quad Q_{YU} = \hat{P}_{\bm{y} \bm{u}_m}.
\end{equation*}
Recall the assumption that \(\bm{y} \in \mathcal{T}_n(Q_Y)\) and that  we have \(\hat{P}_{\bm{u}_m | \bm{y}} = Q_{U|Y}\) under the described coding scheme, where \(Q_{U|Y}\) is the conditional type selected for \(Q_Y\).
Let \(Q_U = \hat{P}_{\bm{u}_m} = Q_Y \cdot Q_{U|Y}\).
Thus, conditioned on \(\bm{Y} = \bm{y} \in \mathcal{H}_n(Q_Y)\), we obtain
\begin{align}
    & \P \{ \bm{X} \in \mathcal{T}_n(Q_{X|U} | \bm{u}_m)  | \bm{Y} = \bm{y}\} \nonumber \\
    & \ndot{=} \max_{Q_{XYU}} \exp \{ -n( D(Q_{X|Y} \| P_{X|Y} | Q_Y) + I_{Q}(X;U | Y) )  \}, \label{eq:chp_strong_converse_remote_log_loss/achievability/upper_bound_Q_X_U}
\end{align}
where we consider all \(Q_{XYU}\) satisfying the two marginal distribution conditions
\begin{equation*}
    Q_{XU} = Q_U \times Q_{X|U} \quad \text{and} \quad Q_{YU} = Q_Y \times Q_{U|Y}.
\end{equation*}
Substituting \eqref{eq:chp_strong_converse_remote_log_loss/achievability/upper_bound_Q_X_U} into \eqref{eq:chp_strong_converse_remote_log_loss/achievability/max_over_Q_X_U} yields that for a \(\bm{Y} = \bm{y} \in \mathcal{H}_n(Q_Y)\), we have
\begin{align}
    &  \P \{ \bar{d}(\bm{X}, g_n(f_n(\bm{y}))) \leq \Delta \} \nonumber \\
    & \ndot{ \geq } \hspace{-12pt} \max_{  \substack{ Q_{X|U} : \\  H(Q_{X|U} | \hat{P}_{\bm{u}_m}) \leq \Delta - \epsilon} } \hspace{-12pt}  \max_{Q_{XYU}} \exp \{ -n( D(Q_{X|Y} \| P_{X|Y} | Q_Y) + I_{Q}(X;U | Y) )  \} \\
    & = \hspace{-12pt} \max_{  \substack{ Q_{X|YU} : \\  H_Q(X|U) \leq \Delta - \epsilon} } \hspace{-12pt}   \exp \{ -n( D(Q_{X|Y} \| P_{X|Y} | Q_Y) + I_{Q}(X;U | Y) )  \}, \label{eq:chp_strong_converse_remote_log_loss/achievability/combination_two_maximizations}
\end{align}
where in \eqref{eq:chp_strong_converse_remote_log_loss/achievability/combination_two_maximizations} we combine the two maximizations, i.e., now wo consider all \(Q_{X|YU}\) such that the joint distribution \(Q_{XYU} = Q_Y Q_{U|Y} Q_{X|YU}\) satisfies \(H_Q(X|U) \leq \Delta - \epsilon\).

Consequently, we see that
\begin{align}
    & \P \{ \bar{d}(\bm{X}, g_n ( f_n ( \bm{Y})) ) \leq \Delta \} \nonumber \\
    & \geq \! \! \! \! \sum_{Q_Y \in \mathcal{P}_n(\mathcal{Y})} \sum_{\bm{y} \in \mathcal{H}_n(Q_Y)} \! \! \!  \P \{ \bm{Y} = \bm{y}\} \times \P \{ \bar{d}(\bm{X}, g_n ( f_n ( \bm{y})) ) \leq \Delta   \} \\
    & \ndot{\geq} \! \! \! \! \sum_{Q_Y \in \mathcal{P}_n(\mathcal{Y})} \sum_{\bm{y} \in \mathcal{H}_n(Q_Y)} \P \{ \bm{Y} = \bm{y}\} \times \nonumber \\
    & \hspace{3cm} \max_{  \substack{ Q_{X|YU} : \\  H_Q(X|U) \leq \Delta - \epsilon} } \hspace{-12pt}   \exp \{ -n( D(Q_{X|Y} \| P_{X|Y} | Q_Y) + I_{Q}(X;U | Y) )  \} \\
    & \ndot{\geq} \! \! \! \! \sum_{Q_Y \in \mathcal{P}_n(\mathcal{Y})} \exp \{ -n ( D(Q_Y\|P_Y) + |I_Q(Y;U) - R|^{+} ) \} \times \nonumber \\
    & \hspace{3cm} \max_{  \substack{ Q_{X|YU} : \\  H_Q(X|U) \leq \Delta - \epsilon} } \hspace{-12pt}   \exp \{ -n( D(Q_{X|Y} \| P_{X|Y} | Q_Y) + I_{Q}(X;U | Y) )  \}  \label{eq:chp_strong_converse_remote_log_loss/achievability/plug_size_lower_bound} \\
    & = \! \! \! \! \sum_{Q_Y \in \mathcal{P}_n(\mathcal{Y})} \max_{  \substack{ Q_{X|YU} : \\  H_Q(X|U) \leq \Delta - \epsilon} } \hspace{-12pt}   \exp \{ -n( D(Q_{XYU} \| P_{XY}Q_{U|Y}) + |I_Q(Y;U) - R|^{+} )  \} \\
    & \geq \! \! \! \max_{Q_Y \in \mathcal{P}_n(\mathcal{Y})} \max_{  \substack{ Q_{X|YU} : \\  H_Q(X|U) \leq \Delta - \epsilon} } \hspace{-12pt}   \exp \{ -n( D(Q_{XYU} \| P_{XY}Q_{U|Y}) + |I_Q(Y;U) - R|^{+} )  \} \\
    & = \! \! \! \! \max_{Q_Y \in \mathcal{P}_n(\mathcal{Y})} \max_{Q_{U|Y}} \max_{  \substack{ Q_{X|YU} : \\  H_Q(X|U) \leq \Delta - \epsilon} } \hspace{-12pt}   \exp \{ -n( D(Q_{XYU} \| P_{XY}Q_{U|Y}) + |I_Q(Y;U) - R|^{+} )  \}    \label{eq:chp_strong_converse_remote_log_loss/achievability/optimize_over_Q_U_Y} \\
    & = \! \! \! \! \max_{  \substack{ Q_{XYU} : \\  H_Q(X|U) \leq \Delta - \epsilon} } \hspace{-12pt}   \exp \{ -n( D(Q_{XYU} \| P_{XY}Q_{U|Y}) + |I_Q(Y;U) - R|^{+} )  \}, \label{eq:chp_strong_converse_remote_log_loss/achievability/final_combination}
\end{align}
where \eqref{eq:chp_strong_converse_remote_log_loss/achievability/plug_size_lower_bound} is due to \eqref{eq:chp_strong_converse_remote_log_loss/achievability/unified_size_lower_bound} and
\begin{equation}
    \P \{ \bm{Y} = \bm{y} \} = \exp\{ - n(D(Q_Y \| P_Y) + H_Q(Y))\}, \qquad \forall \bm{y} \in \mathcal{T}_n(Q_Y);
\end{equation}
in \eqref{eq:chp_strong_converse_remote_log_loss/achievability/optimize_over_Q_U_Y} we assume the best \(Q_{U|Y}\) is selected when constructing the compression codebook (the largest correct reconstruction probability is of interest, so here we use maximization).
Achievability is established after noticing \eqref{eq:chp_strong_converse_remote_log_loss/achievability/final_combination} holds for any \(\epsilon > 0\) as \(n \to \infty\).

\section{Converse Proof of Theorem \ref{thm:strong_converse_exponent_remote_log_loss}}
\label{sec:chp_strong_converse_remote_log_loss/converse}
Now that the achievability of Theorem \ref{thm:strong_converse_exponent_remote_log_loss} is proved, we address the converse in this section.
That is, we show that for a DMS pair \(P_{XY}^n\), every sequence of coding schemes \((f_n, g_n)\) with \(||f_n|| \leq e^{nR}\) must satisfy
\begin{equation}
    \P \{ \bar{d}(X^n, g_n(f_n(Y^n))) \leq \Delta \} \leq e^{-nF_{\mathrm{a}}(R, \Delta)},
\end{equation}
where we recall that \(F_{\mathrm{a}}(R, \Delta)\) denotes the RHS of \eqref{eq:strong_converse_exponent_remote_log_loss}.
To this end, fix an arbitrary sequence of coding schemes \((f_n, g_n)\) with \(||f_n|| \leq e^{nR}\), and denote the non-excess distortion region under \((f_n, g_n)\) by
\begin{equation}
    \mathcal{E}_n^c \triangleq \{ (x^n, y^n) : \bar{d}(x^n, g_n(f_n(y^n))) \leq \Delta \}.
\end{equation}
Hence, our task here is to show that
\begin{equation}
    P_{XY}^n [\mathcal{E}_n^c] \leq e^{-nF_{\mathrm{a}}(R, \Delta)}.
\end{equation}
Consider the conditional distribution given the event \(\mathcal{E}_n^{c}\), defined as
\begin{equation}
    \tilde{P}_{X^nY^n}(x^n, y^n) \triangleq \frac{ P_{XY}^n(x^n, y^n) }{ P_{XY}^n [\mathcal{E}_n^c] }, \qquad \forall (x^n, y^n) \in \mathcal{E}_n^{c}. \label{eq:chp_strong_converse_remote_log_loss/converse/conditional_distribution}
\end{equation}
It immediately follows that we have the following identity
\begin{equation}
    - \log P_{XY}^n [\mathcal{E}_n^c] =  D( \tilde{P}_{X^n Y^n} \| P_{XY}^n). \label{eq:chp_strong_converse_remote_log_loss/converse/change_of_measurement}
\end{equation}
As a result, the task now is reduced to showing
\begin{equation}
    \frac{1}{n} D( \tilde{P}_{X^n Y^n} \| P_{XY}^n) \geq F_{\mathrm{a}}(R, \Delta).
\end{equation}
We achieve this through the following single-letter identity lemma.

\subsection{Single-letter Identity Lemma}
\begin{lemma}
    \label{lem:single_letterization_remote_log_loss}
    For every distribution \(Q_{X^nY^n}\), DMS \(P_{XY}^n\), and mapping \(f_n:\mathcal{Y}^n \to \mathbb{N}^{+}\), we have
    \begin{equation}
        \frac{1}{n}D(Q_{X^n Y^n} \| P_{XY}^n) = D(Q_{X_JY_JU_JJ} \| P_{X_JY_J}Q_{U_J J| Y_J}) + I_Q(Y_J;U_J, J) - \frac{1}{n} H_Q(M),
    \end{equation}
    where \(M = f_n(Y^n)\); \(J \sim Q_J\) is uniformly distributed over \([n]\); \(U_i = (M, X^{i-1}, Y^{i-1})\) for every \(i \in [n]\); and \(Q_{X_JY_JU_JJ} = Q_{X_J | Y_J U_J J}Q_{Y_J}Q_{U_J J| Y_J}\).
\end{lemma}
\begin{proof}
    Let \(P_{X^n Y^n} \triangleq P_{XY}^n\), and to simplify notation we will write
    \begin{equation}
        \sum_{x} g(Q_X) \triangleq \sum_{x \in \mathcal{X}} g(Q_X(x)),
    \end{equation}
    for any alphabet \(\mathcal{X}\), pmf \(Q_X\), and function \(g\).
    Similar notation will also apply to, e.g.,
    \begin{equation}
        \sum_{x , y} g(Q_X, Q_Y) \triangleq \sum_{x \in \mathcal{X} , y \in \mathcal{Y}} g(Q_X(x), Q_Y(y)).
    \end{equation}
    With the above in mind, we now start with
    \begin{align}
        & \frac{1}{n} D(Q_{X^n Y^n} \| P_{XY}^n) \nonumber \\
        & = \frac{1}{n} D(Q_{X^n Y^n} Q_{M | Y^n} \| P_{XY}^n Q_{M|Y^n}) \label{eq:chp_strong_converse_remote_log_loss/converse/split_lemma} \\
        & = \frac{1}{n} D(Q_{X^n |  Y^n M} Q_{Y^n M} \| P_{X^n | Y^n} P_{Y^n M}) \\
        & = \frac{1}{n} D(Q_{Y^nM} \| P_{Y^n M}) + \frac{1}{n} D(Q_{X^n| Y^n M} \| P_{X^n|Y^n} | Q_{Y^n M}  ), \label{eq:chp_strong_converse_remote_log_loss/converse/two_divergence}
    \end{align}
    where we have \(Q_{M|Y^n} = \idc{M = f_n(Y^n)}\) in \eqref{eq:chp_strong_converse_remote_log_loss/converse/split_lemma}.
    Next, we notice that
    \begin{align}
        & \frac{1}{n} D(Q_{Y^nM} \| P_{Y^n M}) \nonumber \\
        & = \frac{1}{n} D(Q_{Y^n} \| P_{Y}^n) \\
        & = \frac{1}{n} \sum_{y^n  } Q_{Y^n} \log \frac{ Q_{Y^n} }{ P_{Y}^n } \\
        & = - \frac{1}{n} H_Q(Y^n) + \frac{1}{n} \sum_{y^n  } Q_{Y^n} \log \frac{ 1 }{ P_{Y}^n } \\
        & = - \frac{1}{n} H_Q(Y^n) - \frac{1}{n} \sum_{y^n  } Q_{Y^n} \log \tilde{Q}_{Y^n} +  \frac{1}{n} \sum_{y^n  } Q_{Y^n} \log \frac{ \tilde{Q}_{Y^n} }{ P_{Y}^n } \label{eq:chp_strong_converse_remote_log_loss/converse/product_Q_Y} \\
        & =  - \frac{1}{n} H_Q(Y^n) - \frac{1}{n} \sum_{i=1}^n \sum_{y_i  } Q_{Y_i} \log Q_{Y_i} +  \frac{1}{n}  \sum_{i=1}^n \sum_{y_i  } Q_{Y_i} \log \frac{ Q_{Y_i} }{ P_{Y_i} } \\
        & =  - \frac{1}{n} H_Q(Y^n) + H_Q(Y_J|J) + D(Q_{Y_J} \| P_{Y_J} | Q_J), \label{eq:chp_strong_converse_remote_log_loss/converse/two_divergence_1}
    \end{align}
    where in \eqref{eq:chp_strong_converse_remote_log_loss/converse/product_Q_Y} we have \(\tilde{Q}_{Y^n} \triangleq \prod_{i=1}^{n}Q_{Y_i}\), in which \(Q_{Y_i}\) is the marginal distribution of \(Q_{Y^n}\).
    Now, we turn our attention to the conditional divergence term.
    Observe that
    \begin{align}
        & \frac{1}{n} D(Q_{X^n| Y^n M} \| P_{X^n|Y^n} | Q_{Y^n M}  )  \nonumber \\
        & = \frac{1}{n} \sum_{(x^n, y^n, m)} Q_{X^n Y^n M} \log \frac{ Q_{X^n| Y^n M} }{ P_{X^n|Y^n}  } \\
        & = \frac{1}{n} \sum_{(x^n, y^n, m)} Q_{X^n Y^n M} \log \frac{ Q_{X^n| Y^n M} }{ Q_{X^n | Y^nU^n}  } + \frac{1}{n} \sum_{(x^n, y^n, m)} Q_{X^n Y^n M} \log \frac{ Q_{X^n | Y^nU^n} }{ P_{X^n | Y^n }  }, \label{eq:chp_strong_converse_remote_log_loss/converse/conditional_divergence}
    \end{align}
    where we have \(U_i = (M, X^{i-1}, Y^{i-1})\) and
    \begin{equation}
        Q_{X^n | Y^n U^n} \triangleq \prod_{i=1}^{n} Q_{X_i | Y_i U_i} = \prod_{i=1}^{n} Q_{X_i | M, X^{i-1}, Y^{i}}.
    \end{equation}
    Next, we see that
    \begin{align}
        & \frac{1}{n} \sum_{(x^n, y^n, m)} Q_{X^n Y^n M} \log \frac{ Q_{X^n| Y^n M} }{ Q_{X^n | Y^nU^n}  }  \nonumber \\
        & = - \frac{1}{n}H_Q(X^n | Y^n, M) + \frac{1}{n} \sum_{(x^n, y^n, m)} Q_{X^n Y^n M} \log \frac{ 1 }{ Q_{X^n | Y^nU^n}  } \\
        & = - \frac{1}{n}H_Q(X^n | Y^n, M) + \frac{1}{n} \sum_{(x^n, y^n, m)} Q_{X^n Y^n M} \log \frac{ 1 }{  \prod_{i=1}^{n} Q_{X_i | M, X^{i-1}, Y^{i}} } \\
        & = - \frac{1}{n}H_Q(X^n | Y^n, M) + \frac{1}{n} \sum_{i=1}^{n} \sum_{(x^{i}, y^{i}, m)} Q_{X^{i} Y^{i} M} \log \frac{ 1 }{  Q_{X_i | M, X^{i-1}, Y^{i}} } \\
        & = - \frac{1}{n}H_Q(X^n | Y^n, M) + \frac{1}{n} \sum_{i=1}^{n} \sum_{(x_{i}, y_{i}, u_i)} Q_{X_{i} Y_{i} U_i} \log \frac{ 1 }{  Q_{X_i | Y_i, U_i} } \\
        & =  - \frac{1}{n}H_Q(X^n | Y^n, M)  + H_Q(X_J | Y_J, U_J, J). \label{eq:chp_strong_converse_remote_log_loss/converse/conditional_divergence_1}
    \end{align}
    In the same fashion, we also have
    \begin{align}
        & \frac{1}{n} \sum_{(x^n, y^n, m)} Q_{X^n Y^n M} \log \frac{ Q_{X^n | Y^nU^n} }{ P_{X^n | Y^n }  } \nonumber \\
        & = \frac{1}{n} \sum_{(x^n, y^n, m)} Q_{X^n Y^n M} \log \frac{ \prod_{i=1}^{n} Q_{X_i | M, X^{i-1}, Y^{i}} }{ \prod_{i=1}^{n}P_{X_i | Y_i }  } \\
        & = \frac{1}{n} \sum_{i=1}^{n}  \sum_{ (x^i, y^i,m) } Q_{X^{i} Y^{i} M} \log \frac{Q_{X_i | M, X^{i-1}, Y^{i}} }{ P_{X_i | Y_i } } \\
        & = \frac{1}{n} \sum_{i=1}^{n} \sum_{ (x_i, y_i, u_i) } Q_{X_i Y_i U_i } \log \frac{Q_{X_i | Y_{i}, U_i} }{ P_{X_i | Y_i } }  \\
        & = \frac{1}{n} \sum_{i=1}^{n} \sum_{ (x_i, y_i, u_i) } Q_{X_i Y_i U_i } \log \frac{Q_{X_i | Y_{i}, U_i} Q_{Y_i U_i} }{ P_{X_i | Y_i } Q_{Y_i U_i} } \\
        & = D( Q_{X_J Y_J U_J} \| P_{X_J | Y_J} Q_{Y_J U_J} | Q_J ). \label{eq:chp_strong_converse_remote_log_loss/converse/conditional_divergence_2}
    \end{align}
    Substituting  \eqref{eq:chp_strong_converse_remote_log_loss/converse/conditional_divergence_1} and \eqref{eq:chp_strong_converse_remote_log_loss/converse/conditional_divergence_2} into \eqref{eq:chp_strong_converse_remote_log_loss/converse/conditional_divergence}, we obtain
    \begin{align}
        & \frac{1}{n} D(Q_{X^n| Y^n M} \| P_{X^n|Y^n} | Q_{Y^n M}  ) \nonumber \\
        & =  - \frac{1}{n}H_Q(X^n | Y^n, M)  + H_Q(X_J | Y_J, U_J, J) + D( Q_{X_J Y_J U_J} \| P_{X_J | Y_J} Q_{Y_J U_J} | Q_J ). \label{eq:chp_strong_converse_remote_log_loss/converse/two_divergence_2}
    \end{align}

    We now return to \eqref{eq:chp_strong_converse_remote_log_loss/converse/two_divergence}.
    From \eqref{eq:chp_strong_converse_remote_log_loss/converse/two_divergence_1} and \eqref{eq:chp_strong_converse_remote_log_loss/converse/two_divergence_2}, we get
    \begin{align}
        & \frac{1}{n} D(Q_{X^n Y^n} \| P_{XY}^n) \nonumber \\
        & = \frac{1}{n} D(Q_{Y^nM} \| P_{Y^n M}) + \frac{1}{n} D(Q_{X^n| Y^n M} \| P_{X^n|Y^n} | Q_{Y^n M}  ) \\
        & = - \frac{1}{n} H_Q(Y^n) + H_Q(Y_J|J) + D(Q_{Y_J} \| P_{Y_J}|Q_J) - \frac{1}{n}H_Q(X^n | Y^n, M)   \nonumber \\
        & \hspace{3cm}  + H_Q(X_J | Y_J, U_J, J) + D( Q_{X_J Y_J U_J} \| P_{X_J | Y_J} Q_{Y_J U_J} |Q_J  ).
    \end{align}
    Notice that
    \begin{align}
        & D(Q_{Y_J} \| P_{Y_J} | Q_J) + D( Q_{X_J Y_J U_J} \| P_{X_J | Y_J} Q_{Y_J U_J} | Q_J ) \nonumber \\
        & = D( Q_{X_J Y_J U_J} \| P_{X_J Y_J} Q_{U_J | Y_J} |Q_J ) \\
        & = D( Q_{X_J Y_J U_J J} \| P_{X_J Y_J} Q_{U_J J| Y_J}).
    \end{align}
    Hence, the remaining task is to show that
    \begin{align}
        & H_Q(Y_J |J) + H_Q(X_J | Y_J, U_J,J) - \frac{1}{n} H_Q(Y^n) - \frac{1}{n}H_Q(X^n | Y^n, M) \nonumber \\
        & = I_Q(Y_J;U_J,J) - \frac{1}{n} H_Q(M).
    \end{align}
    Since \(H_Q(M) = I_Q(Y^n;M)\), we have
    \begin{align}
        & H_Q(Y_J|J) + H_Q(X_J | Y_J, U_J,J) - \frac{1}{n} H_Q(Y^n) - \frac{1}{n}H_Q(X^n | Y^n, M) + \frac{1}{n} I_Q(Y^n;M) \nonumber \\
        & = H_Q(Y_J|J) + H_Q(X_J | Y_J, U_J,J) - \frac{1}{n} H_Q(Y^n |M) - \frac{1}{n}H_Q(X^n | Y^n, M) \\
        & = H_Q(Y_J|J) + H_Q(X_J | Y_J, U_J,J) - \frac{1}{n}H_Q(X^n, Y^n | M)\\
        & =  H_Q(Y_J|J) + H_Q(X_J | Y_J, U_J,J) - \frac{1}{n} \sum_{i=1}^{n}H_Q(X_i, Y_i | M, X^{i-1}, Y^{i-1}) \\
        & = H_Q(Y_J|J) + H_Q(X_J | Y_J, U_J,J) - H_Q(X_J, Y_J | U_J,J) \\
        & = H_Q(Y_J|J) - H_Q(Y_J| U_J,J) \\
        & = I_Q(Y_J;U_J|J) \\
        & = I_Q(Y_J;U_J, J),
    \end{align}
    which completes the proof.
\end{proof}

We are now ready to derive the following single-letter lower bound from Lemma \ref{lem:single_letterization_remote_log_loss}.
\begin{corollary}
    \label{cor:single_letterization_lower_bound_remote_log_loss}
    For every distribution \(Q_{X^nY^n}\), DMS \(P_{XY}^n\), and mapping \(f_n:\mathcal{Y}^n \to [e^{nR}]\), we have
    \begin{equation}
        \frac{1}{n}D(Q_{X^n Y^n} \| P_{XY}^n) \geq D(Q_{X_JY_JU_JJ} \| P_{X_JY_J}Q_{U_J J| Y_J}) + | I_Q(Y_J;U_J,J) - R |^{+},
    \end{equation}
    where \(M = f_n(Y_n)\); \(J \sim Q_J\) is uniformly distributed over \([n]\); \(U_i = (M, X^{i-1}, Y^{i-1})\) for every \(i \in [n]\); and \(Q_{X_JY_JU_JJ} = Q_{X_J | Y_J U_JJ}Q_{Y_J}Q_{U_J J| Y_J}\).
\end{corollary}
\begin{proof}
Due to \(R \geq \frac{1}{n} H_Q(M)\), from Lemma \ref{lem:single_letterization_remote_log_loss}, we see that
\begin{equation}
    \frac{1}{n}D(Q_{X^n Y^n} \| P_{XY}^n) \geq D(Q_{X_JY_JU_JJ} \| P_{X_JY_J}Q_{U_J J| Y_J}) + I_Q(Y_J;U_J,J) - R \label{eq:single_letterization_lower_bound_remote_log_loss_1}
\end{equation}
On other hand, since
\begin{align}
    \frac{1}{n} H_Q(M) & = \frac{1}{n} I_Q(Y^n;M) \\
    & = \frac{1}{n} \sum_{i=1}^{n} I_Q(Y_i; M | Y^{i-1}) \\
    & \leq \frac{1}{n} \sum_{i=1}^{n} I_Q(Y_i; M , Y^{i-1}) \\
    & \leq \frac{1}{n} \sum_{i=1}^{n} I_Q(Y_i; M , Y^{i-1}, X^{i-1}) \\
    & = \frac{1}{n} \sum_{i=1}^{n} I_Q(Y_i; U_i) \\
    & = I_Q(Y_J; U_J,J),
\end{align}
from Lemma \ref{lem:single_letterization_remote_log_loss}, we also see that
\begin{equation}
    \frac{1}{n}D(Q_{X^n Y^n} \| P_{XY}^n) \geq D(Q_{X_JY_JU_JJ} \| P_{X_JY_J}Q_{U_J J| Y_J}). \label{eq:single_letterization_lower_bound_remote_log_loss_2}
\end{equation}
The proof is completed after combining \eqref{eq:single_letterization_lower_bound_remote_log_loss_1} and \eqref{eq:single_letterization_lower_bound_remote_log_loss_2}.
\end{proof}

\begin{remark}
    \label{rem:chp_strong_converse_remote_log_loss/difference}
    Both the converse proofs of Takeuchi and Watanabe \cite{takeuchiTightExponentialStrong2025} and ours proceed from the change of measure argument in \eqref{eq:chp_strong_converse_remote_log_loss/converse/change_of_measurement}.
    Our proofs diverge in the derivation of the lower bound.
    A special case of Corollary \ref{cor:single_letterization_lower_bound_remote_log_loss}, where \(Q_{X^nY^n} = \tilde{P}_{X^nY^n}\) with \(\tilde{P}_{X^nY^n}\) given by \eqref{eq:chp_strong_converse_remote_log_loss/converse/conditional_distribution}, has appeared in \cite[Equation (93)]{takeuchiTightExponentialStrong2025}.
    The derivation of equation (93) in \cite{takeuchiTightExponentialStrong2025} is built upon the method developed in \cite{tyagiStrongConverseUsing2020}; and most importantly the trick in \cite[Proposition 1]{tyagiStrongConverseUsing2020}.
    Applying this trick requires insights to identify the (conditional) mutual information between suitable random variables, which is not straightforward---see (84) and (89) of \cite{takeuchiTightExponentialStrong2025}.
    The identified conditional mutual information in  \cite[Equation (84)]{takeuchiTightExponentialStrong2025} also makes use of the special property of \(\tilde{P}_{X^nY^n}\), meaning the derivation there cannot be easily generalized to hold for any \(Q_{X^nY^n}\).
    Moreover, the two lower bounds in (87) and (92) of  \cite{takeuchiTightExponentialStrong2025} do not appear related to each other immediately.
    In contrast, in our proof, Lemma \ref{lem:single_letterization_remote_log_loss} provides an exact single-letter identity for every \(Q_{X^nY^n}\), leading to Corollary \ref{cor:single_letterization_lower_bound_remote_log_loss} that shares the same generality; and our derivation does not rely on the method from \cite{tyagiStrongConverseUsing2020}.
    Further, both of the lower bounds in \eqref{eq:single_letterization_lower_bound_remote_log_loss_1} and \eqref{eq:single_letterization_lower_bound_remote_log_loss_2} follow immediately from Lemma \ref{lem:single_letterization_remote_log_loss}, revealing an underlying connection between them.
\end{remark}

\subsection{Main Proof}
Let \(M = f_n(Y^n)\), and for every \(m \in [e^{nR}]\) we consider the set
\begin{equation}
    \mathcal{E}_n^c(m) \triangleq \{ x^n : \bar{d}(x^n, g_n(m)) \leq \Delta \}.
\end{equation}
Since \(g_n(m)\) is a soft estimate, from Lemma \ref{lem:connection_log_loss_list_decoding}, we see that
\begin{equation}
    |\mathcal{E}_n^c(m) | \leq \floor{ \exp \{ n\Delta\} }, \qquad \forall m \in [e^{nR}].
\end{equation}
Recall that \(\tilde{P}_{X^n Y^n}\) is supported on the set
\begin{equation}
    \mathcal{E}_n^c = \{ (x^n, y^n) : \bar{d}(x^n, g_n(f_n(y^n))) \leq \Delta \}.
\end{equation}
Therefore, under \(\tilde{P}_{X^n Y^n}\), we must have
\begin{equation}
    H_{\tilde{P}}(X^n | M = m ) \leq \log |\mathcal{E}_n^c(m) | \leq \log  \floor{ \exp \{ n\Delta\} } \leq n \Delta, \qquad \forall m \in [e^{nR}].
\end{equation}
This implies \( H_{\tilde{P}}(X^n | M) \leq n\Delta\), which leads to
\begin{align}
    \Delta & \geq \frac{1}{n} H_{\tilde{P}}(X^n | M) \\
    & = \frac{1}{n} \sum_{i=1}^{n} H_{\tilde{P}}(X_i | M, X^{i-1}) \\
    & \geq \frac{1}{n} \sum_{i=1}^{n} H_{\tilde{P}}(X_i | M, X^{i-1} , Y^{i-1}) \\
    &  = \frac{1}{n} \sum_{i=1}^{n} H_{\tilde{P}}(X_i | U_i)  \\
    & = H_{\tilde{P}}(X_J | U_J,J).
\end{align}
Since \(\tilde{P}_{X^n Y^n}\) satisfies \(H_{\tilde{P}}(X_J | U_J,J) \leq \Delta\), from Corollary \ref{cor:single_letterization_lower_bound_remote_log_loss}, we can assert that
\begin{align}
    \frac{1}{n} D( \tilde{P}_{X^n Y^n} \| P_{XY}^n) \geq \min_{ \substack{ Q_{XYU} : \\ H_{Q}(X|U) \leq \Delta } } D(Q_{XYU} \| P_{XY}Q_{U|Y}) + | I_{Q}(Y;U) - R |^{+},
\end{align}
where we notice \(P_{X_JY_J} = P_{XY}\) and consider the minimization over all \(Q_{XYU}\) satisfying \(H_{Q}(X|U) \leq \Delta\).
The cardinality bound follows from a standard application of the support lemma and is provided in Appendix \ref{apd:cardinality_bound_for_strong_converse_exponent_remote_log_loss}, which completes the converse proof.


\section{Proof of Theorem \ref{thm:chp_sphere_packing_exponent_WAK/source_coding_helper_both_sides_equivalent_IB_source}}
\label{sec:chp_sphere_packing_exponent_WAK/source_coding_helper_both_sides_equivalent_IB_source}
In this section, we establish the equivalence between the IB source and source coding with a helper at both sides.
We first prove the equivalence between the two problems on a code level; then show that the two problems have the exact same minimum decoding error probability.
First, recall Lemma \ref{lem:connection_log_loss_list_decoding}, i.e., for every set \(\mathcal{L}_n \subseteq \mathcal{X}^n\), there exists a soft reconstruction \(\hat{P} \in \mathcal{P}(\mathcal{X}^n)\) that \(\Delta\)-covers \(\mathcal{L}_n\) if and only if \(| \mathcal{L}_n | \leq \floor{\exp{n\Delta}}\).

\subsection{Equivalence Between Codes}
For source coding with a helper at both sides, let \(L = \varphi_n(Y^n)\) and \(M = \tilde{f}_n(X^n, L)\).
Conditioned on \(L = l \), the decoder \(\phi_n(M, l)\) reconstructs a sequence from \(\lfloor e^{nR} \rfloor\) choices according to the forwarded index \(M = m \in \{1,2, \ldots, \lfloor e^{nR} \rfloor \}\).
Denote by \(\mathcal{L}_n(l)\) the set of the \(\lfloor e^{nR} \rfloor\) choices under \(L = l\).
From Lemma \ref{lem:connection_log_loss_list_decoding}, there is a soft estimate \(\hat{P}_l\) that \(R\)-covers \(\mathcal{L}_n(l)\), e.g., we can choose \(\hat{P}_l\) to be uniformly distributed over \(\mathcal{L}_n(l)\).
Therefore, we now can obtain \(e^{nB}\) soft estimates \(\hat{P}_l\) for \(l \in [e^{nB}]\).
It is clear that these soft estimates together with the encoder \( \varphi_n(Y^n)\) form a code for the IB source under rate-distortion pair \((B, R)\).

As for the other direction, consider a code for the IB source under rate-distortion pair \((B, R)\).
Denote by  \( \varphi_n(Y^n)\) the encoder of the code, and \(\hat{P}_l\) for \(l \in [e^{nB}]\) the soft estimates.
Let \(\mathcal{L}_n(l)\) be the set \(R\)-covered by \(\hat{P}_l\) and hence \(|\mathcal{L}_n(l)| \leq \lfloor e^{nR} \rfloor\).
Now, we return our attention to source coding with a helper at both sides; notice that we can construct an encoder-decoder pair \((\tilde{f}_n, \phi_n)\) by choosing a one-to-one mapping between \(\mathcal{L}_n(l) \) and \(\lfloor e^{nR} \rfloor\) for every \(l \in [e^{nB}]\).
\((\tilde{f}_n, \phi_n)\) together with \(\varphi_n\) clearly form an \((n, R, B)\) code for the latter problem.

\subsection{Error Probability}
Recall that \(G: \mathcal{X}^n \to \{1, 2, \ldots, |\mathcal{X}^n|\}\) is the probability rank function that ranks sequences in decreasing order of their probabilities.
For source coding with a helper at both sides, consider the joint distribution \(P_{X^nL}\) induced by \(P_{XY}^n\) and the helper's encoder \( \varphi_n(Y^n)\).
Then, we see that
\begin{align}
    \P \{ X^n \neq \hat{X}^n \} & = \sum_{l = 1}^{2^{nB}} P_L(l) \P \{ X^n \neq \hat{X}^n | L = l \} \\
    & = \sum_{l = 1}^{2^{nB}} P_L(l) \P \{ X^n \neq \phi_n(M, l) | L = l \} \\
    & = \sum_{l = 1}^{2^{nB}} P_L(l) \P \{ X^n \neq \phi_n( \tilde{f}_n(X^n, l) , l) | L = l \} \\
    & \geq \sum_{l = 1}^{2^{nB}} P_L(l) \P \{  G(X^n) > \lfloor e^{nR} \rfloor | L = l \}, \label{chp:connection_WAK_IB_source/sourc_coding_helper_double_sides}
\end{align}
where  in \eqref{chp:connection_WAK_IB_source/sourc_coding_helper_double_sides}, \(G(x^n)\) is the probability rank function under \(P_{X^n|L = l}\).
\eqref{chp:connection_WAK_IB_source/sourc_coding_helper_double_sides} holds because the noiseless link between the transmitter and receiver has rate-limit \(R\); thus the coding scheme minimizing the decoding error probability is the one where both the transmitter \(\tilde{f}_n(X^n, l)\) and receiver \(\phi_n(M, l)\) perform a one-to-one mapping over the top \( \lfloor e^{nR} \rfloor\) sequences with the highest probabilities under \(P_{X^n | L = l}\).

From the converse proof of Theorem \ref{thm:exponent_remote_log_loss} in Section \ref{sec:chp_remote_log_loss/converse/lower_bound_excess_distortion_probability}, we see that \eqref{chp:connection_WAK_IB_source/sourc_coding_helper_double_sides} is the same as the minimum excess-distortion probability of the IB source under rate-distortion pair \((B, R)\), since the optimal soft estimate is also the one supported on the top \( \lfloor e^{nR} \rfloor\) sequences with the highest probabilities under \(P_{X^n | L = l}\).
In conclusion, source coding with a helper at both sides is equivalent to the IB source, as each code for one problem can be transformed into a code for the other, leading to the same error probability for both problems.

\section{Concluding Remarks}
\label{sec:conclusions}

We studied the excess distortion probability of the IB source, and established the exact error exponent and strong converse exponent by deriving matching upper and lower exponential bound.
The achievable schemes for both cases were designed using the type covering lemma, with the excess distortion probabilities analyzed by investigating the intersection of conditional type classes.
In particular, for the strong converse exponent, we handled the rate-limit \(R\) by choosing the best \(e^{nR}\) codewords from the covering set.
When deriving the converse bound for the error exponent, we further developed the conventional sphere packing approach by introducing a conditional mutual information term, also known as the soft Markov chain.
In establishing the converse bound for the strong converse exponent, we extended a known single-letterization technique to multiterminal settings by introducing an auxiliary random variable.

We next established a code-level connection between the IB source and the WAK problem.
We showed that the IB source under the excess-distortion criterion is equivalent to source coding with a helper at both sides, which enables us to uncover the link to the WAK problem.
From this, we re-drived the best known sphere packing exponent of the WAK problem and provided it with an operational interpretation.

We now conclude by discussing potential future work.
The CEO problem \cite{bergerCEOProblem1996} is a multiterminal variant of remote lossy source coding, where the receiver now has access to rate-limited descriptions from multiple helpers, each observing a noisy version of the source.
As the exact error exponent and strong converse exponent for the IB source have been derived in this paper, it follows naturally to study the corresponding exponents for its multiterminal variant, the CEO problem under logarithmic loss.
The rate region of the latter problem has been fully characterized by Courtade and Weissman in \cite[Section III]{courtadeMultiterminalSourceCoding2014}, suggesting that tight exponential error bounds may also be derivable.
Thus, it is worth exploring the parallel exponents for the CEO problem under logarithmic loss.
Besides the CEO problem, it is also of interest to apply the converse techniques developed in this paper to other multiterminal coding problems, and seek to derive improved exponential error bounds for them.

\appendices

\section{Proofs of Propositions and Corollaries}
\label{sec:proofs_prop_cor}

We prove the propositions and corollaries in this appendix.

\subsection{Proof of Proposition \ref{prop:remote_log_loss_exponent_rate_distortion_function}}
\label{sec:chp_remote_log_loss/exponent_rate_distortion_function}
We show that \(E(R, \Delta) = 0\) if and only if \(R \leq R(\Delta)\).
For fixed \(Q_Y\), define
\begin{equation}
    E(Q_Y, R, \Delta) \triangleq \max_{\substack{  Q_{U|Y}:  \\ I(Q_Y, Q_{U|Y}) \leq R }   }  \min_{ \substack{  Q_{X|YU} : \\ H_Q(X|U) \geq \Delta }   } D(Q_{XYU} \| P_{XY}Q_{U|Y}).
\end{equation}
Hence, we have
\begin{equation}
    E(R, \Delta) = \min_{Q_Y} E(Q_Y, R, \Delta).
\end{equation}
For a fixed \(\Delta\), it is clear that \(E(R, \Delta) = 0\) if and only if \(E(P_Y, R, \Delta) = 0\), since \(E(Q_Y, R, \Delta) > 0\) if \(Q_Y \neq P_{Y}\).
Due to the maximization over \(Q_{U|Y}\), \(E(P_Y, R, \Delta) = 0\) is true if and only if for every \(Q_{U|Y}\) such that \(I(P_Y, Q_{U|Y}) \leq R\), we can choose \(Q_{XYU} = P_{XY}Q_{U|Y}\), i.e., for every \(Q_{U|Y}\) such that \(I(P_Y, Q_{U|Y}) \leq R\), we have \(H(X|U) \geq \Delta\) under the joint distribution \(P_{XYU} = P_{XY}Q_{U|Y}\).
This happens when
\begin{equation}
    R \leq  \min_{Q_{U|Y}: H(X|U) \leq \Delta} I(P_Y, Q_{U|Y}) = R(\Delta),
\end{equation}
which completes the proof.

\subsection{Proof of Corollary \ref{cor:exponent_lossy_log_loss}}
\label{sec:chp_remote_log_loss/lossy_speical_case}
Recall the method of Lagrangian multipliers, from which we have
\begin{align}
    \min_{x: g(x) \leq t} f(x) & = \min_{x} \max_{\rho \geq 0}f(x) + \rho(g(x) - t), \label{eq:chp_remote_log_loss/lagrangian_multipliers_min}
\end{align}
if \(f(x) > - \infty\).
For example, \eqref{eq:chp_remote_log_loss/lagrangian_multipliers_min} can be verified by noticing \( \max_{\rho \geq 0} f(x) + \rho(g(x) - t) = \infty\) for \(x\) such that \(g(x) > t\).
First, we write
\begin{align}
    & E(R, \Delta) \nonumber \\
    & = \min_{Q_Y} \! \! \! \max_{\substack{  Q_{U|Y}:  \\ I(Q_Y, Q_{U|Y}) \leq R }   }  \min_{ \substack{  Q_{X|YU} : \\ H_Q(X|U) \geq \Delta }   } D(Q_{XYU} \| P_{XY}Q_{U|Y}) \\
    & = \min_{Q_Y} \! \! \! \max_{\substack{  Q_{U|Y}:  \\ I(Q_Y, Q_{U|Y}) \leq R }   }  \min_{ \substack{  Q_{X|YU} : \\ H_Q(X|U) \geq \Delta }   } D(Q_{X|YU} Q_Y Q_{U|Y}  \| P_{X|Y} P_{Y} Q_{U|Y}) \\
    & = \min_{Q_Y} \! \! \! \max_{\substack{  Q_{U|Y}:  \\ I(Q_Y, Q_{U|Y}) \leq R }   }  \min_{ \substack{  Q_{X|YU} : \\ H_Q(X|U) \geq \Delta }   } D(Q_Y \| P_{Y}) + D(Q_{X|YU}Q_Y Q_{U|Y}  \| P_{X|Y} Q_{Y} Q_{U|Y}) \nonumber \\
    & = \min_{Q_Y} D(Q_Y \| P_{Y}) + f(Q_Y),
\end{align}
where
\begin{equation}
    f(Q_Y) \triangleq \hspace{-10pt} \max_{\substack{  Q_{U|Y}:  \\ I(Q_Y, Q_{U|Y}) \leq R }   }  \min_{ \substack{  Q_{X|YU} : \\ H_Q(X|U) \geq \Delta }   } \hspace{-5pt} D(Q_{X|YU} Q_Y Q_{U|Y}  \| P_{X|Y} Q_{Y} Q_{U|Y}).
\end{equation}
Since \(\P\{X=Y\} = 1\), i.e., \(P_{X|Y}\) is an identity mapping, we see that
\begin{align}
    & f(Q_Y) \nonumber \\
    & = \hspace{-14pt} \max_{\substack{  Q_{U|Y}:  \\ I(Q_Y, Q_{U|Y}) \leq R }   } \hspace{-10pt} \min_{ \substack{  Q_{X|YU} }   } \max_{\rho \geq 0} D(Q_{X|YU} Q_Y Q_{U|Y} \| P_{X|Y} Q_{Y} Q_{U|Y}) + \rho( \Delta - H_Q(X|U) ) \label{eq:chp_remote_log_loss/lossy_speical_case/lagrangian}\\
    & = \hspace{-14pt} \max_{\substack{  Q_{U|Y}:  \\ I(Q_Y, Q_{U|Y}) \leq R }   } \hspace{-10pt} \max_{\rho \geq 0} \min_{ \substack{  Q_{X|YU} }   }  D(Q_{X|YU} Q_Y Q_{U|Y} \| P_{X|Y} Q_{Y} Q_{U|Y}) + \rho( \Delta - H_Q(X|U) ) \label{eq:chp_remote_log_loss/lossy_speical_case/minimax} \\
    & = \max_{\substack{  Q_{U|Y}:  \\ I(Q_Y, Q_{U|Y}) \leq R }   } \max_{\rho \geq 0}  \rho( \Delta - H_Q(Y|U) ) \label{eq:chp_remote_log_loss/lossy_speical_case/identity_mapping} \\
    & = \max_{\substack{  Q_{U|Y}:  \\ I(Q_Y, Q_{U|Y}) \leq R }   } \max_{\rho \geq 0}  \rho( \Delta + I(Q_Y, Q_{U|Y}) - H(Q_Y) ) \\
    & = \max_{\rho \geq 0} \max_{\substack{  Q_{U|Y}:  \\ I(Q_Y, Q_{U|Y}) \leq R }   }  \rho( \Delta + I(Q_Y, Q_{U|Y}) - H(Q_Y) ) \\
    & = \max_{\rho \geq 0} \rho ( \Delta + R - H(Q_Y) ).
\end{align}
where \eqref{eq:chp_remote_log_loss/lossy_speical_case/lagrangian} follows from \eqref{eq:chp_remote_log_loss/lagrangian_multipliers_min};
\eqref{eq:chp_remote_log_loss/lossy_speical_case/minimax} results from the minimax theorem, as it can be verified that both the divergence term and \(-H_Q(X|U)\) are convex in \(Q_{X|YU}\);
\eqref{eq:chp_remote_log_loss/lossy_speical_case/identity_mapping} is because the minimization over \(Q_{XYU}\) in \eqref{eq:chp_remote_log_loss/lossy_speical_case/minimax} is achieved if and only if \(Q_{X|YU} = P_{X|Y}\) due to the identity mapping (otherwise the value of the divergence will be infinity); and we have \(Q_{YU} = Q_Y Q_{U|Y}\) in \eqref{eq:chp_remote_log_loss/lossy_speical_case/identity_mapping}.
Hence, if \(\P\{X=Y\} = 1\), then we can assert that
\begin{align}
    E(R, \Delta) & = \min_{Q_Y} D(Q_Y \| P_{Y}) + f(Q_Y) \\
    & = \min_{ Q_Y } D(Q_Y \| P_{Y}) + \max_{\rho \geq 0} \rho ( \Delta + R - H(Q_Y) ) \\
    & = \min_{Q_Y: H(Q_Y) \geq R + \Delta} D(Q_Y \| P_Y),
\end{align}
which completes the proof.

\subsection{Proof of Proposition \ref{prop:chp_strong_converse_remote_log_loss_rate_regime}}
\label{sec:chp_strong_converse_remote_log_loss/rate_regime}
The proof is reminiscent of a similar proof in \cite{merhavGeneralizedStochasticLikelihood2017}.
Consider the identity that for every constant \(a\), we have \(|a|^{+} = \max_{\rho \in [0,1]} \rho a\).
From this, we can write
\begin{align}
    F(R, \Delta) & = \min_{ \substack{ Q_{XYU} : \\ H_{Q}(X|U) \leq \Delta } } D(Q_{XYU} \| P_{XY}Q_{U|Y}) + | I_{Q}(Y;U) - R |^{+} \\
    & = \min_{ \substack{ Q_{XYU} : \\ H_{Q}(X|U) \leq \Delta } } \max_{\rho \in [0,1] } D(Q_{XYU} \| P_{XY}Q_{U|Y}) + \rho (I_{Q}(Y;U) - R).
\end{align}
Hence, \(F(R, \Delta) > 0\) means that for every \(Q_{XYU}\) satisfying \(H_Q(X|U) \leq \Delta\), there exists a \(\rho \in [0,1]\) such that
\begin{equation}
    D(Q_{XYU} \| P_{XY}Q_{U|Y}) + \rho (I_{Q}(Y;U) - R) > 0,
\end{equation}
i.e.,
\begin{equation}
    R < \frac{1}{\rho} D(Q_{XYU} \| P_{XY}Q_{U|Y}) + I_Q(Y;U). \label{eq:chp_stong_converse_remote_log_loss/rate_regime/R_bound_1}
\end{equation}
Therefore, \(F(R, \Delta) > 0\) is equivalent to
\begin{align}
    R & < \min_{ \substack{ Q_{XYU} : \\ H_{Q}(X|U) \leq \Delta } } \max_{\rho \in [0,1] }  \frac{1}{\rho} D(Q_{XYU} \| P_{XY}Q_{U|Y}) + I_Q(Y;U) \label{eq:chp_stong_converse_remote_log_loss/rate_regime/R_bound_2} \\
    & = \min_{ \substack{ P_{U|Y} : \\ H(X|U) \leq \Delta } } I(Y;U) \label{eq:chp_stong_converse_remote_log_loss/rate_regime/minimization_solution} \\
    & = R(\Delta),
\end{align}
where in \eqref{eq:chp_stong_converse_remote_log_loss/rate_regime/R_bound_2} the minimization over \(Q_{XYU}\) is because \eqref{eq:chp_stong_converse_remote_log_loss/rate_regime/R_bound_1} holds for every \(Q_{XYU}\) and the maximization over \(\rho\) is due to the existence of such \(\rho\); in \eqref{eq:chp_stong_converse_remote_log_loss/rate_regime/minimization_solution} we notice that the minimum in \eqref{eq:chp_stong_converse_remote_log_loss/rate_regime/R_bound_2} is achieved if and only if \(D(Q_{XYU} \| P_{XY}Q_{U|Y}) = 0\) due to the maximization over \(\rho \in [0,1]\) (in particular \(\rho = 0\)), which gives rise to \(Q_{XY} = P_{XY}\) and the Markov chain \(X \to Y \to U\).

\subsection{Proof of Corollary \ref{cor:strong_converse_exponent_lossy_log_loss}}
\label{sec:chp_strong_converse_remote_log_loss/lossy_speical_case}
Recall \eqref{eq:chp_remote_log_loss/lagrangian_multipliers_min}, i.e.,
\begin{align}
    \min_{x: g(x) \leq t} f(x) & = \min_{x} \max_{\rho \geq 0}f(x) + \rho(g(x) - t), \label{eq:chp_strong_converse_remote_log_loss/lagrangian_multipliers_min}
\end{align}
if \(f(x) > - \infty\).
First, we write
\begin{align}
    & F(R, \Delta) \nonumber \\
    & = \min_{ \substack{ Q_{XYU} : \\ H_{Q}(X|U) \leq \Delta } } D(Q_{XYU} \| P_{XY}Q_{U|Y}) + | I_{Q}(Y;U) - R |^{+}\\
    & = \min_{ \substack{ Q_{XYU} : \\ H_{Q}(X|U) \leq \Delta } } D(Q_{X|YU} Q_Y Q_{U|Y}  \| P_{X|Y} P_{Y} Q_{U|Y})+ | I_{Q}(Y;U) - R |^{+} \\
    & = \min_{ \substack{ Q_{XYU} : \\ H_{Q}(X|U) \leq \Delta } } D(Q_Y \| P_{Y}) + D(Q_{X|YU}Q_Y Q_{U|Y}  \| P_{X|Y} Q_{Y} Q_{U|Y}) + | I_{Q}(Y;U) - R |^{+}\nonumber \\
    & = \min_{Q_Y, Q_{U|Y}} D(Q_Y \| P_{Y}) + | I_{Q}(Y;U) - R |^{+} + f(Q_Y, Q_{U|Y}),
\end{align}
where
\begin{equation}
    f(Q_Y, Q_{U|Y}) \triangleq   \min_{ \substack{  Q_{X|YU} : \\ H_Q(X|U) \leq \Delta }   } \hspace{-5pt} D(Q_{X|YU} Q_Y Q_{U|Y}  \| P_{X|Y} Q_{Y} Q_{U|Y}).
\end{equation}
Since \(\P\{X=Y\} = 1\), i.e., \(P_{X|Y}\) is an identity mapping, we see that
\begin{align}
    & f(Q_Y, Q_{U|Y}) \nonumber \\
    & =   \min_{ \substack{  Q_{X|YU} }   } \max_{\rho \geq 0} D(Q_{X|YU} Q_Y Q_{U|Y} \| P_{X|Y} Q_{Y} Q_{U|Y}) + \rho( H_Q(X|U) - \Delta ) \label{eq:chp_strong_converse_remote_log_loss/lossy_speical_case/lagrangian}\\
    & =   \max_{\rho \geq 0} \min_{ \substack{  Q_{X|YU} }   }  D(Q_{X|YU} Q_Y Q_{U|Y} \| P_{X|Y} Q_{Y} Q_{U|Y}) + \rho( H_Q(X|U) - \Delta ) \label{eq:chp_strong_converse_remote_log_loss/lossy_speical_case/minimax} \\
    & = \max_{\rho \geq 0}  \rho( H_Q(Y|U) - \Delta ), \label{eq:chp_strong_converse_remote_log_loss/lossy_speical_case/identity_mapping}
\end{align}
where \eqref{eq:chp_strong_converse_remote_log_loss/lossy_speical_case/lagrangian} follows from \eqref{eq:chp_strong_converse_remote_log_loss/lagrangian_multipliers_min};
\eqref{eq:chp_strong_converse_remote_log_loss/lossy_speical_case/minimax} results from the minimax theorem;
\eqref{eq:chp_strong_converse_remote_log_loss/lossy_speical_case/identity_mapping} is because the minimization over \(Q_{XYU}\) in \eqref{eq:chp_strong_converse_remote_log_loss/lossy_speical_case/minimax} is achieved if and only if \(Q_{X|YU} = P_{X|Y}\) due to the identity mapping (otherwise the value of the divergence will be infinity); and we have \(Q_{YU} = Q_Y Q_{U|Y}\) in \eqref{eq:chp_strong_converse_remote_log_loss/lossy_speical_case/identity_mapping}.
Hence, if \(\P\{X=Y\} = 1\), then we can assert that
\begin{align}
    F(R, \Delta) & = \min_{Q_Y, Q_{U|Y}} D(Q_Y \| P_{Y}) + | I_{Q}(Y;U) - R |^{+} + f(Q_Y, Q_{U|Y}) \\
    & = \min_{Q_Y, Q_{U|Y}} D(Q_Y \| P_{Y}) + | I_{Q}(Y;U) - R |^{+} + \max_{\rho \geq 0} \rho ( H_Q(Y|U) - \Delta ) \\
    & = \min_{ \substack{ Q_Y , Q_{U|Y}: \\ H_Q(U|Y) \leq \Delta } } D(Q_Y \| P_Y) + | I_{Q}(Y;U) - R |^{+} \\
    & = \min_{Q_Y} D(Q_Y \| P_Y) +  \min_{ \substack{  Q_{U|Y}: \\ H_Q(Y|U) \leq \Delta } } | I_{Q}(Y;U) - R |^{+} \\
    & = \min_{Q_Y} D(Q_Y \| P_Y) +  \min_{ \substack{  Q_{U|Y}: \\ H_Q(Y|U) \leq \Delta } } | H_Q(Y) - H_Q(Y|U) - R |^{+} \\
    & = \min_{Q_Y} D(Q_Y \| P_Y) + | H_Q(Y)  - R - \Delta |^{+},
\end{align}
which completes the proof.

\section{Proof of Lemma \ref{lem:probability_intersection_conditional_type_class}}
\label{apd:proof_probability_intersection_conditional_type_class}
Notice that
\begin{equation}
    \mathcal{T}_n(Q_{Y|U} | \bm{u}) = \bigcup_{Q_{Y|X}} \mathcal{T}_n(Q_{Y|X} | \bm{x}) \cap \mathcal{T}_n(Q_{Y|U} | \bm{u}).
\end{equation}
Thus, it follows that
\begin{align}
    P_{Y|X}^n[  \mathcal{T}_n(Q_{Y|U} | \bm{u})   | \bm{x} ] & \leq \sum_{Q_{Y|X}} P_{Y|X}^n[ \mathcal{T}_n(Q_{Y|X} | \bm{x}) \cap \mathcal{T}_n(Q_{Y|U} | \bm{u})   | \bm{x} ].
\end{align}
Consider a fixed conditional type \(Q_{Y|X}\).
For every \(\bm{y} \in \mathcal{T}_n(Q_{Y|X} | \bm{x})\), we have
\begin{equation}
    P_{Y|X}^n(\bm{y} | \bm{x}) = \exp\{-n( D(Q_{Y|X} \| P_{Y|X} | \hat{P}_{\bm{x}}) + H(Q_{Y|X} | \hat{P}_{\bm{x}}) )\}.
\end{equation}
Then, from Lemma \ref{lem:intersection_conditional_type_class}, we see that
\begin{align}
    & P_{Y|X}^n[ \mathcal{T}_n(Q_{Y|X} | \bm{x}) \cap \mathcal{T}_n(Q_{Y|U} | \bm{u})   | \bm{x} ] \nonumber \\
    & = \exp\{-n( D(Q_{Y|X} \| P_{Y|X} | \hat{P}_{\bm{x}}) + H(Q_{Y|X} | \hat{P}_{\bm{x}}) )\} \times | \mathcal{T}_n(Q_{Y|X} | \bm{x}) \cap \mathcal{T}_n(Q_{Y|U} | \bm{u})  | \nonumber \\
    & \ndot{=} \max_{Q_{XYU}} \exp\{-n( D(Q_{Y|X} \| P_{Y|X} | \hat{P}_{\bm{x}}) + I_{Q}(Y;U|X) )\},
\end{align}
where the maximization is over all \(Q_{XYU}\) satisfying the three marginal distribution conditions
\begin{equation*}
    Q_{XY} = \hat{P}_{\bm{x}} \times Q_{Y|X}, \ Q_{YU} = \hat{P}_{\bm{u}} \times Q_{Y|U},\ \text{and} \ Q_{XU} = \hat{P}_{\bm{x} \bm{u}}.
\end{equation*}
Hence, we obtain
\begin{align}
    & P_{Y|X}^n[  \mathcal{T}_n(Q_{Y|U} | \bm{u})   | \bm{x} ] \nonumber \\
    & \leq \sum_{Q_{Y|X}} P_{Y|X}^n[ \mathcal{T}_n(Q_{Y|X} | \bm{x}) \cap \mathcal{T}_n(Q_{Y|U} | \bm{u})   | \bm{x} ] \\
    & \ndot{=} \max_{Q_{Y|X}} \max_{Q_{XYU}} \exp\{-n( D(Q_{Y|X} \| P_{Y|X} | \hat{P}_{\bm{x}}) + I_{Q}(Y;U|X) )\} \\
    & = \max_{Q_{XYU}} \exp\{-n( D(Q_{Y|X} \| P_{Y|X} | \hat{P}_{\bm{x}}) + I_{Q}(Y;U|X) )\}, \label{eq:chp_preliminaries/combination_maximizations}
\end{align}
where in \eqref{eq:chp_preliminaries/combination_maximizations} we combine the two maximizations, i.e., now we consider all \(Q_{XYU}\) satisfying the two marginal distribution conditions
\begin{equation*}
    Q_{YU} = \hat{P}_{\bm{u}} \times Q_{Y|U} \quad \text{and} \quad Q_{XU} = \hat{P}_{\bm{x} \bm{u}}.
\end{equation*}
In the same fashion, we also have
\begin{align}
    & P_{Y|X}^n[  \mathcal{T}_n(Q_{Y|U} | \bm{u})   | \bm{x} ] \nonumber \\
    & \geq \max_{Q_{Y|X}} P_{Y|X}^n[ \mathcal{T}_n(Q_{Y|X} | \bm{x}) \cap \mathcal{T}_n(Q_{Y|U} | \bm{u})   | \bm{x} ] \\
    & \ndot{=} \max_{Q_{XYU}} \exp\{-n( D(Q_{Y|X} \| P_{Y|X} | \hat{P}_{\bm{x}}) + I_{Q}(Y;U|X) )\},
\end{align}
which completes the proof.

\begin{remark}
\label{rem:apd_proof_intersection_conditional_type_class}
There is a different proof for Lemma \ref{lem:probability_intersection_conditional_type_class}.
For a fixed \(\bm{x}, \bm{u}, \) and \(Q_{Y|U}\), we have
\begin{equation}
    \mathcal{T}_n(Q_{Y|U} | \bm{u} ) = \bigcup_{Q_{Y|XU}} \mathcal{T}_n(Q_{Y|XU} | \bm{x}, \bm{u}),
\end{equation}
where we consider all conditional type \(Q_{Y|XU}\) such that the joint distribution \(Q_{XYU} = \hat{P}_{\bm{x} \bm{u}}Q_{Y|XU}\) satisfies the marginal distribution condition \(Q_{YU} = \hat{P}_{\bm{u}} \times Q_{Y|U}\).
It follows that
\begin{align}
    P_{Y|X}^n[ \mathcal{T}_n(Q_{Y|U} | \bm{u}) | \bm{x} ] & \ndot{=} \max_{Q_{Y|XU}} \exp\{-nD(Q_{Y|XU} \| P_{Y|X} | \hat{P}_{\bm{x} \bm{u}})\} \label{eq:apd_probability_intersection_conditional_type_class_alternative_proof} \\
    & = \max_{Q_{Y|XU}} \exp\{-n( D(Q_{Y|X} \| P_{Y|X} | \hat{P}_{\bm{x}}) + I_{Q}(Y;U|X) )\}.
\end{align}
where in \eqref{eq:apd_probability_intersection_conditional_type_class_alternative_proof} we can think of \(P_{Y|X}\) as a channel \(\tilde{P}_{Y|XU} \triangleq P_{Y|X}\).
While this second proof is perhaps more concise, we prefer the first one as it better highlights the meaning of the term \(I_{Q}(Y;U|X)\) in the exponent.
\end{remark}

\section{Proof of Lemma \ref{lem:reverse_markov_inequality} and Lemma \ref{lem:reverse_wyner_inequality}}
\label{apd:proofs_reverse_markov_wyner_inequality}

\subsection{Proof of Lemma \ref{lem:reverse_markov_inequality}}
Consider the random variable \(\tilde{X} \triangleq d- X\).
From Markov's inequality, we have
\begin{equation}
    \P\{X \leq a\} = \P \{ \tilde{X} \geq d- a \} \leq \frac{\E[\tilde{X}]}{d-a} = \frac{d- \E [X]}{d - a}.
\end{equation}
It follows that
\begin{equation}
    \P\{X > a\} = 1- \P\{X \leq a\} \geq \frac{ \E[X] - a }{ d - a}.
\end{equation}


\subsection{Proof of Lemma \ref{lem:reverse_wyner_inequality}}
Notice that
\begin{align}
    H(X) - \E[\log f(X)] & = \sum_{x \in \mathcal{X}} P_X(x) \log \frac{1}{P_X(x) f(x)} \\
    & \leq \log \sum_{x \in \mathcal{X}} \frac{1}{f(x)} \\
    & \leq \log (1 + \log |\mathcal{X}|),
\end{align}
where the last step follows from \(\sum_{i=1}^{n} 1/i \leq 1 + \log n\).
This short proof follows \cite{dunhamOptimalNoiselessCoding1980}.
For the special case where \(f(x)\) is the probability rank function, i.e., \(f(x) = G(x)\),  Wyner \cite{wynerUpperBoundEntropy1972} showed that \(\E[\log G(X)] \leq H(X)\).
Rényi entropy extensions of the two inequalities appeared in, e.g., \cite[Theorem 1]{courtadeCumulantGeneratingFunction2014}.

\section{Proofs of Cardinality Bounds}

\subsection{Proof of the Cardinality Bound in Theorem \ref{thm:exponent_remote_log_loss}}
\label{apd:cardinality_bound_for_error_exponent_remote_log_loss}
For any auxiliary alphabet \(\mathcal{U}\) with \(|\mathcal{U}| > 2 |\mathcal{X}| |\mathcal{Y}| + 2\),  let
\begin{equation}
    E(Q_Y, R, \Delta) \triangleq \max_{\substack{  Q_{U|Y}:  \\ I(Q_Y, Q_{U|Y}) \leq R }   }  \min_{ \substack{  Q_{X|YU} : \\ H_Q(X|U) \geq \Delta }   } D(Q_{XYU} \| P_{XY}Q_{U|Y}). \label{eq:apd_cardinality_bound_E_IB_source_error_exponent}
\end{equation}
For an auxiliary alphabet \(\mathcal{U}^{\prime}\) with \(|\mathcal{U}^{\prime}| \leq 2 |\mathcal{X}| |\mathcal{Y}| + 2\), let
\begin{equation}
    E^{\prime}(Q_Y, R, \Delta) \triangleq \max_{\substack{  Q_{U^{\prime}|Y}:  \\ I(Q_Y, Q_{U^{\prime}|Y}) \leq R }   }  \min_{ \substack{  Q_{X|YU^{\prime}} : \\ H_Q(X|U^{\prime}) \geq \Delta }   } D(Q_{XYU^{\prime}} \| P_{XY}Q_{U^{\prime}|Y}).
\end{equation}
The task here is to show that
\begin{equation}
    \min_{Q_Y} E(Q_Y, R, \Delta)= \min_{Q_Y} E^{\prime}(Q_Y, R, \Delta).
\end{equation}
We accomplish this by showing that \(E(Q_Y, R, \Delta) = E^{\prime}(Q_Y, R, \Delta)\) holds for every \(Q_Y\).
It is obvious that \(E(Q_Y, R, \Delta) \geq E^{\prime}(Q_Y, R, \Delta)\), due to \(|\mathcal{U}| > |\mathcal{U}^{\prime}|\).
Hence, what remains is to show that \(E(Q_Y, R, \Delta) \leq E^{\prime}(Q_Y, R, \Delta)\).

To this end, assume \(Q_{U|Y}^{\ast}\) is a solution to the RHS of \eqref{eq:apd_cardinality_bound_E_IB_source_error_exponent}, i.e., \(I(Q_Y, Q_{U|Y}^{\ast}) \leq R\) and
\begin{equation}
    E(Q_Y, R, \Delta) =  \min_{ \substack{  Q_{X|YU} : \\ H_Q(X|U) \geq \Delta }   } D( Q_Y Q_{U|Y}^{\ast} Q_{X|YU} \| P_{XY}Q_{U|Y}^{\ast}).
\end{equation}
Let \(Q_{U}^{\ast} = Q_Y \cdot Q_{U|Y}^{\ast}\) and denote by \(Q_{Y|U}^{\ast}\) the reverse channel induced by \(Q_Y\) and \(Q_{U|Y}^{\ast}\).
Hence, we have
\begin{equation}
    E(Q_Y, R, \Delta) =  \min_{ \substack{  Q_{X|YU} : \\ H_Q(X|U^{\ast}) \geq \Delta }   } D( Q_U^{\ast}Q_{Y|U}^{\ast} Q_{X|YU} \| Q_U^{\ast}Q_{Y|U}^{\ast} P_{X|Y}). \label{eq:apd_cardinality_bound_remote_log_loss}
\end{equation}
Since the object function \(D( Q_U^{\ast}Q_{Y|U}^{\ast} Q_{X|YU} \| Q_U^{\ast}Q_{Y|U}^{\ast} P_{X|Y})\) is strictly convex in \(Q_{X|YU}\) and  \(\{ Q_{X|YU} : H_Q(X|U^{\ast}) \geq \Delta\}\) is also a convex set, we can use Lagrangian under the KKT conditions to solve the right hand side of \eqref{eq:apd_cardinality_bound_remote_log_loss}.
Denote by \(\mathcal{L}(Q_U^{\ast} Q_{Y|U}^{\ast} Q_{X|YU})\) the corresponding Lagrangian  dual function, and define
\begin{equation}
    \varphi(Q_U^{\ast} Q_{Y|U}^{\ast} Q_{X|YU} ) \triangleq \pdv{\mathcal{L}(Q_U^{\ast} Q_{Y|U}^{\ast} Q_{X|YU})}{Q_{X|YU}}.
\end{equation}
Let \(Q_{X|YU}^{\ast}\) be the solution to the convex optimization problem, i.e.,
\begin{equation}
    \varphi(Q_U^{\ast} Q_{Y|U}^{\ast} Q_{X|YU}^{\ast} ) = 0.
\end{equation}
Since both the objection function \(D( Q_U^{\ast}Q_{Y|U}^{\ast} Q_{X|YU} \| Q_U^{\ast}Q_{Y|U}^{\ast} P_{X|Y})\) and the constraint \(H_Q(X|U^{\ast})\) are linear functions of \(Q_{U}^{\ast}\), we see that \(\varphi(Q_U^{\ast}Q_{Y|U}^{\ast} Q_{X|YU} )\) is also a linear function of \(Q_{U}^{\ast}\).
In other words, there exits some function \(\tilde{\varphi}( Q_{XY|U} )\) such that
\begin{equation}
    \varphi(Q_U^{\ast} Q_{Y|U}^{\ast} Q_{X|YU}^{\ast} ) = \sum_{u}Q_{U}^{\ast}(u) \tilde{\varphi}( Q_{XY|U}^{\ast}(\cdot | u) ) = 0,
\end{equation}
where \(Q_{XY|U}^{\ast} = Q_{Y|U}^{\ast} Q_{X|YU}^{\ast}\).

As \(Q_{U|Y}^{\ast}\) and \( Q_{X|YU}^{\ast}\) are the solutions to the optimizations, we have
\begin{align}
    E(Q_Y, R, \Delta) & = D( Q_Y Q_{U|Y}^{\ast} Q_{X|YU}^{\ast} \| P_{XY}Q_{U|Y}^{\ast}) \\
    & = D(Q_{XY}^{\ast} \| P_{XY}) + I_{Q^{\ast}} (X;U|Y) \\
    & = D(Q_{XY}^{\ast} \| P_{XY}) + H_{Q^{\ast}}(X|Y) - H_{Q^{\ast}} (X|YU).
\end{align}
Now, define
\begin{equation}
    f_{xy}(Q_{XY}) = Q_{XY}(x, y) \qquad \forall x, y
\end{equation}
as well as
\begin{align}
    f_1( Q_{XY} ) & = H_Q(X|Y), \\
    f_2 ( Q_{XY} ) & = H(Q_X), \\
    f_3( Q_{XY} ) & = H(Q_Y).
\end{align}
It follows that
\begin{align}
    \sum_{u}Q_{U}^{\ast}(u) \tilde{\varphi}( Q_{XY|U}^{\ast}(\cdot | u) ) & = 0, \label{eq:apd_cardinality_bound_remote_log_loss_Lagrangian} \\
    \sum_{u}Q_{U}^{\ast}(u) f_{xy}(Q_{XY|U}^{\ast}(\cdot | u)) & = Q_{XY}^{\ast}(x,y), \label{eq:apd_cardinality_bound_remote_log_loss_xy} \\
    \sum_{u}Q_{U}^{\ast}(u) f_{1}(Q_{XY|U}^{\ast}(\cdot | u)) & = H_{Q^{\ast}}(X|YU), \\
    \sum_{u}Q_{U}^{\ast}(u) f_{2}(Q_{XY|U}^{\ast}(\cdot | u)) & = H_{Q^{\ast}}(X|U), \\
    \sum_{u}Q_{U}^{\ast}(u) f_{3}(Q_{XY|U}^{\ast}(\cdot | u)) & = H_{Q^{\ast}}(Y|U).
\end{align}
Notice that \eqref{eq:apd_cardinality_bound_remote_log_loss_Lagrangian} includes \(|\mathcal{X}||\mathcal{Y}|\) functions and \eqref{eq:apd_cardinality_bound_remote_log_loss_xy} includes \(|\mathcal{X}||\mathcal{Y}| -1\) functions.\footnote{We can ignore one element from \(\mathcal{X} \times \mathcal{Y}\) because the sum of probabilities must be \(1\).}
From the support lemma \cite[Appendix C]{gamalNetworkInformationTheory2011}, we see that there exists a pair \((Q_{U^{\prime}}^{\ast}, Q_{XY|U^{\prime}}^{\ast})\) such that
\begin{align}
    \sum_{u^{\prime}}Q_{U^{\prime}}^{\ast}(u^{\prime}) \tilde{\varphi}( Q_{XY|U^{\prime}}^{\ast} (\cdot | u^{\prime}) ) & = 0, \label{eq:apd_cardinality_bound_remote_log_loss_Lagrangian_support_lemma} \\
    \sum_{u^{\prime}}Q_{U^{\prime}}^{\ast}(u^{\prime}) f_{xy}(Q_{XY|U^{\prime}}^{\ast}(\cdot | u^{\prime})) & = Q_{XY}^{\ast}(x,y) \\
    \sum_{u^{\prime}}Q_{U^{\prime}}^{\ast}(u^{\prime}) f_{1}(Q_{XY|U^{\prime}}^{\ast}(\cdot | u^{\prime})) & = H_{Q^{\ast}}(X|YU), \\
    \sum_{u^{\prime}}Q_{U^{\prime}}^{\ast}(u^{\prime}) f_{2}(Q_{XY|U^{\prime}}^{\ast}(\cdot | u^{\prime})) & = H_{Q^{\ast}}(X|U),\label{eq:apd_cardinality_bound_remote_log_loss_Delta_support_lemma} \\
    \sum_{u^{\prime}}Q_{U^{\prime}}^{\ast}(u^{\prime}) f_{3}(Q_{XY|U^{\prime}}^{\ast}(\cdot | u^{\prime})) & = H_{Q^{\ast}}(Y|U). \label{eq:apd_cardinality_bound_remote_log_loss_Y_U_support_lemma}
\end{align}
Let \(Q_{XYU^{\prime}}^{\ast} = Q_{U^{\prime}}^{\ast}Q_{XY|U^{\prime}}^{\ast} = Q_Y Q_{U^{\prime} | Y}^{\ast} Q_{X | Y U^{\prime}}^{\ast}\).
By \eqref{eq:apd_cardinality_bound_remote_log_loss_Lagrangian_support_lemma} and \eqref{eq:apd_cardinality_bound_remote_log_loss_Delta_support_lemma}, we see that \(Q_{U^{\prime} | Y}^{\ast}\) and \(Q_{X|YU^{\prime}}^{\ast}\) must satisfy
\begin{equation}
    Q_{X|YU^{\prime}}^{\ast} = \argmin_{ \substack{  Q_{X|YU^{\prime}} : \\ H_Q(X|U^{\prime}) \geq \Delta }   } D(Q_Y Q_{U^{\prime} | Y}^{\ast} Q_{X|YU^{\prime}} \| P_{XY}Q_{U^{\prime}|Y}^{\ast}).
\end{equation}
Since we consider the maximization over all \(Q_{U^{\prime} | Y}\) satisfying \(I(Q_Y, Q_{U^{\prime} | Y}) \leq R\) in \(E^{\prime}(Q_Y, R, \Delta)\) while we also have \(I(Q_Y, Q_{U^{\prime} | Y}^{\ast}) \leq R\), it follows that
\begin{equation}
    E^{\prime}(Q_Y, R, \Delta) \geq D(Q_Y Q_{U^{\prime} | Y}^{\ast} Q_{X|YU^{\prime}}^{\ast} \| P_{XY}Q_{U^{\prime}|Y}^{\ast}).
\end{equation}
Because
\begin{equation}
    D(Q_Y Q_{U^{\prime} | Y}^{\ast} Q_{X|YU^{\prime}}^{\ast} \| P_{XY}Q_{U^{\prime}|Y}^{\ast}) = E(Q_Y, R, \Delta),
\end{equation}
we conclude that
\begin{equation}
    E^{\prime}(Q_Y, R, \Delta) \geq E(Q_Y, R, \Delta),
\end{equation}
which completes the proof.

\subsection{Proof of the Cardinality Bound in Theorem \ref{thm:strong_converse_exponent_remote_log_loss}}
\label{apd:cardinality_bound_for_strong_converse_exponent_remote_log_loss}
For an arbitrary auxiliary alphabet \(\mathcal{U}\) satisfying \(|\mathcal{U}| > |\mathcal{X}||\mathcal{Y}| + 2\), we write
\begin{equation}
        F(R, \Delta) = \min_{ \substack{ Q_{XYU} : \\ H_{Q}(X|U) \leq \Delta } } D(Q_{XYU} \| P_{XY}Q_{U|Y}) + | I_{Q}(Y;U) - R |^{+}. \label{eq:apd_cardinality_bound_F_IB_source_strong_converse_exponent}
\end{equation}
For \(\mathcal{U}^{\prime}\) satisfying \(|\mathcal{U}^{\prime}| \leq |\mathcal{X}||\mathcal{Y}| + 2\), we write
\begin{equation}
        F^{\prime}(R, \Delta) = \min_{ \substack{ Q_{XYU^{\prime}} : \\ H_{Q}(X|U^{\prime}) \leq \Delta } } D(Q_{XYU^{\prime}} \| P_{XY}Q_{U^{\prime}|Y}) + | I_{Q}(Y;U^{\prime}) - R |^{+}.
\end{equation}
The task is to show that \(F(R, \Delta) = F^{\prime}(R, \Delta)\).
We clearly have \(F(R, \Delta) \leq F^{\prime}(R, \Delta)\) due to \(|\mathcal{U}| > |\mathcal{U}^{\prime}|\) and the minimization.
Hence, what remains is to show \(F(R, \Delta) \geq F^{\prime}(R, \Delta)\).

Assume \(Q_{XYU}^{\ast}\) is a solution to the RHS of  \eqref{eq:apd_cardinality_bound_F_IB_source_strong_converse_exponent}, i.e., \(H_{Q^{\ast}}(X|U) \leq \Delta\) and
\begin{align}
    F(R, \Delta) & = D(Q_{XYU}^{\ast} \| P_{XY}Q_{U|Y}^{\ast}) + | I_{Q^{\ast}}(Y;U) - R |^{+} \\
    & = D(Q_{XY}^{\ast} \| P_{XY}) + I_{Q^{\ast}}(X;U|Y) + | I_{Q^{\ast}}(Y;U) - R |^{+} \\
    & = D(Q_{XY}^{\ast} \| P_{XY}) + H_{Q^{\ast}}(X|Y) - H_{Q^{\ast}}(X|YU) + | I_{Q^{\ast}}(Y;U) - R |^{+}
\end{align}
Now, define
\begin{equation}
    f_{xy}(Q_{XY}) = Q_{XY}(x, y) \qquad \forall x, y
\end{equation}
as well as
\begin{align}
    f_1( Q_{XY} ) & = H_Q(X|Y), \\
    f_2 ( Q_{XY} ) & = H(Q_X), \\
    f_3( Q_{XY} ) & = H(Q_Y).
\end{align}
It follows that
\begin{align}
    \sum_{u}Q_{U}^{\ast}(u) f_{xy}(Q_{XY|U}^{\ast}(\cdot | u)) & = Q_{XY}^{\ast}(x,y), \label{eq:apd_cardinality_bound_strong_converse_remote_log_loss_xy} \\
    \sum_{u}Q_{U}^{\ast}(u) f_{1}(Q_{XY|U}^{\ast}(\cdot | u)) & = H_{Q^{\ast}}(X|YU), \\
    \sum_{u}Q_{U}^{\ast}(u) f_{2}(Q_{XY|U}^{\ast}(\cdot | u)) & = H_{Q^{\ast}}(X|U), \\
    \sum_{u}Q_{U}^{\ast}(u) f_{3}(Q_{XY|U}^{\ast}(\cdot | u)) & = H_{Q^{\ast}}(Y|U).
\end{align}
where \eqref{eq:apd_cardinality_bound_strong_converse_remote_log_loss_xy} includes \(|\mathcal{X}||\mathcal{Y}| - 1\) functions.
From the support lemma \cite[Appendix C]{gamalNetworkInformationTheory2011}, we see that there exists a pair \((Q_{U^{\prime}}^{\ast}, Q_{XY|U^{\prime}}^{\ast})\) with \(|\mathcal{U}^{\prime}| \leq |\mathcal{X}||\mathcal{Y}| + 2\) such that
\begin{align}
    \sum_{u^{\prime}}Q_{U^{\prime}}^{\ast}(u^{\prime}) f_{xy}(Q_{XY|U^{\prime}}^{\ast}(\cdot | u^{\prime})) & = Q_{XY}^{\ast}(x,y) \\
    \sum_{u^{\prime}}Q_{U^{\prime}}^{\ast}(u^{\prime}) f_{1}(Q_{XY|U^{\prime}}^{\ast}(\cdot | u^{\prime})) & = H_{Q^{\ast}}(X|YU), \\
    \sum_{u^{\prime}}Q_{U^{\prime}}^{\ast}(u^{\prime}) f_{2}(Q_{XY|U^{\prime}}^{\ast}(\cdot | u^{\prime})) & = H_{Q^{\ast}}(X|U),\label{eq:apd_cardinality_bound_strong_converse_remote_log_loss_Delta_support_lemma} \\
    \sum_{u^{\prime}}Q_{U^{\prime}}^{\ast}(u^{\prime}) f_{3}(Q_{XY|U^{\prime}}^{\ast}(\cdot | u^{\prime})) & = H_{Q^{\ast}}(Y|U). \label{eq:apd_cardinality_bound_strong_converse_remote_log_loss_Y_U_support_lemma}
\end{align}
It is clear that under \(Q_{XYU^{\prime}}^{\ast} = Q_{U^{\prime}}^{\ast} Q_{XY|U^{\prime}}^{\ast}\), we have \(H_{Q^{\ast}}(X|U^{\prime}) = H_{Q^{\ast}}(X|U) \leq \Delta\) and
\begin{align}
    F(R, \Delta) & = D(Q_{XYU}^{\ast} \| P_{XY}Q_{U|Y}^{\ast}) + | I_{Q^{\ast}}(Y;U) - R |^{+} \\
    & =  D(Q_{XYU^{\prime}}^{\ast} \| P_{XY}Q_{U^{\prime}|Y}^{\ast}) + | I_{Q^{\ast}}(Y;U^{\prime}) - R |^{+} \\
    & \geq \min_{ \substack{ Q_{XYU^{\prime}} : \\ H_{Q}(X|U^{\prime}) \leq \Delta } } D(Q_{XYU^{\prime}} \| P_{XY}Q_{U^{\prime}|Y}) + | I_{Q}(Y;U^{\prime}) - R |^{+} \\
    & = F^{\prime}(R, \Delta),
\end{align}
which completes the proof.

\bibliography{ref}
\bibliographystyle{IEEEtran}

\end{document}

%% file: figures/IB_source_model.tex
\begin{tikzpicture}
        \node at (0,0) [draw, rectangle,  name = tx, label= above:{\small Source}, minimum height=0.8cm,minimum width=2cm] {$X^n$};
        \node at (3,0) [draw, rectangle,  name = channel, label= above:{\small Channel}, minimum height=0.8cm,minimum width=1.8cm] {$P_{Y|X}$};
        \node at (6,0) [draw, rectangle,  name = relay, label= above:{\small \smash{Helper} }, minimum height=0.8cm,minimum width=1.8cm] {$f_n$};
        \node at (10,0) [draw, rectangle,  name = rx, label= above:{\small Receiver}, minimum height=0.8cm,minimum width=2cm] {$g_n$};
        \node at (12,0) [name = output] { $\hat{P}_m$ } ;

        \draw [->] (tx.east) -- (channel.west);
        \draw [->] (channel.east) --  (relay.west);
        \draw [->] (relay.east) -- node[align = center, above]{\small $R$ } (rx.west);
        \draw [->] (rx.east) -- (output.west);
\end{tikzpicture}




%% file: figures/source_coding_helper_both_transmitter_receiver.tex
\begin{tikzpicture}
        \node at (0,0) [draw, rectangle, name = source, label= above:{\small Source },  minimum height=0.8cm, minimum width=1.8cm]  {$X^n$};

        \node at (3,0) [draw, rectangle, name = channel, label= above:{\small Channel }, minimum height=0.8cm,minimum width=1.8cm] {$P_{Y | X} $};
        \node at (6,0) [draw, rectangle, name = helper, label= above:{\small \smash{Helper} }, minimum height=0.8cm,minimum width=1.8cm] {$\varphi_n$};
        \node at (10,0) [draw, rectangle, name = receiver, label= above:{\small Receiver}, minimum height=0.8cm,minimum width=1.8cm] {$\phi_n$};
        \node at (6,-1.5) [draw, rectangle, name = transmitter, label= below:{\small Transmitter}, minimum height=0.8cm,minimum width=1.8cm] {$\tilde{f}_n$};
        \node at (12,0) [name = output] { $\hat{X}^n$ } ;

        \draw [->] (source.east) -- (channel.west);
        \draw [->] (channel.east) --  (helper.west);
        \draw [->] (helper.east) -- node[align = center, above]{\small $B$ } (receiver.west);
        \draw [->] (0,-1.5) -- (transmitter.west);
        \draw  (0,-1.5) --(source.south);
        \draw (transmitter.east) -- (10,-1.5);
       \draw [->] (10,-1.5) --(receiver.south);
       \draw [->] (receiver.east) -- (output.west);
       \draw [->] (helper.south) -- (transmitter.north);

\end{tikzpicture}



